\documentclass{templates/sig-alternate-10pt-2013}
\usepackage{times}
\let\crnotice\mycrnotice%
\let\confname\myconfname%

\usepackage{epsfig, endnotes}
\usepackage{amssymb}
\usepackage{amsmath}
\usepackage{amsfonts}

\usepackage[usenames,dvipsnames]{color}
\definecolor{mycitegreen}{RGB}{0,90,0}
\definecolor{myrefblue}{RGB}{90,0,0}
\definecolor{red}{RGB}{140,0,0}
\definecolor{green}{RGB}{0,80,0}
\usepackage[draft=true, final=true, colorlinks=true, citecolor=mycitegreen, linkcolor=myrefblue, anchorcolor=myrefblue, pagebackref=false]{hyperref}
\usepackage{graphics,enumitem}
\usepackage{epsfig}
\usepackage{graphicx}
\usepackage{floatrow}
\usepackage[font={bf,small}]{caption}
\usepackage{multirow}
\usepackage{array}
\usepackage{color}
\usepackage{breakurl}
\usepackage{xspace}

\usepackage{amsthm}
\usepackage{boxedminipage}
\usepackage{graphicx}
\usepackage{caption}
\usepackage{subcaption}
\usepackage{tikz}
\usetikzlibrary{shapes.geometric, arrows}
\usepackage{appendix}

\newcommand{\fixmec}[2]{{\color{#2}{[#1]}}}
\newcommand{\fixme}[1]{\fixmec{#1}{red}}

\newcommand{\sriram}[1]{\fixmec{[SR: #1]}{magenta}}

\newcommand{\panic}[1]{\vspace{-#1 plus 1pt minus 1pt}}
\newcommand{\panictwo}[1]{\vspace{-#1 plus 3pt minus 3pt}}
\newcommand{\ncaption}[1]{
  \panic{1pt}
  \renewcommand{\baselinestretch}{0.9}
  \caption{#1}
  \panictwo{8pt}
  \renewcommand{\baselinestretch}{1}
}

\newcommand{\eat}[1]{}

\def\compactify{\itemsep=0pt \topsep=0pt \partopsep=0pt \parsep=0pt}
 \let\latexusecounter=\usecounter

\usepackage{ifpdf}
\ifpdf
\else
\fi

\makeatletter
\newif\if@restonecol
\makeatother

\usepackage[linesnumbered,figure,vlined]{algorithm2e}

\SetCommentSty{mycommfont}

\usepackage{etoolbox}

\newcommand*{\mathcolor}{}
\def\mathcolor#1#{\mathcoloraux{#1}}
\newcommand*{\mathcoloraux}[3]{%
  \protect\leavevmode
  \begingroup
    \color#1{#2}#3%
  \endgroup
}

\newcommand{\cut}[1]{}

\newcommand{\rnum}[1]{\lowercase\expandafter{\romannumeral #1\relax}}

\let\oldnl\nl
\newcommand{\nonl}{\renewcommand{\nl}{\let\nl\oldnl}}

\makeatletter
\patchcmd{\maketitle}{\@copyrightspace}{}{}{}
\makeatother

\newcommand{\sktilde}{\raise.17ex\hbox{$\scriptstyle\sim$}}

\newcommand{\codecomment}{\textcolor{blue}}
\newcommand{\camera}{}

\newcommand{\xref}[1]{\S\ref{#1}}

\newcommand{\chline}{\vspace{-3pt} \line(1,0){225} }

\newcommand{\name}{{\textsc{\small DagPS}}}
\newcommand{\namesec}{{\sc{DagPS}}}

\newenvironment{Itemize}%
{
\begin{itemize}%
\setlength{\itemsep}{0in}%
\setlength{\topsep}{-.1in}%
\setlength{\partopsep}{-.1in}%
\setlength{\parsep}{-.1in}%
\setlength{\parskip}{0in}}%
{\end{itemize}}

\newif\iflongversion
\longversionfalse

\title{
Do the Hard Stuff First:
Scheduling Dependent Computations in Data-Analytics Clusters
}

\author{
	{Robert Grandl$^{\ddag}$,  Srikanth Kandula,   Sriram Rao,   Aditya Akella$^{\ddag}$, Janardhan Kulkarni}\\
	Microsoft and University of Wisconsin-Madison$^{\ddag}$\\
}

\begin{document}

\maketitle
\begin{sloppypar}
\noindent{\bf Abstract-- }
We present {\name}, a scheduler that improves cluster utilization and job completion times
by packing tasks with multi-resource requirements and inter-dependencies.
While the underlying scheduling problem is intractable in general, {\name} is nearly optimal on the job DAGs
that appear in production clusters at a large enterprise.
Our key insight is that carefully handling the long-running tasks and those with 
tough-to-pack resource requirements will lead to good schedules for DAGs.
However, which subset of tasks to treat carefully is {\em a priori} unclear. 
{\name} offers a novel search procedure that evaluates various possibilities and outputs a valid schedule. 
An online component enforces the schedules desired by the various jobs running on the cluster.
In addition, it packs tasks and, for any desired fairness scheme, guarantees bounded unfairness.
We evaluate {\name} on a 200 server cluster using traces of over 20,000 DAGs
collected from a large production cluster.
Relative to the state-of-the art schedulers, {\name} speeds up half of the jobs by over 30\%.

\theoremstyle{definition}
\newtheorem{lemma}{Lemma}

\section{Introduction} 
\label{sec:intro} 
DAGs (directed acyclic graphs) are a powerfully general abstraction
for scheduling problems. Scheduling network transfers of a multi-way join or the work in a
geo-distributed analytics job and many others can be represented as
DAGs.  However, scheduling even one DAG is known to be an NP-hard
problem~\cite{Mastrolilli:2008,MastrolilliS:2009}.
\cut{
The common abstraction makes it possible to adapt
previously-studied DAG scheduling
algorithms~\cite{tpkelly1,StripPart,Goldberg:1997} for the problems at
hand.  DAG scheduling is, however, known to be an NP-hard
problem~\cite{Mastrolilli:2008,MastrolilliS:2009}.  
}

Consequently, existing work focuses on special cases of the DAG scheduling
problem using simplifying assumptions such as: ignore dependencies, only consider chains, assume
only two types of resources or only one machine or that the
vertices have similar resource
requirements~\cite{tpkelly1,StripPart,psv,Goldberg:1997,lmr,reverse-fit,ImKMP15}.
However, the assumptions that underlie these approaches
often do not hold in practical settings, motivating us to 
take a fresh look at this problem.

We illustrate the challenges in the context of job DAGs in
data-analytics clusters. Here, each DAG vertex represents a computational task and edges
encode input-output dependencies. 
Programming models such as SparkSQL, Dryad and Tez~\cite{tez,sparksql,dryad} lead to job DAGs that violate  
many of the above assumptions.
Traces from a large cluster reveal
that (a) DAGs have complex structures with the median job having a
depth of seven and a thousand tasks, (b) there is substantial
variation in compute, memory, network and disk usages across 
tasks~(stdev./avg in requirements is nearly $1$), (c) task runtimes range from sub-second to hundreds of seconds,
and (d) clusters suffer from resource fragmentation across machines.
The net effect of these challenges, based on our analysis, is that the  
completion times of jobs in this production cluster can be improved by 50\%  
for half the DAGs.  

The problem is important because data-analytics clusters run thousands of mission critical jobs each day in enterprises and in the cloud~\cite{bdas,bdas2}.
Even modest improvements in job throughput significantly improves the ROI (return-on-investment)
of these clusters; and quicker job completion reduces the lag between data collection and decisions (i.e., ``time to insight'')
which potentially increases revenue~\cite{velocityLatency}.

%

\vskip .05in
To identify a good schedule for one DAG, we observe that the pathologically bad schedules in today's approaches mostly arise due to these reasons:
(a) long-running tasks have no other work to overlap
with them and (b) the tasks that are runnable do not pack well with
each other.  Our core idea, in response, is rather simple: identify the potentially {\em troublesome} tasks, such as those that run for a very long time
or are hard to pack, and place them first on a {\em virtual resource-time space}.
This space would have $d+1$ dimensions when tasks require $d$ resources; the last 
dimension being time. Our claim is that placing the troublesome tasks first
leads to a good schedule since the remaining
tasks can be placed into resultant holes in this space. 

\vskip .05in
Unfortunately, scheduling one DAG well does not suffice.
Production cluster schedulers have many concurrent jobs, online arrivals
and short-lived tasks~\cite{yarn,mesos,sparrow,omega}. Together, these impose a strict time-budget on scheduling. 
Also, sharing criteria such as fairness have to be addressed during scheduling.
Hence, 
production clusters are forced to use simple, online heuristics.

We ask whether it is possible to efficiently schedule complex DAGs {\em while} retaining the
advantageous properties of today's production schedulers such as reacting in an online manner, considering multiple objectives etc.

\vskip .05in
To this end, we design a new cluster scheduler~{\name}.
At job submission time or soon thereafter, {\name} builds a preferred schedule
for a single job DAG by placing the troublesome tasks first.  {\name}
solves two key challenges in realizing this idea: (1) the best choice
of troublesome tasks is intractable to compute and (2) {\em dead-ends}
may arise because tasks are placed out-of-order~(e.g., troublesome go first)
and it is apriori unclear how much slack space should be set aside.
{\name} employs a performant search procedure to address the first challenge
and has a placement procedure that provably avoids dead-ends for the second challenge.
\autoref{fig:graphene_overview} shows an example. 
\begin{figure}[t!]
\centering
\includegraphics*[width=3.3in]{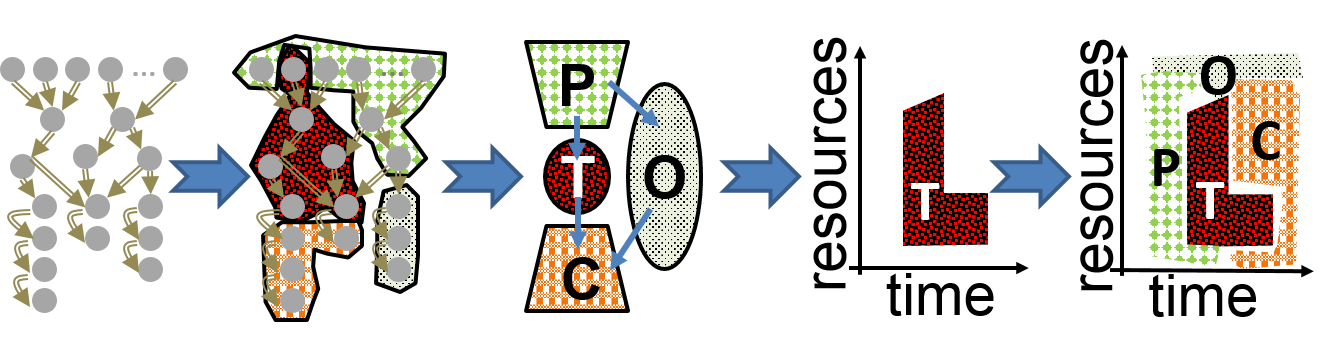}
\vskip -.05in
\ncaption{Shows steps taken by {\name} from a DAG on the left to its schedule on the
right. Troublesome tasks~${\tt T}$~(in red) are placed
first. The remaining tasks~(parents~${\tt P}$, children~${\tt C}$ and other~${\tt O}$) 
are placed {\em on top} of ${\tt T}$ in a careful order to ensure compactness and respect
dependencies.\label{fig:graphene_overview}}
\end{figure}

The schedules constructed for each DAG are passed on 
to a second online component of {\name} which coordinates
between the various DAGs running in the cluster and also reconciles between their multiple, potentially discordant, objectives. 
For example, a fairness scheme such as DRF may
require a certain job to get resources next,
but multi-resource packing--which we use to reduce resource fragmentation--or the preferred schedules above may indicate that some other
task should be picked next.
Our reconciliation heuristic, colloquially, attempts to {\em follow
the majority}; that is it can violate an objective, say fairness,
when multiple other objectives counterweight it.
However, to maintain predictable
performance, our reconciliation heuristic 
limits maximum unfairness to an operator-configured threshold.

We have implemented the two components of {\name} in Apache YARN and Tez
and have experimented with jobs from TPC-DS, TPC-H and other benchmarks on a 200 server
cluster. Further, we also evaluate {\name} in simulations on 20,000 DAGs
from a production cluster.

To summarize, we make theoretical as well as practical
contributions in this work.  Our key contributions are:
\begin{Itemize} 
\item A characterization of the DAGs seen in production at a large
  enterprise and an analysis of the performance of various DAG
  scheduling algorithms~(\xref{sec:packdeps}).
\item A novel DAG scheduler that combines multi-resource packing and
  dependency awareness~(\xref{sec:design}).
\item An online scheduler that mimics the preferred schedules for all
  the jobs on the cluster while bounding unfairness~(\xref{sec:runtime})
  for many models of fairness~\cite{yarnFS,kellyPF,drf}.
\item A new lower bound on the completion time of a DAG~(\xref{sec:betterlb}). Using this we show that the schedules built by {\name}'s offline component
  are within $1.04$ times ${\tt OPT}$ for half of the production DAGs; three quarters are within $1.13$ times
  and the worst is $1.75$ times ${\tt OPT}$.
\item An implementation that we intend to release as open source~(\xref{sec:system}).
\item
Our experiments show that 
{\name} improves the completion time of half of the DAGs by $19-31$\%; the
number varies across benchmarks.  The improvement for production DAGs
is at the high end of the range because these DAGs are more complex and have 
diverse resource demands.
\end{Itemize}

Lastly, while our work is presented in the context of cluster
scheduling, as noted above, similar DAG scheduling problems arise in
other domains.  We offer early results in~\autoref{subsec:other-domains}
from applying {\name} to scheduling DAGs arising in distributed build
systems~\cite{bazel,cloudmake} and in request-response
workflows~\cite{kwiken,c3}.

\section{Primer on Scheduling Job DAGs}
\label{sec:packdeps}


\subsection{Problem definition}
\label{subsec:perdagproblem}
Let each job be represented as a directed acyclic graph $\mathcal{G} = \{V, E\}$. Each node in $V$ is a task with demands for various resources.
Edges in $E$ encode precedence constraints between tasks. Many jobs can simultaneously run in a cluster. The cluster is a group of servers
organized as per some network topology.

{\name} considers task demands along four resource dimensions~(cores, memory, disk and network bandwidth). Depending on placement, tasks may need resources at more than one machine~(e.g., if input is remote) or along network paths. The network bottlenecks are near the edges~(at the source or destination servers and top-of-rack switches) in today's datacenter topologies~\cite{vl2,fattree,jupiterRising,fbdcn}. Some systems require users to specify the DAG $\mathcal{G}$ explicitly~\cite{hadoop,mapreduce,spark} whereas others use a query optimizer to generate $\mathcal{G}$~\cite{SCOPE}.  Production schedulers already allow users to specify task demands~(e.g., [1 core, 1~GB] is the default for tasks in Hadoop 2.6). Note that such annotation tends to be incomplete~(network and disk usage is not specifiable) and is in practice significantly overestimated since tasks that exceed their specified usage will be killed. Similar to other reports~\cite{rope,ccloud,optimus}, up to $40$\% of the jobs in the examined cluster are {\em recurring}. For such jobs, {\name} uses past job history to estimate task runtimes and resource needs. For the remaining ad-hoc jobs, {\name} uses profiles from similar jobs and adapts these profiles online~(see~\xref{sec:system}).


Given a set of concurrent jobs $\{\mathcal{G}\}$, the cluster scheduler maps tasks on to machines while meeting resource capacity limits and dependencies between tasks. 
\cut{
Tasks may be allocated a resource vector $\mathbf{a}$ that is less than their (peak) demands, i.e. $\mathbf{a} \leq \mathbf{r}$. Task duration $d$ is a convex function of resource allocation $\mathbf{a}$ defined for $\mathbf{a} \in [{\bf 0}, \mathbf{r}]$. That is, allocating less than peak resources causes the task to take longer to run. 
}
Improving {\bf performance}---measured in terms of the job throughput~(or
makespan) and the average job completion time---is crucial, while
also maintaining {\bf fairness}---measured in terms of how resources are divided amongst
groups of jobs per some requirement~(e.g., DRF or slot-fairness).  




\begin{figure}[t!]
\begin{subfigure}[b]{0.5\textwidth}
\centering
\includegraphics[width=2in]{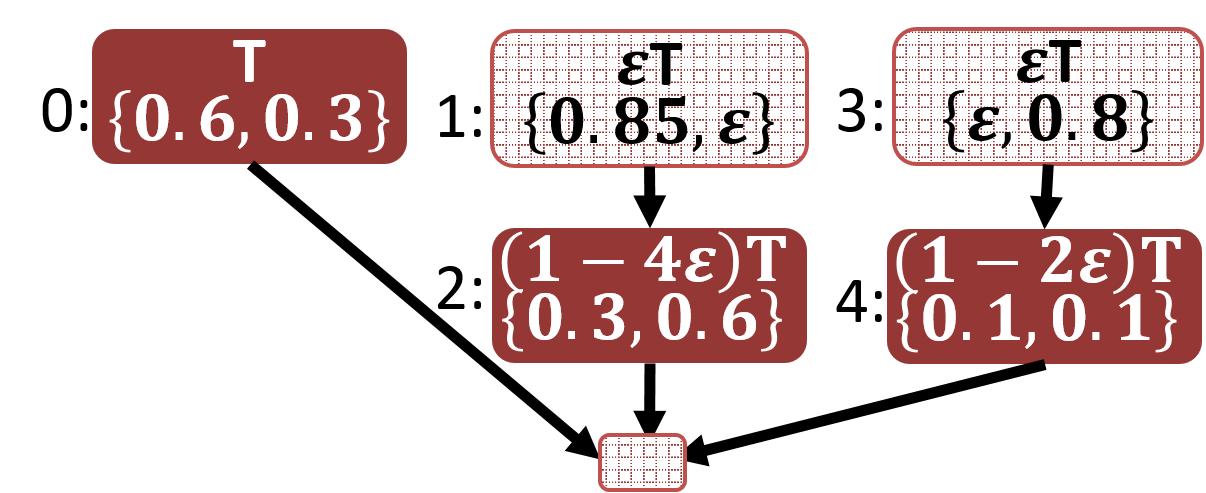}
\end{subfigure}
\begin{subfigure}[b]{\textwidth}
\begin{small}
\begin{tabular}{l|l|r|c}
{\bf Technique} & {\bf Execution Order} & {\bf Time} & {\bf Worst-case}\\\hline
${\tt OPT}$ & {$\{t_1, t_3\} \rightarrow \{t_0, t_2, t_4\} \rightarrow$} & {$T$} & {$-$}\\
${\tt CPSched}$ & {$t_0 \rightarrow t_3 \rightarrow t_4 \rightarrow t_1 \rightarrow t_2\rightarrow$} & {$3T$} & {$O(n)\times{\tt OPT}$}\\
{\bf Tetris} & {$t_0\rightarrow t_1\rightarrow t_2\rightarrow t_3\rightarrow t_4\rightarrow$} & {$3T$} & {$O(d)\times{\tt OPT}$}\\
\hline
\end{tabular}
\end{small}
\end{subfigure}
\vspace{0mm}\ncaption{An example DAG where Tetris~\cite{tetris} and Critical Path Scheduling 
take $3\times$ longer than the optimal algo~${\tt OPT}$. Here, {\name} equals ${\tt OPT}$. Details are in~\xref{subsec:problem}. Assume $\varepsilon\rightarrow0$.\label{fig:eg1}}
\end{figure}

\subsection{An illustrative example}
\label{subsec:problem}
We use the DAG shown in Figure~\ref{fig:eg1} to illustrate the scheduling issues.
Each node represents one task: the node labels represent
the task duration~(top) and the demands for two
resources~(bottom). Assume that the total resource available is $1$
for both resources and let $\varepsilon$ represent a small value. 

Intuitively, a good schedule would overlap the long-running tasks
shown with a dark background.  The resulting optimal schedule~(${\tt
  OPT}$) is shown in the table (see Figure~\ref{fig:eg1}). ${\tt OPT}$
overlaps the execution of all the long-running tasks-- $t_0, t_2$ and
$t_4$-- and finishes in $T$.  However, such long-running/resource
intensive tasks can be present anywhere in the DAG, and it is unlikely
that greedy local schedulers can overlap these tasks.  
To compare, the table also shows the schedules generated by a
typical DAG scheduler, and a state-of-the-art {\em packer} which
carefully packs tasks onto machines to maximize resource utilization. We discuss them next.

DAG schedulers such as critical path scheduling~(${\tt CPSched}$)
pick tasks along the critical path (CP) in the DAG.  The CP for a
task is the longest path from the task to the job output.  The figure
also shows the task execution order with ${\tt CPSched}$.\footnote{CP
  of $t_0, t_1, t_3$ is $T, T(1-3\varepsilon)$ and $T(1-\varepsilon)$
  respectively.  The demands of these tasks ensure that they cannot
  run simultaneously.}  ${\tt CPSched}$ ignores the resources needed
by tasks and does not pack. Consequently, for this example,
${\tt CPSched}$ performs poorly because it does not schedule tasks
that are not on the critical path first~(such as $t_1, t_3$) even
though doing so reduces resource fragmentation by overlapping the
long-running tasks.

On the other hand, packers such as, Tetris~\cite{tetris}, pack tasks to machines by matching along multiple resource dimensions. 
Tetris greedily picks the task with the highest value of the dot product between task's demand 
vector and the available resource vector. The figure also shows the task execution order with 
Tetris.\footnote{Tetris' packing score for each task, in descending order, is $t_0$=$t_2$=$0.9$, $t_1$=$0.85$, $t_3$=$0.8$ and $t_4$=$0.2$.} 
Tetris does not account for dependencies. Its packing heuristic only considers the tasks that are currently 
schedulable. In this example, Tetris performs poorly because it will not choose locally inferior 
packing options~(such as running $t_1$ instead of $t_0$) even when doing so can 
lead to a better global packing.

{\name} achieves the optimal schedule for this example. When searching
for troublesome subsets, it will consider the subset $\{t_0, t_2,
t_4\}$ because these tasks run for much longer. As shown in
Figure~\ref{fig:graphene_overview}, the troublesome tasks will be
placed first. Since there are no dependencies among them, they will
run at the same time. The parents~($\{t_1, t_3\}$) and any
children are then placed {\em on top}; i.e., compactly before and
after the troublesome tasks.

\subsection{Analyzing DAGs in Production}
\label{subsec:DAG-analysis}
\label{subsec:potential-gains}
We examined the production jobs from a cluster of tens of thousands of servers at a large enterprise. We also analyzed jobs from a $200$ server cluster that ran Hive~\cite{hive} jobs and jobs from a high performance computing cluster~\cite{condor}.

\begin{figure}[t!]
\centering
\includegraphics*[width=3.2in]{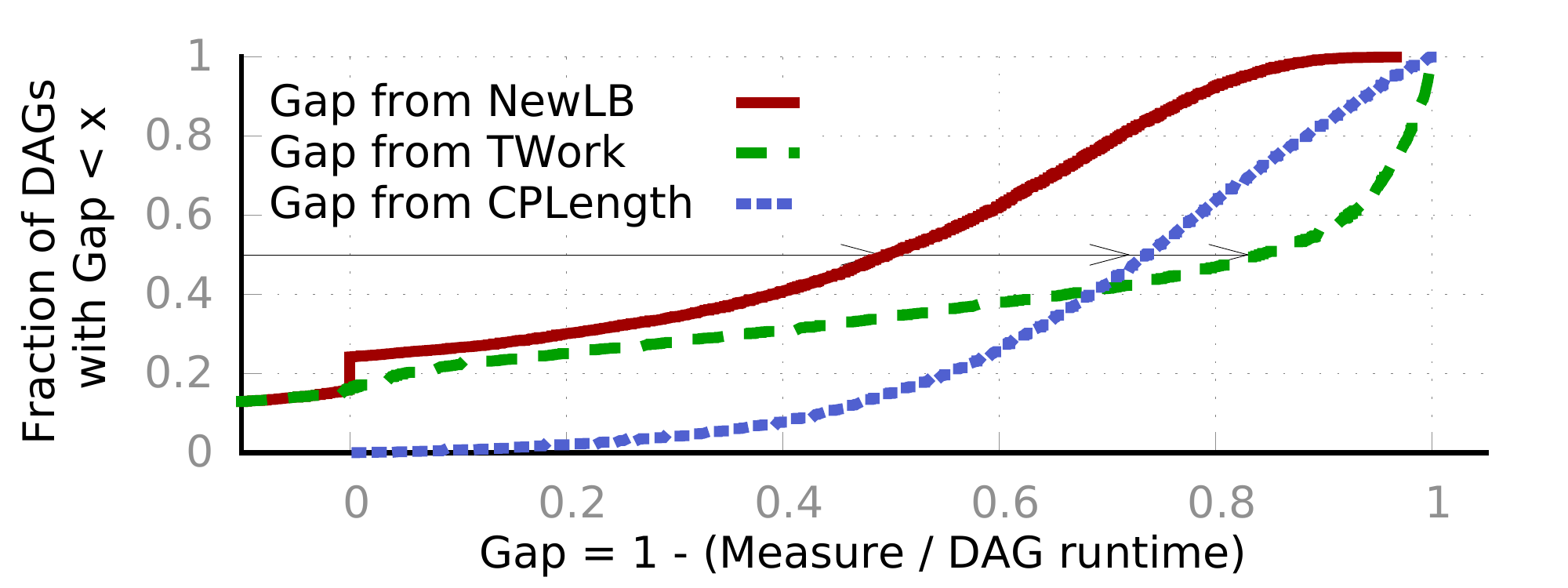}
\ncaption{CDF of \emph{gap} between DAG runtime and several measures.  Gap is computed as $1-\frac{\mbox{measure}}{\mbox{DAG runtime}}$.\label{fig:gap}}
\end{figure}
To quantify potential gains, we compare the runtime of the DAGs in production to three measures. The first measure, ${\tt CPLength}$ is the duration of the DAG's critical path. If the available parallelism is infinite, the DAG would finish within ${\tt CPLength}$. The second measure, ${\tt TWork}$, is the total work in the DAG normalized by the cluster share of that DAG. If there were no dependencies and perfect packing, a DAG would finish within ${\tt TWork}$. In practice, both of these measures are quite loose-- the first ignores all the
work off the critical path and the second ignores dependencies. Hence, our third measure is a new improved lower bound~${\tt NewLB}$ that uses the specific structure of data-parallel DAGs.  Further details are in~\xref{sec:betterlb} but intuitively ${\tt NewLB}$ leverages the fact that unlike random DAGs, all the tasks in a job stage~(e.g., a map or reduce or join) have similar dependencies, durations and resource needs. 

Figure~\ref{fig:gap} plots a CDF of the gap over all DAGs for these three measures. Observe that half of the jobs have a gap of over $70$\% for both ${\tt CPLength}$ and ${\tt TWork}$. The gap relative to ${\tt NewLB}$ is smaller, indicating that the newer bound is tighter, but the gap is still over $50$\% for half of the jobs. That is, they take over two times longer than they could.

A few issues are worth noting for this result. First, some DAGs finish faster than their ${\tt TWork}$ and ${\tt NewLB}$ measures. This is because our production scheduler is work conserving and can give jobs more than their fair share. Second, we know that jobs take longer in production because of runtime artifacts such as task failures or stragglers~\cite{mantri,late}. What fraction of the gap is explained due to these reasons? When computing the job completion times to use in this result, we attempted to explicitly avoid these issues as follows. First, we chose the fastest completion time from among groups of related recurring jobs. It is unlikely that every execution suffers from failures. Second, we shorten the completion time of a job by deducting all periods when the job has fewer than $10$ tasks running concurrently. This explicitly corrects for stragglers--one or a few tasks holding up job progress. Hence, we believe that the remaining gap is likely due to the scheduler's inability to pack tasks with dependencies.

\begin{table}[t!]
\centering
{\small
\begin{tabular}{p{.6in}||c|c|c|c|c|c|c}
& {CPU} & {Mem.} & \multicolumn{2}{c|}{Network} & \multicolumn{2}{c|}{Disk} \\\cline{4-7}
& & & Read & Write & Read & Write\\\hline
Enterprise: Private Stack& 0.76 & 1.01 & 1.69 & 7.08 & 1.39 & 1.94 \\\hline
Enterprise: Hive & 0.89 & 0.42 & 0.77 & 1.34 & 1.59 & 1.41  \\\hline
HPC: Condor & 0.53 & 0.80 & N/A & N/A & \multicolumn{2}{c|}{1.55 (R+W)} \\
\end{tabular}
}
\ncaption{Coefficient-of-variation~($={\mbox{stdev.}}/{\mbox{avg.}}$) of tasks' demands for various resource. Across three examined frameworks, tasks exhibit substantial variability~(CoV $\sim 1$) for many resources.\label{tab:resource_cov}}
\end{table}
\begin{table}[t!]
\centering
{\small
\begin{tabular}{p{.7in}|c|c|c|c|c}
\multirow{2}{*}{``Work'' that is \ldots}& \multicolumn{5}{c}{Percentage of total work in the DAG}\\
& [0-20) & [20-40) & [40-60) & [60-80) & [80-100]\\\hline
{on CriticalPath} & 14.6\% & 13.2\% & 15.2\% & 15.3\% & {\bf 41.6}\%\\
{``unconstrained''} & 15.6\% & 20.4\% & 14.2\% & 16.4\% & {\bf 33.3}\%\\
{``unordered''} & 0 & 3.6\% & 11.8\% & 27.9\% & {\bf 56.6}\% \\\hline
\end{tabular}
}\ncaption{Bucketed histogram of where the work lies in DAGs. Each entry denotes the fraction of all DAGs that have the metric labeled on the row in the range denoted by the column. For example, $14.6\%$ of DAGs have $[0, 20)\%$ of their total work on the critical path.\label{tab:whereIsMass}}
\end{table}

\vskip .05in
To understand the causes for the performance gap further, we characterize the DAGs along the following dimensions:

\textbf{What do the DAGs look like?} By depth, we refer to the number of tasks on the critical path. A map-reduce job has depth $2$. We find that the median DAG has depth 7. Further, we find that the median~(75th percentile) task in-degree and out-degree are 7~(48) and 1~(4) respectively. If DAGs are  chains of tasks, in- and out-degree's will be $1$. 
\iflongversion
\begin{figure}[t!]
\includegraphics[width=3in]{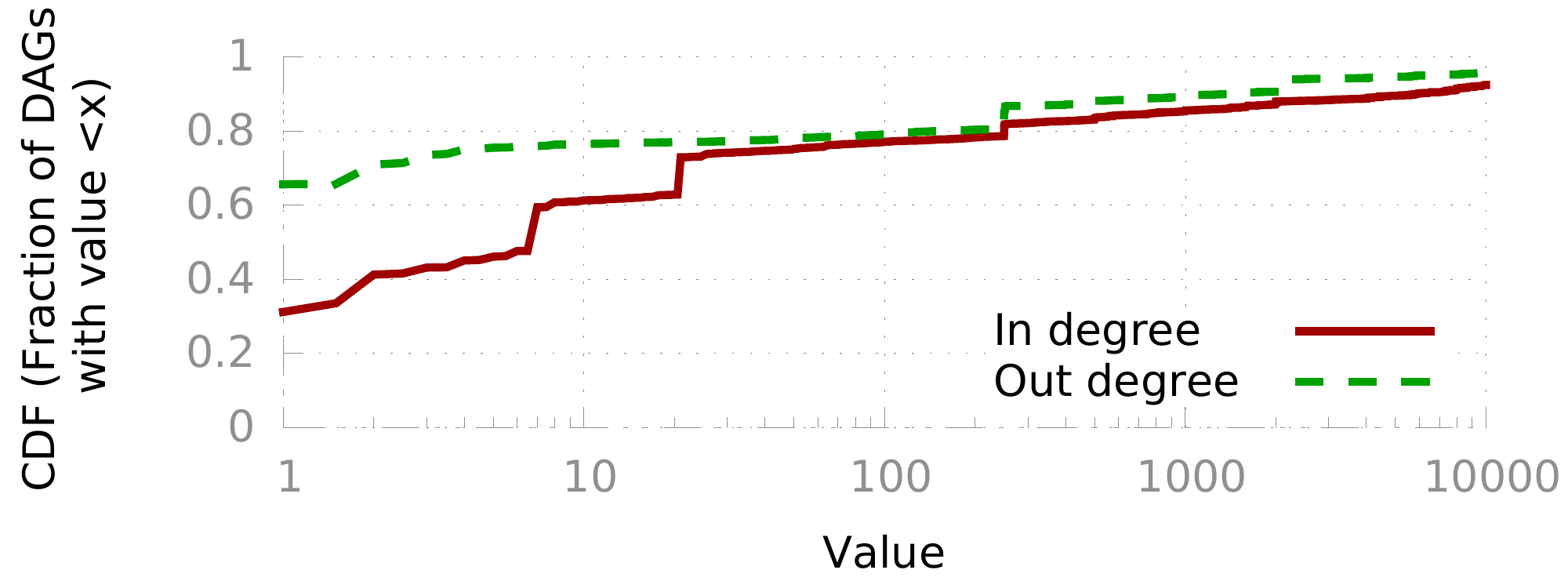}
\ncaption{Characterizing structural properties of the DAGs\label{fig:dag_more_detail}}
\end{figure}
Figure~\ref{fig:dag_more_detail} plots various structural properties of the DAGs including degree distributions.
\else
A more detailed characterization of DAGs including tree widths and path widths has been omitted for brevity.
\fi
Our summary is that the vast majority of DAGs have complex structures.

{\bf How diverse are the resource demands of tasks?}  Table~\ref{tab:resource_cov} shows the coefficient-of-variation~(CoV) across tasks for various resources.
We find that the resource demands vary substantially. 
The variability is possibly due to differences in work at each task: some are compute heavy~(e.g., user-defined code that processes videos) whereas other tasks are memory heavy~(e.g., in-memory sorts).

{\bf Where does the work lie in a DAG?}
We now focus on the more important parts of each DAG-- the tasks that do more work~(measured as the product of task duration and resource needs). Let ${\tt CPWork}$ be the total work in the tasks that lie on the critical path. From Table~\ref{tab:whereIsMass}, $42$\% of DAGs have ${\tt CPWork}$ above $80$\%. DAG-aware schedulers may do well for such DAGs. Let ${\tt UnconstrainedWork}$ be the total work in tasks with no parents (i.e., no dependencies). We see that roughly $33$\% of the DAGs have ${\tt UnconstrainedWork}$ above $80$\%. Such DAGs will benefit from packers. The above cases are not mutually exclusive and together account for $54$\% of DAGs. For the other $46$\% of DAGs, neither packers nor criticality-based schedulers may work well. 

Let ${\tt MaxUnorderedWork}$ be the largest work in a set of tasks that are neither parents nor children of each other. Table~\ref{tab:whereIsMass} shows that $57$\% of DAGs have ${\tt MaxUnorderedWork}$ above $80$\%. That is, if ancestors of the unordered tasks were scheduled appropriately, substantial gains can accrue from packing the tasks in the maximal unordered subset.

From the above analysis, we observe that (1) production jobs have
large DAGs that are neither a bunch of unrelated stages nor a chain of stages, 
and (2) a packing+dependency-aware scheduler can offer substantial improvements.

\subsection{Analytical Results}
\label{subsec:hits}
We take a step back to offer some more general comments.
First, DAG schedulers have to be aware of dependencies. 
That is, considering just the runnable tasks does not suffice.

\begin{lemma}
Any scheduling algorithm, deterministic or randomized, that does not account for the DAG structure
is at least $\Omega(d)$ times ${\tt OPT}$ where $d$ is the number of resources.
\label{lem:dag_awareness}
\end{lemma}
\vspace{-.1in}
\iflongversion
\begin{proof}
\begin{figure}[t!]
\centering
\includegraphics*[width=3in]{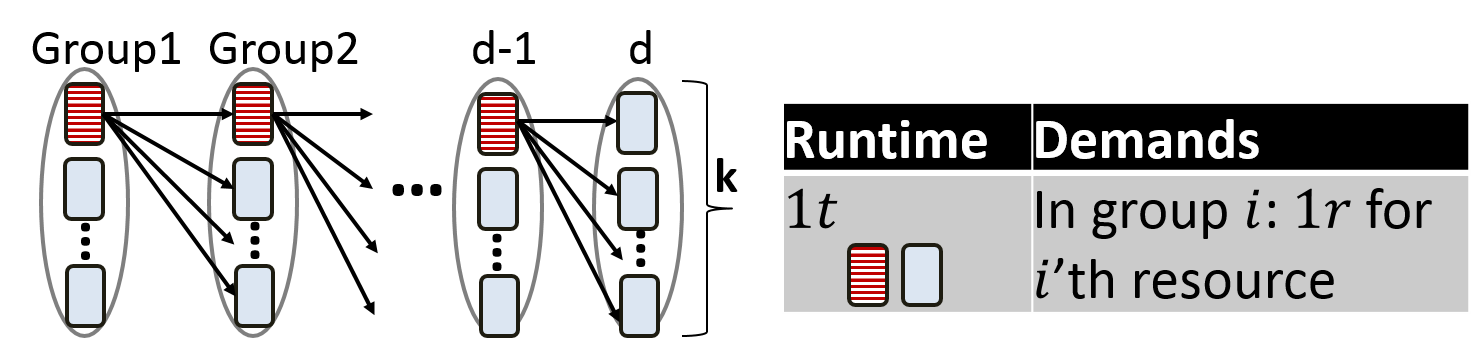}
\ncaption{A counter-example DAG that shows any scheduler not considering DAG structure
will be $\Omega(d)$ times ${\tt OPT}$.\label{fig:dag_ce}}
\end{figure}
\iflongversion
We first handle deterministic schedulers and then extend to randomized schedulers. Consider the DAG structure in Figure~\ref{fig:dag_ce} which is a chain of $d$ groups. All the tasks run for $1t$. Each group~(oval) has $k$ tasks and all of the tasks in the $i$'th group only use the $i$'th resource and require $1r$. Assume the capacity for all $d$ resources is $1r$. In each group there is a certain task colored red that is the parent of all tasks in the next level. This information is unavailable (and unused) by schedulers that do not consider the DAG structure. Hence, regardless of which order the scheduler picks tasks, an adversary can choose the last task in that group to be the red one. This leads to a schedule that takes $k d t$ time. Observe however that {$\tt OPT$} only requires $(k+d-1)t$ since it can schedule the red tasks with children one after the other~(in $(d-1)t$) and afterward one task from each group can run simultaneously~(hence $kt$ more steps). Thus, all deterministic schedulers are $\Omega(d)\times{\tt OPT}$.

To extend to randomized algorithms, we use Yao's max-min principle~\cite{Motwani:1995}. Recall that the principle states that the lower bound on any randomized algorithm is the same as lower bound on deterministic algorithms with randomized input. If we randomize the choice of the red task in each group,  the expected time at which the red task will finish is $kt/2$ and hence the expected schedule time is $k(d+1)t/2$ which is still $\Omega(d)\times{\tt OPT}$.
\else
Figure~\ref{fig:dag_ce} shows an adversarial DAG for which {\em any} scheduler that ignores dependencies will take $\Omega(d)$ times ${\tt OPT}$. Assume cluster capacity is $1r$ for each of the $d$ resources. The DAG has $d$ groups, each having a task filled with red dashes that is the parent of all the tasks in the next group. This information is unavailable (and unused) by schedulers that do not consider the DAG structure. Hence, regardless of which order the scheduler picks tasks, an adversary can choose the last task in a group to be the red one. Hence, such schedulers will take $k d t$ time. {$\tt OPT$} only requires $(k+d-1)t$ since it can schedule the red tasks first~(in $(d-1)t$) and afterwards one task from each group can run simultaneously~($kt$ more steps). We use Yao's max-min principle~\cite{Motwani:1995}~(the lower bound on any randomized algorithm is the same as lower bound on deterministic algorithms with randomized input) to extend the counter-example. If we randomize the choice of the red task, the expected time at which the red task will finish is $kt/2$ and hence the expected schedule time is $k(d+1)t/2$ which is still $\Omega(d)$ times ${\tt OPT}$.
\fi
\end{proof}
\else
For deterministic algorithms, the proof follows from designing an adversarial DAG for any scheduler. We extend to randomized algorithms by using Yao's max-min principle~(see~\autoref{subsec:dag_awareness}).
\fi
Lemma~\ref{lem:dag_awareness} applies to all multi-resource packers~\cite{tetris,vector,apx,onlineVectorPacking} since they ignore dependencies.

Second, and less formally, we note that schedulers have to be aware of resource heterogeneity. 
Many known scheduling algorithms have poor worst-case performance. In particular:
\begin{lemma}
Critical path scheduling can be $\Omega(n)$ times ${\tt OPT}$ where $n$ is the number of tasks in a DAG and Tetris can be $(2d-2)$ times ${\tt OPT}$.
\label{lem:worst_case}
\end{lemma}
\iflongversion
\begin{proof}
\begin{figure}[t!]
\centering
\includegraphics[width=3in]{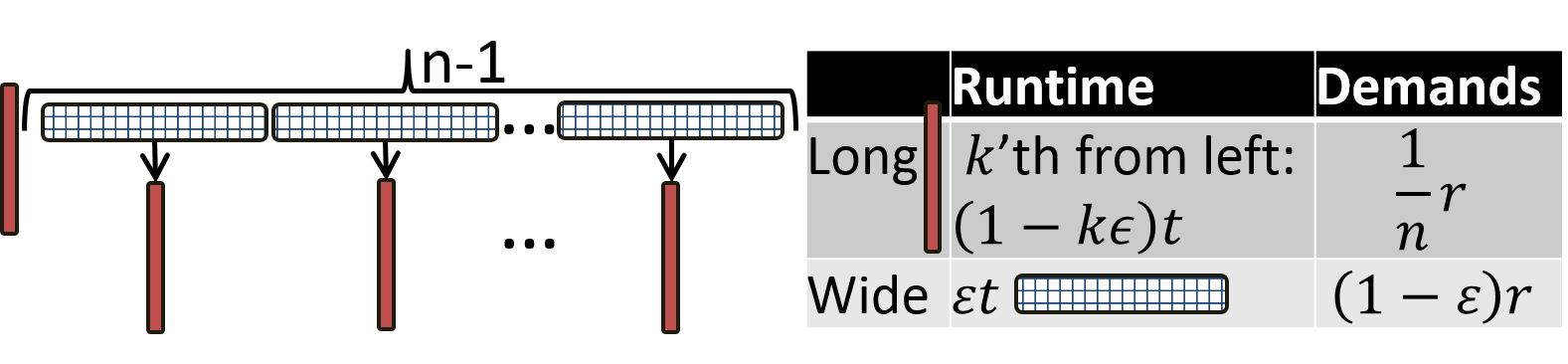}
\ncaption{An example DAG where critical path scheduling is $O(n)$ times ${\tt OPT}$ where $n$ is the number of nodes in the DAG.\label{fig:cpntimesworse}}
\end{figure}
\begin{figure}[t!]
\centering
\includegraphics[width=3in]{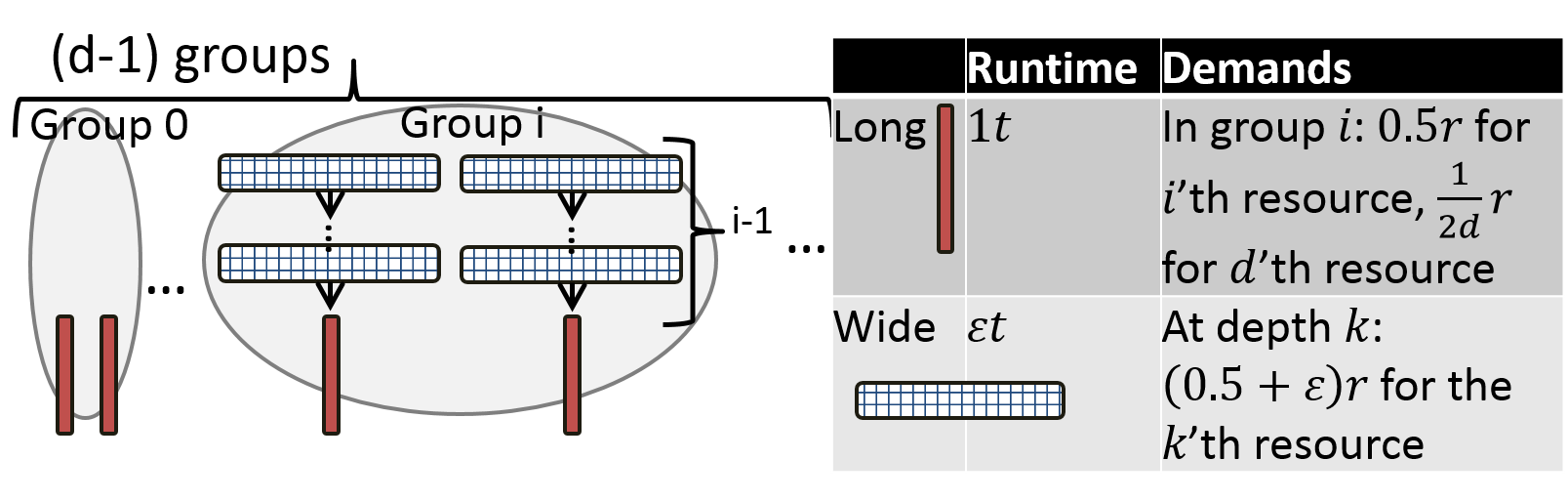}
\ncaption{An example DAG where Tetris~\cite{tetris} is $2d-2$ times ${\tt OPT}$ when tasks use $d$ kinds of resources.\label{fig:tetrisdtimesworse}}
\end{figure}
\iflongversion
\begin{proof}
\else
\noindent{\it Proof of Lemma~\ref{lem:worst_case}: }
\fi
\iflongversion
Figure~\ref{fig:cpntimesworse} shows an example DAG where CPSched takes $n$ times worse than ${\tt OPT}$ for a DAG with $2n$ tasks. As before, assume capacity is $1r$. The long tasks have duration $\sim 1t$ and demand $\frac{1}{n}r$ whereas the wide tasks have duration $\varepsilon t$ and require $(1-\varepsilon)r$. The critical path lengths of the various tasks are such that CPSched alternates between one long task and one wide task left to right. However, it is possible to overlap all of the long running tasks. This DAG completes at $\sim nt$ and $\sim 1t$ with CPSched and {\tt OPT} respectively.

Fig.~\ref{fig:tetrisdtimesworse} shows an example where Tetris~\cite{tetris} is $2d$ times ${\tt OPT}$. The task durations are $1t$ for long tasks and $\varepsilon t$ for wide tasks. There are two long tasks per group and those in the $i$'th group require $\frac{1}{2}r$ of the $i$'th resource and $\varepsilon r$ of other resources. Further, group $i$ has two chains of wide tasks. All wide tasks at depth $i$ require $(\frac{1}{2}+\varepsilon)r$ of the $i$'th resource. As in the above example, all long tasks can run together, hence ${\tt OPT}$ finishes in $1t$. However, Tetris greedily schedules the task with the highest dot-product between  task demands and available resources.  The DAG is constructed such that whenever a long task is runnable, it will have a higher score than any wide task. Further, for every long task that is not yet scheduled, there exists at least one wide parent that cannot overlap any long task that may be scheduled earlier. Hence, Tetris takes $2dt$ which is $2d$ times ${\tt OPT}$.

Combining these two principles, we conjecture that it is possible to find similar examples for any scheduler that ignores dependencies or is not resource-aware.
\else
Figure~\ref{fig:cpntimesworse} shows an example DAG where ${\tt CPSched}$ takes $n$ times worse than ${\tt OPT}$ for a DAG with $2n$ tasks. The critical path lengths of the various tasks are such that ${\tt CPSched}$ alternates between one long task and one wide task left to right. However, it is possible to overlap all of the long running tasks. This DAG completes at $\sktilde nt$ and $\sktilde 1t$ with ${\tt CPSched}$ and ${\tt OPT}$ respectively.
Fig.~\ref{fig:tetrisdtimesworse} shows an example where Tetris~\cite{tetris} is $2d-2$ times ${\tt OPT}$. As in the above example, all long tasks can run together, hence ${\tt OPT}$ finishes in $1t$. Tetris greedily schedules the task with the highest dot-product between task demands and available resources.  The DAG is constructed such that whenever a long task is runnable, it will have a higher score than any wide task. Further, for every long task that is not yet scheduled, there exists at least one wide parent that cannot overlap any long task that may be scheduled earlier. Hence, Tetris takes $(2d-2)t$ which is $(2d-2)$ times ${\tt OPT}$. Combining these two principles, we conjecture that it is possible to find similar examples for any scheduler that ignores dependencies or is not resource-aware.
\fi
\iflongversion
\end{proof}
\else
\qed
\fi

\end{proof}
\else
The proof is by designing adversarial DAGs for each scheduler~(see \autoref{sec:worst-case}).
\fi

To place these results in context, note that $d$ is about $4$~(cores, memory, network, disk) and can be larger when tasks require resources at other servers or on many network links. Further, the median DAG has hundreds of tasks~($n$). {\name} is close to ${\tt OPT}$ on all of the described examples. Furthermore, {\name} is within $1.04$ times optimal for half of the production DAGs~(estimated using our new lower bound).

Finally, we note the following:
\begin{lemma}
If there were no precedence constraints and tasks were malleable, ${\tt OPT}$ is achievable by a greedy algorithm.
\label{lem:simplicity}
\end{lemma}
We say a task is malleable if assigning any (non-negative) portion of its demand $p$ will cause it to make progress at rate $p$. In particular, tasks can be paused ($p=0$) at any time which is also referred to as tasks being preemptible.  
\iflongversion
The proof follows from our description of the greedy algorithm in~\autoref{sec:simple_greedy}.
\else
The proof follows by describing the simple greedy algorithm which we omit here for brevity.
\fi

Our summary is that practical DAGs are hard to schedule because of their complex structure as well as discretization issues when tasks need multiple resources~(fragmentation, task placement etc.)

\eat{
Today's big data schedulers~\cite{hadoop,tez,hive} will execute the task that finds a suitable 
placement first and when there is contention use fairness to arbitrate across jobs. Both these criteria are exogenous 
from the perspective of a DAG and hence unlikely to lead to good schedules. To ground these arguments, for illustration purposes,
we perform Monte Carlo simulations that (a) pick the next task to run at random 
from among those that are schedulable and (b) are work-conserving. The expected
runtime is $2.2$T and the likelihood of generating the optimal schedule is $0.22$.}

\section{Novel ideas in {\namesec}}
\label{sec:ideas}

Cluster scheduling is the problem of matching tasks to machines.
Every practical scheduler today does so in an online manner but
has very tight timing constraints since clusters have 
thousands of servers, many jobs each having many pending tasks and 
tasks that finish in seconds or less~\cite{spark,yarn}. 
Given such stringent time budget, carefully considering large DAGs
seems hopeless. 

As noted in \xref{sec:intro}, a key design decision in {\name} is to divide this problem
into two parts. An offline component constructs careful schedules for a 
single DAG. We call these the {\em preferred schedules}. A second online
component enforces the preferred  schedules of the various
jobs running in the cluster. We elaborate on each of these parts below.  
Figure~\ref{fig:g_overview} shows an example of how the two parts 
may inter-operate in a YARN-style architecture.
Dividing a complex problem into parts and independently solving each part
often leads to a sub-optimal solution. Unfortunately, we have no guarantees 
for our particular division. However, it can scale to large clusters 
and outperforms the state-of-art in experiments.

\vskip 0.15in
To find a compact schedule for a single DAG, our idea is to place the troublesome tasks, i.e., those that can lead to a poor schedule, first onto a virtual space. Intuitively, this maximizes the likelihood that any {\em holes}, un-used parts of the resource-time space, can be filled by other tasks. However, finding the best choice of troublesome tasks is as hard as finding a good schedule for the DAG. We use an efficient search strategy that mimics dynamic programming: it picks subsets that are more likely to be useful and avoids redundant exploration. Further, placing troublesome tasks first can lead to {\em dead-ends}. We define dead-end to be an arrangement of a subset of the DAG in the virtual space on which the remaining tasks cannot be placed without violating dependencies. Our strategy is to divide the DAG into subsets of tasks and place one subset at a time. 
While intra-subset dependencies are trivially handled by schedule construction, inter-subset dependencies are handled by restricting the order in which the various subsets are placed. We prove that the resultant placement has no {\em dead-ends}.

\vskip 0.15in
The online component has to co-ordinate between some potentially discordant directives. Each job running in the cluster offers a preferred schedule for its tasks~(constructed as above). Fairness models such as DRF may dictate which job~(or queue) should be served next. The set of tasks that is advantageous for packing~(e.g., maximal use of multiple resources) can be different from both the above choices. We offer a simple method to reconcile these various directives. Our idea is to compute a real-valued score for each pending task that incorporates the above aspects {\em softly}. That is, the score trades-off violations on some directives if the other directives weigh strongly against it. For example, we can pick a task that is less useful from a packing perspective if it appears much earlier on the preferred schedule. Two key novel aspects are  judiciously overbooking resources and bounding the extent of unfairness. Overbooking allows schedules that overload a machine or a network link if the cost of doing so~(slowing-down of all tasks using that resource) is less than the benefit~(can finish more tasks).  

The offline component of {\name} is described next; the online component is described in Section~\ref{sec:runtime}.

\begin{figure}[t!]
\centering
\includegraphics[width=3.2in]{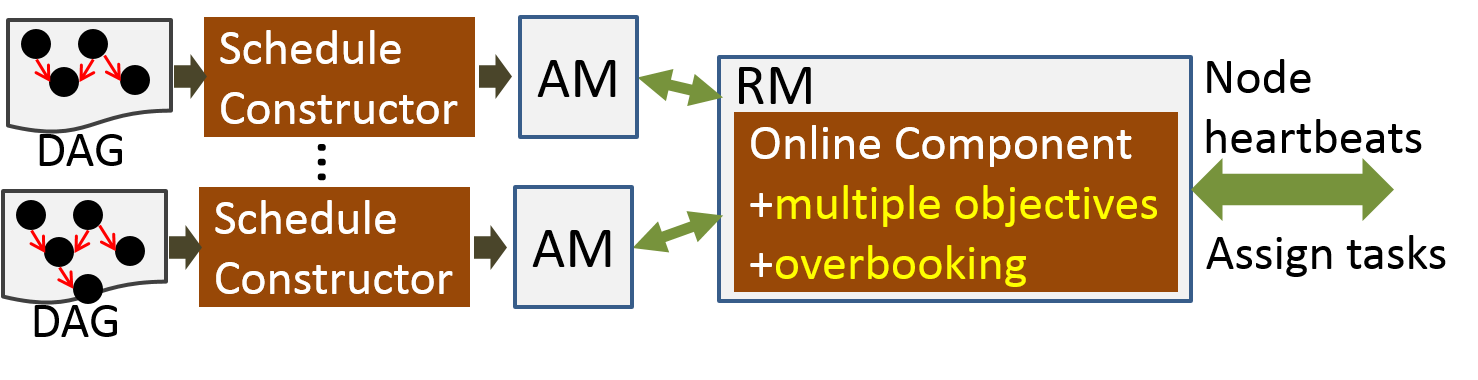}
\vspace{-1.5mm}\ncaption{{\name} builds schedules per DAG at job submission. The runtime 
component handles online aspects. AM and RM refer to the YARN's application 
and resource manager components.\label{fig:g_overview}}
\end{figure}

\eat{
A key new idea in {\name} is to build the schedule by looking at the
whole DAG in a manner akin to dynamic programming. 
\sriram{From hereon, it is a repeat of what is in the Intro.  The next
3 paras should be in 3 sentences?}
Imagine the
schedule to be some layout of the tasks on a $d+1$ dimension space,
where the first $d$ dimensions correspond to resources and the last
corresponds to time. {\name} binds tasks to resources by starting with
the {\em troublesome} tasks. That is, tasks with a long duration or
those with large or peculiarly-shaped resource requirements are laid
out first in this space. Since the space is sparsely populated
initially, doing so allows the troublesome tasks to be
placed compactly. The remaining tasks are subsequently packed to
fill any remaining {\em holes} in the $d+1$ dimension space leading to an
overall compact layout.

We next refine the above idea to account for dependencies between
tasks. {\em Within} the set of troublesome tasks, {\name} uses two
different orders. Let a task be {\em ready} if all its {\em
  ancestor} tasks have already been mapped on to this space. Iteratively picking a
ready task and laying it out {\em after} its last ancestor finishes
will satisfy dependencies. We call this the {\em forward} order.
Alternately, tasks can also be laid out in an analogous {\em backward}
order. That is, a task is ready after all its descendants have been
laid out and iteratively place a ready task before its earliest
descendant begins. These two orders lead to very different layouts.
A few other decisions also affect the schedule: which ready task to place next when 
many tasks are ready and where to place that task among the many possible 
locations in space. {\name} offers a strategy to quickly search among these
various alternatives while pruning the search space as needed~(\xref{sec:greedypacking}).

A final complication is to avoid deadlocks. 
Consider tasks $t_1 \rightarrow t_2 \rightarrow t_3$. If $t_1$ and
$t_3$ are troublesome, they will be placed in space first, and $t_2$
is deferred for later. But the placement of $t_1$ and $t_3$ may not
leave enough space~(time and resources) for $t_2$ to be placed in
between. In the general case, $t_2$ may be some subset of the DAG
and it is apriori unknown how much space should be set aside when
placing the troublesome tasks. Setting aside more space will avoid
deadlocks but leads to less compact schedules. {\name} places tasks in
a carefully chosen order so as to build compact schedules without
deadlocks~(\xref{sec:greedypacking}).
}

\section{Scheduling one DAG}
\label{sec:design} 

\begin{algorithm}[t!]
{\small
\nonl	\chline\\
\nonl	\textbf{Definitions:} In DAG $\mathcal{G}$, $t$ denotes a task and $s$ denotes a stage, i.e., a collection of similar tasks.\\ 
\nonl	Let $\mathcal{V}$ denote all the stages (and hence the tasks) in $\mathcal{G}$.\\
\nonl	Let $\mathcal{C}(s, \mathcal{G})$, $\mathcal{P}(s, \mathcal{G})$, $\mathcal{D}(s, \mathcal{G})$, $\mathcal{A}(s, \mathcal{G})$, $\mathcal{U}(s, \mathcal{G})$ denote the children, parents, descendants, ancestors and unordered neighbors of $s$ in $\mathcal{G}$.\\
\nonl	For clarity, $\mathcal{U}(s, \mathcal{G}) =\mathcal{V} - \mathcal{A}(s, \mathcal{G}) - \mathcal{D}(s, \mathcal{G}) - \{s\}$\\
\nonl	\chline\\
		\textbf{Func:} \textsf{BuildSchedule}:\\
		\textbf{Input:} $\mathcal{G}$: a DAG, $m$: number of machines\\
		\textbf{Output:} An ordered list of tasks $t \in \mathcal{G}$ \\

		\iflongversion
			${\tt Ans} \leftarrow \{\}$\\
			\renewcommand{\gg}{\mathcal{G}'}
			\ForEach{\mbox{dag} $\mathcal{G}' \in {\tt CutDAGs}(\mathcal{G})$}
			{
		\else
			\renewcommand{\gg}{\mathcal{G}}
		\fi
		$\mathcal{S}_{\mbox{best}} \leftarrow \varnothing$\codecomment{// best schedule for $\gg$ thus far}\\
		\ForEach{\mbox{sets } $\{{\tt T, O, P, C}\} \in {\tt CandidateTroublesomeTasks}(\gg)$}
		{
		Space $\mathcal{S} \leftarrow$ CreateSpace$(m)$ \codecomment{//resource-time space}\\
		$\mathcal{S} \leftarrow {\tt PlaceTasks}({\tt T}, \mathcal{S}, \gg)$\codecomment{// trouble goes first}\\
		$\mathcal{S} \leftarrow {\tt TrySubsetOrders}(\{{\tt OCP}, {\tt OPC}, {\tt COP}, {\tt POC}\}, \mathcal{S}, \gg)$
		\\
		\lIf{$\mathcal{S} < \mathcal{S}_{\mbox{best}}$} {$\mathcal{S}_{\mbox{best}} \leftarrow \mathcal{S}$ \codecomment{//keep the best schedule}}
		}
		\iflongversion
			${\tt Ans} \leftarrow {\tt Ans} \cup {\tt OrderTasks}(\gg, \mathcal{S}_{\mbox{best}})$\codecomment{// concatenate schedules}\\
			}
		\else
			\Return ${\tt OrderTasks}(\gg, \mathcal{S}_{\mbox{best}})$
		\fi
		
\nonl	\chline\\		
}
\ncaption{Pseudocode for constructing the schedule for a DAG. Helper methods are in Figure~\ref{fig:ps_sched_helpers}. 
\label{fig:ps_sched}
}
\end{algorithm}
{\name} builds the schedule for a DAG in three steps.
Figure~\ref{fig:graphene_overview} illustrates these steps and Figure~\ref{fig:ps_sched} has a simplified pseudocode.
First, {\name} identifies some troublesome tasks and divides the DAG into four subsets~(\xref{sec:searchspace}). 
Second, tasks in a subset are packed greedily onto the virtual space while respecting dependencies~(\xref{sec:eachset}).
Third, {\name} carefully restricts the order in which different subsets are placed such that  the
troublesome tasks go first and there are no dead-ends~(\xref{sec:greedypacking}). 
{\name} picks the most compact schedule after iterating over many choices for troublesome tasks.
We discuss some enhancements in~\xref{sec:enhancements}.
The resulting schedule is passed on to the online component~(\xref{sec:runtime}). 




\subsection{Searching for troublesome tasks}
\label{sec:searchspace} 
To identify {\em troublesome} tasks, {\name} computes two scores per task. The first, ${\tt LongScore}$, divides the task duration by the maximum value across all tasks.  Tasks with a higher score are more likely to be on the critical path and can benefit from being placed first because other work can overlap with them. The second, ${\tt FragScore}$, reflects the packability of tasks in a stage~(e.g., a map or a reduce). It is computed by dividing the total work in a stage~($\tt TWork$ defined in~\xref{subsec:problem}) by how long a greedy packer will take to schedule that stage. Tasks that are more difficult to pack would have a lower ${\tt FragScore}$. Given thresholds $l$ and $f$, {\name} picks tasks with ${\tt LongScore} \geq l$ or ${\tt FragScore} \leq f$. Intuitively, doing so biases towards selecting tasks that are more likely to hurt the schedule because they are too long or too difficult to pack.
 {\name} iterates over different values for the $l$ and $f$ thresholds to find a compact schedule. 

To speed up this search, (1) rather than choose the threshold values arbitrarily, {\name} picks values that are discriminative, i.e., those that allow different subsets of tasks to be considered as troublesome and (2) {\name} remembers the set of troublesome tasks that were already explored~(by previous settings of the thresholds) so that it will construct a schedule only once per unique troublesome set. 

As shown in Figure~\ref{fig:ps_sched_helpers},
the set~${\small\sf T}$ is a closure over the chosen troublesome tasks. 
That is, ${\small\sf T}$ contains the troublesome tasks and all tasks that lie on a path in the DAG between two troublesome tasks. 
The parent and child subsets ${\small\sf P}$, ${\small\sf C}$ consist of tasks that are not in ${\small\sf T}$ 
but have a descendant or ancestor in ${\small\sf T}$ respectively. 
The subset ${\small\sf O}$ consists of the remaining tasks.

   \begin{algorithm}[t!]
{\small
\nonl		See Definitions atop Fig.~\ref{fig:ps_sched}.\\
\nonl 		\chline\\
\iflongversion
		\textbf{Func:} \textsf{CutDAGs}:\\
		\textbf{Input:} $\mathcal{G}$: input DAG
		\textbf{Output:} $\mathcal{L}$: ordered list of DAGs\\
		$\mathcal{L} \leftarrow \{ \mathcal{G}\}$\\
		toProcess.push($\mathcal{G}$)\\
		\While{\mbox{! toProcess}.\mbox{empty}()}
		{
		$\mathcal{G'} \leftarrow$ toProcess.pop()\\

		\ForEach{$\mbox{stage}\ \  s \in \mathcal{G}'$}
		{
		\If{$\mathcal{U}(s, \mathcal{G}') == \varnothing$\codecomment{// no unordered neighbors}}
		{
		$\{\mathcal{G}_1, \mathcal{G}_2\} \leftarrow$ \{$\mathcal{A}(s, \mathcal{G}') \cup s$, $\mathcal{D}(s, \mathcal{G}')$\}\codecomment{// cut at $s$}\\
		Replace $\mathcal{G}'$ with $\{\mathcal{G}_1, \mathcal{G}_2\}$ in $\mathcal{L}$ \\
		toProcess.push($\mathcal{G}_1$)\\
		toProcess.push($\mathcal{G}_2$)\\
		\textbf{break}\\
		}
		}
		}
\nonl 	\chline\\
\else
\fi
 	\textbf{Func}: \textsf{CandidateTroublesomeTasks}: \\
 	\textbf{Input}: DAG $\mathcal{G}$; \textbf{Output}: list $\mathcal{L}$ of sets ${\tt T, O, P, C}$\\
\nonl 	\codecomment{// choose a candidate set of troublesome tasks; per choice, divide $\mathcal{G}$ into four sets}\\
 	$\mathcal{L} \leftarrow \varnothing$\\
 	$\forall v \in \mathcal{G}, {\tt LongScore}(v) \leftarrow {v.\mbox{duration}} / {\max_{v' \in \mathcal{G}} v'.\mbox{duration}}$\\
 	$\forall v \in \mathcal{G}, v \mbox{ in stage } s, {\tt FragScore}(v) \leftarrow {\tt TWork}(s) / {\tt ExecutionTime}(s)$\\
 	\ForEach{$l \in \delta, 2\delta, \ldots 1$}
 	{
 	\ForEach{$f \in \delta, 2\delta, \ldots 1$}
 	{
 	${\tt T} \leftarrow \{v \in \mathcal{G} | {\tt LongScore}(v) \geq l \mbox{ or } {\tt FragScore}(v) \leq f\}$\\
 	${\tt T} \leftarrow {\tt Closure}({\tt T})$\\
 	\lIf{${\tt T} \in \mathcal{L}$}{{\bf continue}  \codecomment{ // ignore duplicates}}
 	${\tt P} \leftarrow \bigcup_{v \in {\tt T}}{\mathcal{A}(v, \mathcal{G})}$; 
 	${\tt C} \leftarrow \bigcup_{v \in {\tt T}}{\mathcal{D}(v, \mathcal{G})}$;\\
 	$\mathcal{L} \leftarrow \mathcal{L} \cup \{{\tt T, \mathcal{V} - {\tt T} - {\tt P} - {\tt C}, P, C}\}$\\
 	}
 	}
 \nonl	\chline\\
 	}
 	\ncaption{Identifying various candidates for troublesome tasks and dividing the DAG into four subsets.\label{fig:ps_sched_helpers}
	}
\end{algorithm}

\subsection{Compactly placing tasks}
\label{sec:eachset} 
Given a subset of tasks and a partially occupied space, how best 
to pack the tasks while respecting dependencies? 
One can choose to place the parents first or the children first. We call these
the {\em forward} and {\em backward} placements respectively. More formally, 
the forward placement recursively picks a task all of whose ancestors 
have already been placed on the space and puts it at the earliest possible time after its latest finishing ancestor.
The backward placement is analogously defined.
Intuitively, both placements respect dependencies but can lead to very different schedules
since greedy packing yields different results based on 
which tasks are placed first. Figure~\ref{fig:ps_sched_helpers1}:${\tt PlaceTasksF}$ 
shows one way to do this. Traversing the tasks in either placement has $n\log n$
complexity for a subset of $n$ tasks and if there are $m$ machines, placing tasks 
greedily has $n \log(mn)$ complexity.

\subsection{Subset orders that guarantee feasibility}
\label{sec:greedypacking} 
For each division of DAG into subsets~${\tt T, O, P, C}$, 
{\name} considers these four orders:~$\mbox{\small\sf TOCP}, \mbox{\small\sf TOPC}, \mbox{\small\sf TPOC}$ or $\mbox{\small\sf TCOP}$. 
That is, in the {\small\sf TOCP} order, it first places all tasks in {\small\sf T}, then tasks in {\small\sf O}, then tasks in {\small\sf C} and finally all tasks in {\small\sf P}.
Intuitively, this helps because the troublesome subset ${\small\sf T}$ is always placed first. 
Further, we will shortly prove that these are the only orders beginning with ${\small\sf T}$ that will avoid {\em dead-ends}.

A subtle issue is worth discussing. 
Only one of the forwards or backwards placements~(described above in~\xref{sec:eachset}) are appropriate for some subsets of tasks.
For example, tasks in ${\tt P}$ cannot be placed {\em forwards} since some descendants of these tasks may already have been placed~(such as those in ${\tt T}$).
As we saw above, the forwards placement places a task after its last finishing ancestor but ignores descendants and can hence violate dependencies if used for ${\tt P}$.
Analogously, tasks in ${\tt C}$ cannot be placed {\em backwards}.
Tasks in ${\tt O}$ can be placed in one or both placements, depending on the inter-subset order. 
Finally, since the tasks in ${\tt T}$ are placed onto an empty space they can be placed either forwards or backwards.
Formally, this logic is encoded in Figure~\ref{fig:ps_sched_helpers1}:${\tt TrySubsetOrders}$. We prove the following lemma.

\vskip -.3in
\begin{lemma} 
(Correctness) The method described in~\xref{sec:searchspace}--\xref{sec:greedypacking} 
satisfies all dependencies and is free of {\em dead-ends}.
(Completeness) Further, the method explores every order that places troublesome tasks first and is free of {\em dead-ends}.
\end{lemma}

\noindent
We omit a detailed proof due to space constraints.
Intuitively however, the proof follows from (1) all four subsets are closed and hence intra-subset dependencies
are respected by both the placements in~\xref{sec:eachset}, (2) the inter-subset orders and the corresponding
restrictions to only use forwards and/or backwards placements specified in~\xref{sec:greedypacking} ensure
dependencies across subsets are respected and finally, (3) every other order that begins with ${\tt T}$ either violates dependencies
or leads to a dead-end~(e.g., in ${\tt TPCO}$, placing tasks in ${\tt O}$ can dead-end because some ancestors 
and descendants have already been placed).
\cut{
\begin{proof}
\vskip -.05in
The proof follows by construction. Observe that all four subsets ${\tt T}, {\tt O}, {\tt P}, {\tt C}$
are closed. That is, if two tasks belong to the same subset every task on path between these tasks
also belongs to that subset. Further, our choice of placement for each subset ensures that a 
task is placed forwards if and only if no descendants have been placed already~(and vice-versa for backwards).
Finally, for every other order of placing tasks, we can find a counter-example DAG that either
violates dependencies or leads to a dead-end.
\end{proof}
}
\begin{algorithm}[t!]
{\scriptsize
 	 \textbf{Func:} \textsf{PlaceTasksF}: \codecomment{// forward placement}\\
 	 \textbf{Inputs}: $\mathcal{V}$: subset of tasks to be placed, $\mathcal{S}$: space (partially filled), $\mathcal{G}$: a DAG\\
 	 \textbf{Output}: a new space with tasks in $\mathcal{V}$ placed atop $\mathcal{S}$\\
 	 $\mathcal{S} \leftarrow {\tt Clone}(\mathcal{S})$\\
 	 finished placement set $\mathcal{F} \leftarrow \{ v \in \mathcal{G} | v \mbox{ already placed in } \mathcal{S}\}$\\

 	 \While{\mbox{\bf true}}
 	 {
 	 ready set $\mathcal{R} \leftarrow \{ v \in \mathcal{V} - \mathcal{F} | \mathcal{P}(v, \mathcal{G}) \mbox{ already placed in } \mathcal{S}\}$\\
 	  	\lIf{$\mathcal{R} = \varnothing$}{{\bf break}  \codecomment{ // all done}}
 	 $v' \leftarrow \mbox{task in } \mathcal{R} \mbox{ with longest runtime}$\\
 	 $t \leftarrow \max_{v \in \mathcal{P}(v, \mathcal{G})} {\tt EndTime}(v, \mathcal{S})$\\
 	 \codecomment{// place $v'$ at earliest time $\geq t$ when its resource needs can be met}\\
 	 $\mathcal{F} \leftarrow \mathcal{F} \cup v'$\\ 
 	 }
\nonl 	 \chline\\
 	  	\textbf{Func:} \textsf{PlaceTasks}($\mathcal{V}, \mathcal{S}, \mathcal{G}$):\codecomment{// inputs and output are same as PlaceTasksF}\\
 	\Return{$\min\left({\tt PlaceTasksF}(\mathcal{V}, \mathcal{S}, \mathcal{G}), {\tt PlaceTasksB}(\mathcal{V}, \mathcal{S}, \mathcal{G})\right)$}\\
\nonl 	\chline\\

\iflongversion 	 
 	 \textbf{Func:} \textsf{PlaceTasksB}: \codecomment{// only backwards, inputs and outputs same as ${\tt PlaceTasks}$}\\
 	 \textbf{Input}: $\mathcal{V}$, $\mathcal{S}$, $\mathcal{G}$; \textbf{Output}: $\mathcal{S}'$\\
 	 $\mathcal{S} \leftarrow \mbox{\tt Clone}(\mathcal{S}_{\mbox{in}})$\\
 	 finished placement set $\mathcal{F} \leftarrow \{ v \in \mathcal{G} | v \mbox{ already placed in } \mathcal{S}\}$\\
 	 ready set $\mathcal{R} \leftarrow \{ v \in \mathcal{V} - \mathcal{F} | \mathcal{C}(v, \mathcal{G}) \mbox{ already placed in } \mathcal{S}\}$\\
 	 \While{$R \neq \varnothing$}
 	 {
 	 $v' \leftarrow \mbox{task in } \mathcal{R} \mbox{ with longest runtime}$\\
 	 $t \leftarrow \min_{v \in \mathcal{C}(v, \mathcal{G})} {\tt BeginTime}(v, \mathcal{S})$\\
 	 \codecomment{// place $v'$ at the latest time $\leq t - v'.\mbox{duration}$ when its resource demands can be met}\\
 	 $\mathcal{F} \leftarrow \mathcal{F} \cup v'$\\ 
 	 $\mathcal{R} \leftarrow \mathcal{R} \cup \{v \in \mathcal{P}(v', \mathcal{G}) -\mathcal{F}  | \mathcal{C}(v, \mathcal{G}) \mbox{ placed in } \mathcal{S}\}$\\
 	 }
\else
 	 \textbf{Func:} \textsf{PlaceTasksB}: \codecomment{// only backwards, analogous to ${\tt PlaceTasksF}$.}\\
\fi
\nonl 	 \chline\\
 	 
 	 \textbf{Func:} \textsf{TrySubsetOrders}:\\
 	 \textbf{Input}: $\mathcal{G}$: a DAG, $\mathcal{S}_{\mbox{in}}$: space with tasks in ${\tt T}$ already placed\\
 	 \textbf{Output:} Most compact placement of all tasks.\\

 	 $\mathcal{S}_1, \mathcal{S}_2, \mathcal{S}_3, \mathcal{S}_4 \leftarrow {\tt Clone}(\mathcal{S}_{\mbox{in}})$\\

	\Return $\min$( \codecomment{// pick the most compact among all feasible orders}\\
	${\tt PlaceTasksF}({\tt C}, {\tt PlaceTasksB}({\tt P}, ({\tt PlaceTasks}({\tt O}, \mathcal{S}_1, \mathcal{G})), \mathcal{G}), \mathcal{G})$,\codecomment{// ${\tt OPC}$}\\
	${\tt PlaceTasksB}({\tt P}, {\tt PlaceTasksF}({\tt C}, ({\tt PlaceTasks}({\tt O}, \mathcal{S}_2, \mathcal{G})), \mathcal{G}), \mathcal{G})$,\codecomment{// ${\tt OCP}$}\\
	${\tt PlaceTasksB}({\tt P}, {\tt PlaceTasksB}({\tt O}, ({\tt PlaceTasksF}({\tt C}, \mathcal{S}_3, \mathcal{G})), \mathcal{G}), \mathcal{G})$,\codecomment{// ${\tt COP}$}\\
	${\tt PlaceTasksF}({\tt C}, {\tt PlaceTasksF}({\tt O}, ({\tt PlaceTasksB}({\tt P}, \mathcal{S}_4, \mathcal{G})), \mathcal{G}), \mathcal{G})$\codecomment{// ${\tt POC}$}\\
	);\\
 \nonl	 \chline\\

%
 	}
	\ncaption{Pseudocode for the functions described in~\xref{sec:eachset} and~\xref{sec:greedypacking}.\label{fig:ps_sched_helpers1}
	}
\end{algorithm}
\subsection{Enhancements}
\label{sec:enhancements}
\label{sec:simplify} 
We note a few enhancements. 
First, due to barriers it is possible to partition a DAG into parts that are totally ordered.
Hence, any  schedule for the DAG is a concatenation of per-partition schedules.
This lowers complexity because one execution of ${\tt BuildSchedule}$ will 
be replaced by several executions each having fewer tasks. $24$\% of the production DAGs can be split into four or more
parts.  Second, and along similar lines, whenever possible we reduce complexity by 
reasoning over stages. Stages are collections of tasks and are $10$ to $10^3$ times fewer in number than tasks. 
Finally, 
we carefully choose our data-structures~(e.g., a time and resource indexed hash map of free regions 
in space) so that the most frequent operation, picking a region in resource-time space where a task will fit as described in~\xref{sec:eachset}, can be executed efficiently. 
\cut{
We recursively cut until no more cuts are possible.
Figure~\ref{fig:ps_sched_helpers}:${\tt CutDAGs}$ shows how to do this in linear time.  
Cuts are made possible by barriers in the computation.
Cutting helps because the complexity of the rest of schedule construction is $O(n^2)$
for $n$ tasks~(it is typically $n\log n$ except for some corner cases).
In our traces, $27$\% of DAGs break into two parts, $13$\% into three parts, and 
}

\eat{
\subsection{Placing the troublesome tasks in space}
Recall that the troublesome tasks~(i.e., those in {\small\sf T}) are placed in space first, in either
the forward or the backward order~(\xref{sec:ideas}). 

We greedily pack these tasks without any unnecessary gaps
and within the least amount of total time.  Given a set of tasks to
place, {\name} greedily picks the task with the longest duration and
places it at the best possible time that satisfies both resource
capacity and dependencies. To respect dependencies, the best time to place is after the last ancestor finishes (for the forward ordering) or before the first descendant begins (for the backward ordering). To reduce
fragmentation, the best time is the one wherein the task's demands are
below available capacity.

\subsection{Packing the remaining tasks}
After placing the tasks in {\small\sf T},
recall that {\name} follows one of the four orders:  $\mbox{\small\sf TOCP}, \mbox{\small\sf TOPC}, \mbox{\small\sf TPOC}$ or $\mbox{\small\sf TCOP}$. 
That is, with {\small\sf TOCP}, it first places all
tasks in {\small\sf O}, then all tasks in {\small\sf C} and finally
all tasks in {\small\sf P}.  In
practice, we find that binding the subset with more work first leads to the more compact
schedule.  We believe this is because it avoids the case where binding a few tasks
constrains the binding of many others. Furthermore, tasks 
in subset {\small\sf C} can be placed in only the forward order.  
This is because in all sequences only the ancestors 
of tasks in {\small\sf C} happen to be placed already. Analogously, tasks in {\small\sf P} 
can only be placed in the backward order. Tasks in {\small\sf O} can be placed in either or both orders, depending on the sequence. For example, in ${\tt TOPC}$, tasks in {\small\sf O} can be placed in either order since none of their parents or children 
have been placed. However, in ${\tt TPOC}$, they can only go forward. With these added constraints, {\name} 
places tasks in each subset in the same manner as described in~\xref{sec:eachset}. 
See Fig.~\ref{fig:ps_sched_helpers}:${\tt BuildSchedule}$. Notice that the best-possible 
check above greedily places these tasks to fill holes.
}


\begin{algorithm}[t!]
{\scriptsize
\nonl	\chline\\	
		\textbf{Func:} \textsf{FindAppropriateTasksForMachine}:\\
		\textbf{Input:} ${\bf m}$: vector of available resources at machine; 
		$\mathcal{J}$: set of jobs with task details$\{t_{\mbox{duration}}, {\bf t}_{\mbox{demands}}, t_{\mbox{priScore}}\}$; ${\tt deficit}$: counters for fairness;\\
		\textbf{Parameters:} $\kappa$: unfairness bound; ${\tt rp}$: remote penalty\\
		\textbf{Output:} $\mathcal{S}$, the set of tasks  to be allocated on the machine\\
		
		$\mathcal{S} \leftarrow \varnothing$


		\While{ true }
		{
		\ForEach{$\mbox{task } t$}
		{
		$\{{\tt pScore}_t, {\tt oScore}_t\} \leftarrow \{0, 0\}$\\
		${\tt rPenalty}_t \leftarrow t$ is locality sensitive ? ${\tt rp}$ : $1$\\
		
		\If{${\bf t_{\mbox{demands}}} \leq {\bf m}$ \codecomment{// fits?}}
		{
		${\tt pScore}_t \leftarrow \left( {\bf m} \cdot {\bf t_{\mbox{demands}}} \right) {\tt rPenalty}_t$\codecomment{// dot product}\\
		}
		\Else
		{
		\iflongversion
		\codecomment{// what-if analysis: ``overbook or wait''.}\\
		$\forall \mbox{tasks }t'\mbox{ affected by }t\mbox{ running at } m, \mbox{let } {\tt before}(t'), $\\
		${\tt after}(t') \mbox{ be expected completion times before and}$\\
		$\mbox{after placing }t\mbox{ at }m$\\
		${\tt benefit} = {\tt nextSchedOpp} + t_{\mbox{duration}} - {\tt after}(t)$\\
		${\tt cost} = \sum_{\mbox{aff. tasks }t'} \left({\tt after}(t') - {\tt before}(t')\right)$\\
		\lIf{${\tt benefit} > {\tt cost}$}{${\tt oScore}_t = {\tt benefit}-{\tt cost}$}
		\else
		compute ${\tt oScore}_t$ \codecomment{// overbooking score omitted for brevity.}\\
		\fi
		}
		$\mbox{ job } j \ni t, {\tt srpt}_j \leftarrow \sum_{\mbox{pending u}  \in j} u_{\mbox{duration}} * \left|{\bf u_{\mbox{demands}}}\right|$\\
		
		\fbox{${\tt perfScore}_t \leftarrow t_{\mbox{priScore}}\left\{{\tt pScore}_t, {\tt oScore}_{t}\right\} - \eta {\tt srpt}_j$}
		}
		
		$t^{\mbox{best}} \leftarrow \arg \max \{{\tt perfScore}_t | t\}$\codecomment{// task with highest perf score}\\
		\lIf{$t^{\mbox{best}} = \varnothing$}{{\bf break} \codecomment{// no new task can be scheduled on this machine}}
		$g' \leftarrow \mbox{jobgroup with highest deficit counter}$\\
		\fbox{
		\lIf{${\tt deficit}_{g'} \geq \kappa C$}{ $t^{\mbox{best}} \leftarrow \arg\max\{ {\tt perfScore}_t | t \in g'\}$}}
		
		$\mathcal{S} \leftarrow \mathcal{S} \cup t^{\mbox{best}}$\\
		
		\iflongversion
		${\bf m} \leftarrow \left[{\bf m} - {\bf t^{\mbox{best}}_{\mbox{demands}}}\right]_{0+}$ \\
		\else
		\codecomment{// detail: reduce available resources ${\bf m}$.}
		\fi

		${\tt deficit}_g \leftarrow {\tt deficit}_g +$
\fbox{
		$f({\bf t^{\mbox{best}}_{\mbox{demands}}}) *\left\{\begin{array}{ll} {\tt fairShare}_g - 1 & t\in\mbox{ jobgroup }g\\ {\tt fairShare}_g & \mbox{otherwise}\end{array} \right.$
		}\\

		}
\nonl		\chline\\
}
	\ncaption{Simplified pseudocode for the online component.\label{fig:ps_online}
	}
\end{algorithm}

\section{Scheduling many DAGs}
\label{sec:runtime}
\label{subsec:runtime}
We describe our online algorithm that matches tasks to machines while co-ordinating discordant objectives: fairness, packing and enforcing the per-DAG schedules built by~\xref{sec:design}. We offer the pseudocode in Figure~\ref{fig:ps_online} for completeness but focus only on (1) how the various objectives are co-ordinated and (2) how unfairness is bounded.

The pseudocode shows how various individual objectives are estimated. Packing score per task ${\tt pScore}_t$ is a dot product between task demands and available resources~\cite{tetris}. Using remote resources un-necessarily, for example by scheduling a locality-sensitive task~\cite{locality} at another machine, is penalized by the value ${\tt rPenalty}_t$. The value ${\tt srpt}_j$ estimates the remaining work in a job and is used to prefer short jobs which lowers average job completion time. We claim no novelty thus far.  Suppose that $t_{\tt priScore}$ is the order over tasks required by the schedule from~\xref{sec:design}; $t_{\tt priScore}$ is computed by ranking tasks in increasing order of their begin time and then dividing the rank by the number of tasks in the DAG so that the value is between $1$ (task that begins first) and $0$ (for the last task).

An initial combination of the above goals happens in the computation
of ${\tt perfScore}_t$. See the first box in Figure~\ref{fig:ps_online}. 
A task will have non-zero ${\tt pScore}_t$
only if its demands fit within available resources. 
Else, it can have a non-zero ${\tt oScore}_t$ if it is worth overbooking. 
We use a lexicographic ordering between these two values. That is, tasks
with non-zero ${\tt pScore}$ beat any value of ${\tt oScore}$.
Multiplying with $t_{\tt priScore}$ steers the search towards tasks earlier in the constructed schedule.
Finally, $\eta$ is a parameter that is automatically updated based on the 
average ${\tt srpt}$ and ${\tt pScore}$. Subtracting $\eta\cdot {\tt srpt}_j$ prefers shorter jobs.
Intuitively, the combined value ${\tt perfScore}_t$ softly enforces the various objectives. For example, if some 
task is preferred by all individual objectives~(belongs to shortest job, is most packable, is next in the preferred schedule), then it will have the highest ${\tt perfScore}$. When the objectives are discordant, colloquially, the task preferred by a majority of objectives will have the highest ${\tt perfScore}$. 

To bound unfairness, we use one additional step. 
We explicitly measure unfairness using deficit counters~\cite{drr}. 
When the maximum unfairness~(across jobgroups or queues) is above the specified threshold~$\kappa C$, where $C$ is the cluster capacity, {\name} picks only among tasks belonging to the most unfairly treated jobgroup. This is shown in the second box in Figure~\ref{fig:ps_online}. Otherwise {\name} picks the task with the highest ${\tt perfScore}$. It is easy to see that this bounds unfairness by $\kappa C$. 
Further, we can support a variety of fairness schemes by choosing how to change the deficit counter.
For example, choosing $f() = 1$ mimics slot fairness~(see third box in Figure~\ref{fig:ps_online}), and $f() =$ demand of the dominant resource mimics DRF~\cite{drf}.

\eat{

The online scheduler tries to schedule tasks in this priority order.
Overall, the runtime component has the following role: whenever resources become available, choose which pending task to assign these resources to, while respecting these requirements.


\vskip .05in
\noindent{\bf Priority: }
The priority of a task is its start-time based order in the constructed schedule. Tasks earlier in the schedule have higher priority. Priority is in $[0,1]$, higher is better. Priorities are normalized across DAGs. \sriram{How?}

\vskip .05in
\noindent{\bf Overbooking \& Packing: }
A packability score is computed as the dot product between the tasks' demands and server's available resources. The score is in $[0, 1]$. Higher is better. Remote tasks, that have inputs elsewhere in the cluster, are penalized a fraction of their score ({\em How much?}). For tasks that do not fit, {\name} computes an overbooking score (see~\xref{sec:overbooking}). This o-score is also in $[0, 1]$ and higher is better. And, multiples the packability score with the o-score. Doing so, intuitively picks online tasks that lower fragmentation~(high dot product), prefer those that use only local resources~(remote penalty), prefer those that fit and among the rest prefer those that benefit more from overbooking~(multiplying with o-score). \sriram{Need to be precise here.}

\vskip .05in
\noindent{\bf Shortness: }
Per running job, {\name} computes a score that reflects the amount of work remaining in that job. The intuition is to prefer jobs with less remaining work. Similar to SRPT~\cite{srpt}, we use this to lower average job completion time.

Combining the above measures, the {\em performance-score} of a task is a weighted sum of the packability score and shortness score, multiplied by task's priority.

\vskip .05in
\noindent{\bf Bounded Unfairness: }
Per queue or job group, {\name} maintains deficit counters to reflect how far the resources allocated to that group are from their desired share. The counters are per resource.  As with DRR~\cite{drr}, negative counter value indicates allocation is above-share and positive below-share. Per queue, the resource with the smallest counter value is the dominant resource. The more negative that value is the greater the amount of resource that is unfairly allocated to the queue. Using these counters, {\name} operates in one of two modes: it either (a) picks the best performance-scoring task from among those that keep maximal unfairness below a threshold or (b) if no such task is available picks the best performance-scoring task belonging to the queue that is maximally below share at that time.
Intuitively, (b) is a fair allocation and (a) picks best performing choice within bounded unfairness.
}

\iflongversion
\subsection{Judicious overbooking}
\label{sec:overbooking}

The key tussle in overbooking is as follows.
Overbooking improves throughput. For example, suppose tasks with duration $t$ require $0.6r$.
Running two tasks instead of one improves 
throughput by $60\%$ from $1$ task/$t$ to $2$ tasks/$1.2t$. 
Here, we assume task runtime increases linearly with overbooking amount.
Note that the potential gains from overbooking depend on the amount of idle resources that are 
reclaimed by overbooking~($.4r$ in the above case).
Conversely, overbooking is counter productive if resources will become free at some other machine soon. 
Suppose another machine can run the second task at time $+\varepsilon$. Without overbooking, 
tasks will finish at $\{t, t+\varepsilon\}$ and with overbooking both finish at $1.2t$. Even in simple settings, optimal overbooking is NP-hard~\cite{icalp14}.

{\name} offers a heuristic for overbooking.
First, it uses micro-benchmarks to determine how  task runtimes will be delayed when each resource is overbooked.
These functions are concave and vary across resource types.
Next, per potential task to overbook, {\name} runs a what-if analysis to decide between overbooking and 
waiting.
We compute the expected completion times of all affected tasks after overbooking. 
Note that resource overbooking delays these tasks
but the extent of delay can vary.
The ${\tt benefit}$ of overbooking is how much earlier would the new task finish with overbooking versus having to wait for next-free-resource.
The ${\tt cost}$ is the increase in runtime of all the other tasks due to overbooking.
Thus, overbooking score equals ${\tt benefit} - {\tt cost}$.
\else
\fi

\section{A new lower bound}
\label{sec:betterlb}
We develop a new lower bound on the completion time of a DAG of tasks.
As we saw in~\xref{subsec:potential-gains}, previously known lower bounds are very loose.
Since the optimal solution is intractable to compute, without a good lower bound,
it is hard to assess the quality of a heuristic solution such as {\name}.

Equations~\ref{eqn:dag_cplength} and \ref{eqn:dag_twork}
describe the known bounds: critical path length ${\tt CPLen}$ and total work ${\tt TWork}$. 
Equation~\ref{eqn:newlb} is (a simpler form of) our new lower bound. 
At a high level, the new lower bound uses some structural properties of these job DAGs.
Recall that DAGs can be split into parts that are totally ordered~(\xref{sec:enhancements}). 
This lets us pick the best lower bound for each part independently. For a DAG that 
splits into a chain of tasks followed by a group of independent tasks, we could use ${\tt CPLen}$ of the chain plus the ${\tt TWork}$ of the group. 
A second idea is that on a path through the DAG, at least one stage has to complete entirely.
That is, all of the tasks in some stage and at least one task in each other stage on the path
have to complete entirely. This leads us to the ${\tt ModCP}_{\mathcal{G}}$ formula in Equation~\ref{eqn:dag_modcp}
where one stage $s$ along any path $p$ is replaced with the total work in that stage.
\iflongversion
A third idea is that some stages have all-to-all dependencies to all parents and children.
That is all its tasks have to finish {\em after} the last parent task finishes and {\em before} the first child task can start. 
For such stages, we can replace them with their total work.
To see why this helps, consider a stage of $n$ tasks with duration $d$ and demand vector ${\bf r}$.
This stage will now contribute $\max(nd\frac{{\bf r}}{C}, d)$ instead of $d$. When $n$ or ${\bf r}$ are large, this leads to a larger CPLength. 
$34$\% of the  stages in our production DAGs have this property.
Finally, we group tasks having identical parents and children even though their functions differ. E.g., two map stages 
preceding the same reduce stage~(in a join). Larger groups let us cumulatively account for their tasks
which helps the second and third changes above.
\else
A few other ideas are omitted for brevity.
\fi

The take-away is that the new lower bound ${\tt NewLB}$ is much tighter and allows
us to show that {\name} is close to ${\tt OPT}$; since by definition of a lower bound {\name} $\geq {\tt OPT} \geq {\tt NewLB}$.

\begin{figure}[t!]
{\scriptsize
\begin{subequations}
\begin{align}
{\tt CPLen}_{\mathcal{G}}&= \max_{\mbox{path } p \in \mathcal{G}} \sum_{\mbox{task } t \in p} t_{\mbox{duration}}\label{eqn:dag_cplength}\\
{\tt TWork}_{\mathcal{G}} & =  \max_{\mbox{resource } r} \frac{1}{C_r} \sum_{t \in \mathcal{G}} t_{\mbox{duration}} t^r_{\mbox{demands}}\label{eqn:dag_twork}\\
\iflongversion
{\tt ModCP}_{\mathcal{G}} & = \max_{\mbox{path } p \in \mathcal{G}} \max_{\mbox{stage } s \in p} \left( \max({\tt TWork}_s, {\tt CPLen}_s) + \sum_{g \in p-\{s\}} \min_{\mbox{task } t \in g} t_{\mbox{duration}}\right) \label{eqn:dag_modcp}\\
\else
{\tt ModCP}_{\mathcal{G}} & = \max_{p \in \mathcal{G}} \max_{s \in p} ( \max({\tt TWork}_s, {\tt CPLen}_s) + \sum_{s' \in p-\{s\}} \min_{t \in s'} t_{\mbox{dur.}}) \label{eqn:dag_modcp}\\
\fi
{\tt NewLB}_{\mathcal{G}} & = \sum_{\mathcal{G'} \in {\tt Partitions}(\mathcal{G})}
\max({\tt CPLen}_{\mathcal{G'}}, {\tt TWork}_{\mathcal{G'}}, {\tt ModCP}_{\mathcal{G'}})\label{eqn:newlb}
\end{align}
\end{subequations}

}
\vspace{-2.5mm}\ncaption{
Lower bound formulas for DAG $\mathcal{G}$; $p$, $s$, $t$ denote a path through the DAG, a stage and a task respectively. 
$C,$ here, is the capacity available for this job.
We developed ${\tt ModCP}$ and ${\tt NewLB}$.
\label{fig:eqns}}
\end{figure}

\iflongversion
We have the following lemma:
\vskip -.25in
\begin{lemma} \label{thm:lb_tight}
${\tt NewLB}$~(Eqn.~\ref{eqn:newlb}) is a valid lower bound for DAG runtime,
and ${\tt NewLB} \geq \max({\tt CPLen}, {\tt TWork})$.
\end{lemma}
\begin{proof}
First, observe that the maximum of the lower bounds is also a lower bound.
Second, observe that by definition the DAG is cut into parts that have no overlap.
Hence, the lower bound of the DAG is equal to the sum of the lower bounds of the parts.
This supports Eqn.~\ref{eqn:newlb}.

Next, grouping stages with identical parent and child stages is appropriate because (a) there are no dependencies
between these tasks and (b) the definition of a stage has been a group of independent tasks that can run in parallel
and adhere to a specific dependence pattern with tasks in parent and child stages.

Using the total work to be done in a stage instead of the duration of a single task is appropriate.
This holds because either all of the work has to be done in line (when all parents and children have all-to-all edges)
or at least one stage on a path through the DAG will have to finish all of its work. 
This supports Eqns.~\ref{eqn:stage_moddur},~\ref{eqn:dag_modcp}.

Stepping back, this new lower bound was possible because of the abundance of groups of independent tasks that is common in data-parallel DAGs.
As we did not relax either dependence satisfaction or resource capacity limits, this lower bound
is much tighter than other bounds based on linear programs that relax those aspects~\cite{ec15,spaa13,icalp14}.
\end{proof}
\begin{figure}[t!]
\centering
\includegraphics[width=1.5in]{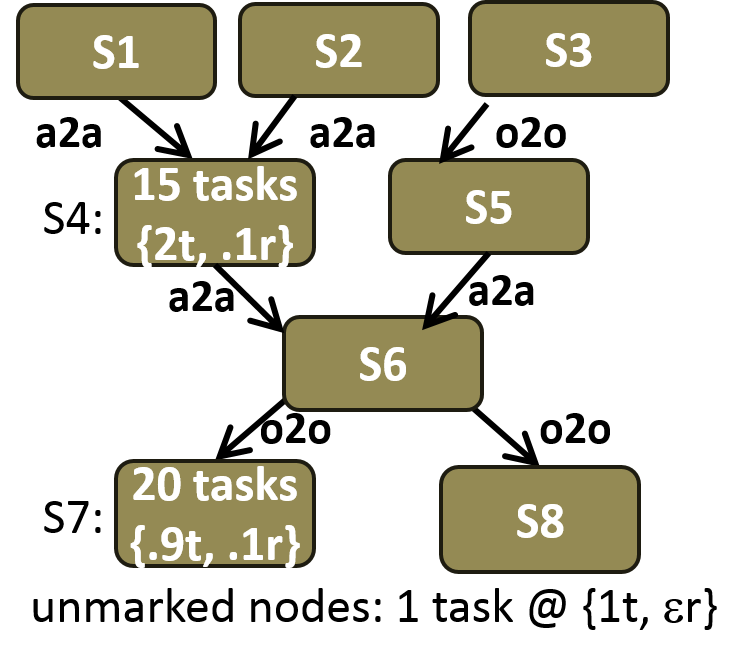}\qquad
\includegraphics[width=1.12in]{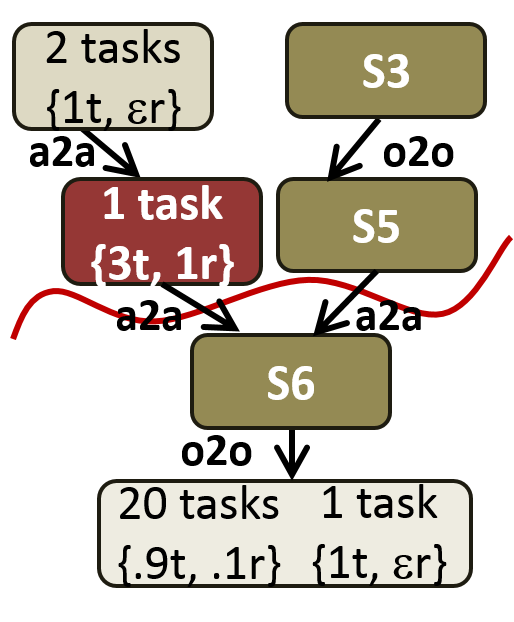}
\ncaption{Original DAG on left. Unmarked nodes have $1$ task each with duration $1t$ and demand $\varepsilon r$. 
The modified DAG used by our lower bound is on the right. Lower bound
improves by $40$\% from $5t$ to $6.8t$. Edge labels ${\tt o2o}, {\tt a2a}$ denote one-to-one and all-to-all 
dependencies between the tasks in the corresponding stages.\label{fig:lb_example}}
\end{figure}
Fig.~\ref{fig:lb_example} shows an example DAG and the various lower bounds. 
Note that, ${\tt CPLen}$ is $5t$ and ${\tt TWork}$ is $4.8t$.  
The longest critical path is $\{S_1, S_2\} \rightarrow S_3 \rightarrow S_6 \rightarrow S_8.$
The only stages with non-trivial work are $S_4$ at $3t$ and $S_7$ at $1.8t$. We assume the cluster capacity $C$ is $1r$.

The DAG on the right shows the modifications used by {\name}.
We can cut at $S_6$ per logic in~\xref{sec:simplify}.
We can replace $S_4$ with its $s_{\mbox{modDur}}$ of $3t$~(see Eqn.~\ref{eqn:stage_moddur}).
We can group $S_7$ and $S_8$ and replace them with their $s_{\mbox{modDur}}$ of $1.8t$.
We can also merge $S_1$ and $S_2$ but their $s_{\mbox{modDur}}$ remains $1t$ because of their 
small resource requirement.
The overall lower bound is now $6.8t$; $4t$ from ${\tt CPLen}$ of the DAG part
above the cut and $2.8t$ from the ${\tt ModCP}$ of the part below the cut.
\else
\fi

\section{{\namesec} System}
\label{sec:system}
\label{sec:app_system}
We have implemented the runtime component~(\xref{sec:runtime}) in the
Apache YARN resource manager~(RM) and the schedule
constructor~(\xref{sec:design}) in the Apache Tez application
master~(AM).  Our schedule constructor implementation finishes in tens of
seconds on all of the DAGs used in experiments; this is in the same ballpark 
as the time to compile and query-optimize these DAGs. Further, recurring
jobs use previously constructed schedules.
\cut{
The YARN RM has an online scheduler that assigns
tasks to machines; the tasks can belong to jobs from many
frameworks~\cite{hadoop,spark}. 
}
Each DAG is managed by 
an instance of the Tez AM which closely resembles other popular frameworks such as FlumeJava~\cite{flumejava} and Dryad~\cite{dryad}. 
The per-job AMs negotiate with the YARN RM for containers to run the job's tasks;
each container is a fixed amount of various resources.  
As part of implementing {\name}, we expanded
the interface between the AM and RM to pass additional
information, such as the job's pending work and tasks'
demands, duration and preferred order.  Due to anonymity
considerations, we are unable to share full details of our code
release.  Here, we describe two key implementation challenges: (a) constructing
profiles of tasks' resource demands and
duration~(\xref{subsec:res_demands}), and (b) efficiently implementing
the new online task matching logic~(\xref{subsec:bundling}).

\cut{Due to anonymity considerations we are unable to share full
details of our code release; however, we have filed JIRAs~(a JIRA is a
design document and a code patch) for the various components. Some
aspects have shipped but other patches are currently under
review. Here, we describe two key challenges: (a) constructing
profiles of tasks' resource demands and
duration~(\xref{subsec:res_demands}), and (b) efficiently implementing
the new online task matching logic~(\xref{subsec:bundling}).
}

\subsection{Profiling Tasks' Requirements}
\label{subsec:res_demands}
We estimate and update the tasks' resource demands and durations as follows.
Recurring jobs are fairly common in production clusters~(up to 40\%~\cite{rope,ccloud,optimus}),
executing periodically on newly arriving data (e.g., updating metrics
for a dashboard). For these jobs, {\name} extracts statistics from
prior runs. In the absence of prior history, we rely on two aspects of data
analytics computations that make it amenable to learn profiles at runtime.
(1) Tasks in a stage~(e.g., map or reduce) have similar profiles and (2) tasks often
run in multiple waves due to capacity limits.  {\name}
measures the progress and resource usage of tasks at runtime.
Using the measurements from in-progress and completed tasks, {\name} refines estimates for the
remaining tasks.
Our evaluation will demonstrate the effectiveness of this approach.
\cut{
That is, starting from ball-park estimates of profiles and then
refining them with runtime measurements allows any mistakes that
the offline schedule constructor may make due to imprecise estimates
to be fixable at runtime.}



\subsection{Efficient Online Matching: Bundling}
\label{subsec:bundling}
We have redesigned the online scheduler in YARN that
matches machines to tasks. 
\cut{
The YARN RM's scheduler performs resource allocation by matching machines
to tasks that have expressed a preference to run on them.  As
part of this work, we significantly enhanced this matching
logic.}
From conversations with Hadoop committers, these code-changes 
help improve matching efficiency and code readability.  

Some background: The matching logic is heartbeat based.
When a machine heartbeats to the RM, the allocator (1) picks an appropriate task to allocate to that machine, (2) adjusts its data structures
(such as, resorting/rescoring) and (3) repeats these steps until all resources on the node have
been allocated or all allocation requests have been satisfied.

As part of this work, we support
{\em bundling} allocations. That is, rather than breaking
the loop after finding the first schedulable task, we maintain a set of
tasks that can all be potentially scheduled on the machine. This
so-called {\em bundle} allows us to schedule multiple tasks in one
iteration, admitting non-greedy choices over multiple tasks. 
For example, if tasks $t_1, t_2, t_3$ are
discovered in that order, it may be better to schedule
$t_2$ and $t_3$ together rather than schedule $t_1$ by itself.
We refactored the scheduler to support bundling; with configurable
choices for (1) which tasks to add to the bundle, (2) when to terminate bundling~(e.g. the bundle has a good set of tasks) and (3) which tasks
to pick from the bundle. 
\section{Evaluation}
\label{sec:eval}
Here, we report results from experiments on a $200$
server cluster and extensive simulations using $20,000$ DAGs
from production clusters. Our key findings are:
\cut{
We evaluated our implementation on a $200$ server
cluster. Each experiment ran $200$ jobs
and lasted about an hour. The jobs were randomly
chosen from publicly available benchmarks as 
well as copies of jobs that ran in Microsoft's production clusters.
Our key results are:
}

\noindent
{\bf (1) } 
In experiments on a large server cluster, relative to ${\tt Tez}$ jobs running on ${\tt YARN}$, 
{\name} improves completion time of half of the jobs by $19$\% to $31$\% across
various benchmarks. A quarter of the jobs improve by $30$\% to $49$\%.
\cut{
Relative to Tez, {\name} improves completion time
of at least $50$\% of the jobs by $19-31$\% across the various benchmarks. 
A quarter of the jobs speed up by as much as $30-49$\%~(\xref{subsubsec:jct}).
}

\vskip .05in
\noindent
{\bf (2) }
On the DAGs from production clusters, schedules constructed by {\name}
are faster by $25$\% for half of the DAGs. A quarter
of the DAGs improve by $57$\%. Further,
by comparing with our new lower bound, these schedules are optimal for
$40$\% of the jobs and within $13$\% of optimal for $75$\% of the
jobs.


As part of the evaluation, we offer detailed comparisons with many alternative
schedulers and sensitivity analysis to cluster load and parameter
choices. We also provide early results on applying {\name} to
DAGs from other domains~(\xref{subsec:other-domains}).

\cut{
that ran in production on Cosmos. These DAGs are more complex~(higher depth, various shapes) relative to those in the benchmarks~(\xref{subsec:eval_other_metrics}). 
\item
The above constructed schedules equal our new lower bound~(${\tt NewLB}$ in~\xref{sec:betterlb})
for $40$\% of the jobs. The gap is below $13$\% for over $75$\% of the jobs. 
Hence, {\name} is close to ${\tt OPT}$ for most production jobs.
\item
We report detailed parameter sensitivity analysis, vary the cluster load, and experiment
with scheduling other DAGs beyond data-parallel jobs~(\xref{subsec:eval_sa}).
}

\subsection{Setup}
\label{subsec:eval_setup}
\noindent{\bf Our experimental cluster} has $200$ servers with
two quad-core Intel E2550 processors~(hyperthreading enabled), 
$128$ GB RAM, $10$
drives, and a $10$Gbps network interface. 
The network has a congestion-free core~\cite{vl2}.

\vskip .05in
\noindent{\bf Workload: } 
Our workload mix consists of jobs from public 
benchmarks---TPC-H~\cite{tpc-h}, TPC-DS~\cite{tpc-ds}, BigBench~\cite{bigbench}, 
and jobs from a production cluster that runs Hive jobs~(E-Hive). 
We also use $20$K DAGs from a private production system in our 
simulations. In each experimental run, jobs arrival is modeled via a Poisson process
with average inter-arrival time of $25$s for $50$ minutes.
Each job is picked at random from the corresponding benchmark.
We built representative inputs and 
varied input size from GBs to tens of TBs such that the average query
completes in a few minutes and the longest finishes in under $10$
minutes on the idle cluster.
A typical experiment run thus has about $200$ jobs and lasts until the last
job finishes. The results presented are the median over three runs.

\vskip .05in
\noindent{\bf Compared Schemes: } 
We experimentally compare {\name} against the following baselines: (1)
${\tt Tez}:$ breadth-first order of tasks in the DAG running atop YARN's Capacity Scheduler
(CS), (2) ${\tt Tez + CP}:$ critical path length based order of tasks
in the DAG atop CS and (3) ${\tt Tez + Tetris}:$ breadth-first order of
tasks in the DAG atop Tetris~\cite{tetris}.

Using simulations, we compare {\name} against the following schemes: (4) ${\tt BFS}:$
breadth first order, (5) ${\tt CP}:$ critical path order, (6) ${\tt
  Random}$ order, (7) ${\tt StripPart}$~\cite{StripPart}, (8) ${\tt
  Tetris}$~\cite{tetris}, and (9) ${\tt
  Coffman-Graham}$~\cite{coffman-graham}.

All of the above schemes except (7) are work-conserving.
(4)--(6) and (8) pick greedily from among the runnable tasks but vary in the specific heuristic.
(7) and (9) require more complex schedule construction,
as we will discuss later.
\cut{
In particular, StripPart~\cite{StripPart} has the tightest
bound among algorithms that pack dependent groups of tasks.
Several of these schemes require fitting on all the resources. 
We also equip each with an overbooking logic akin to the one used by {\name}~(see~\xref{sec:overbooking}).
}
\vskip .05in
\noindent{\bf Metrics: } Improvement in job completion time is our key metric. 
Between two schemes, we measure the {\em normalized gap}
in job completion time. That is, the difference in the runtime achieved for the
same job divided by the runtime of the job with some scheme; the normalization lets us  
compare across jobs with very different runtimes. Other metrics of
interest are makespan, i.e., the time to finish a given set of jobs,
and Jain's fairness index~\cite{jain_fairness} to measure how close
the cluster scheduler comes to the desired allocations.

\cut{We primarily consider the completion time of
each job in the workload. We measure the ``normalized gap" between two
schemes. Per job, ``gap" is the difference in the runtime achieved for
that job by the two schemes.  Also, we divide the gap by the runtime
of one of the two schemes to allow comparing across jobs with very
different runtimes. We also report makespan, i.e., the time to finish
a given set of jobs and use Jain's fairness index~\cite{jain_fairness} to measure how close
the cluster scheduler comes to desired allocations.}

\iflongversion
\begin{figure*}[!]
\begin{subfigure}[b]{0.35\textwidth}
\centering
\includegraphics[width=1.85in, height=1.1in]{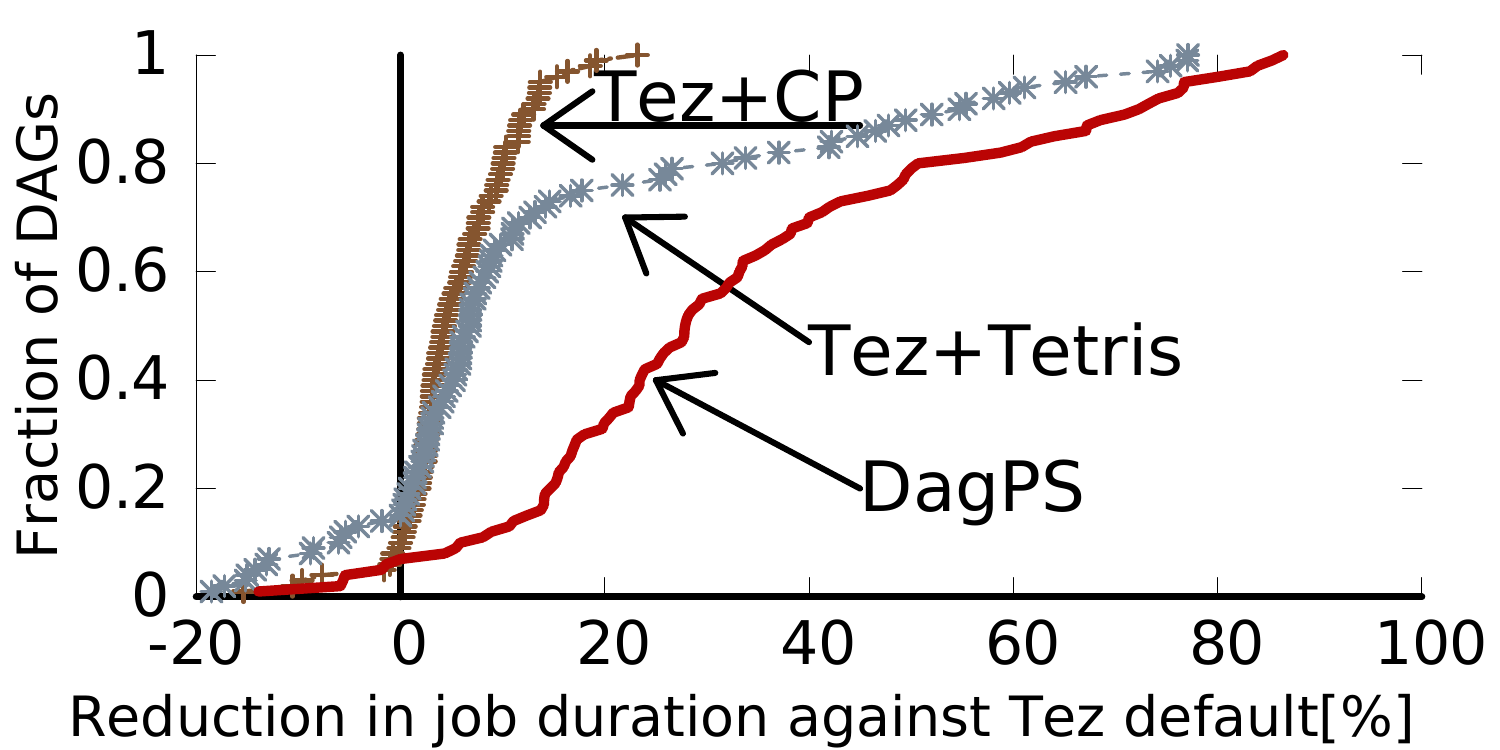}
\ncaption{CDF jobs on TPC-DS workload\label{fig:eval_cdf_tpcds}}
\end{subfigure}
\begin{subfigure}[b]{0.60\textwidth}
\begin{scriptsize}
\begin{tabular}[b]{l|rrr|rrr|rrr}
 & \multicolumn{3}{c|}{$25^{th}$ \%} & \multicolumn{3}{c|}{$50^{th}$ \%} & \multicolumn{3}{c}{$75^{th}$ \%}\\
{Workload} & {T+CP} & {T+T} & {\bf G} & {T+CP} & {T+T} & {\bf G} & {T+CP} & {T+T} & {\bf G} \\
\hline
{TPC-DS} & 2.0 & 1.9 & {\bf 16.0} & 4.1 & 6.5 & {\bf 27.8} & 8.9 & 16.6 & {\bf 45.7} \\
\hline
{TPC-H} & 1.8 & 1.5 & 7.6 & 3.8 & 8.9 & 30.5 & 7.7 & 15.0 & 48.3  \\
\hline
{BigBench} & 4.1 & 2.0 & 5.6 & 6.4 & 6.2 & 25.0 & 21.7 & 18.5 & 33.3 \\
\hline
{MS-Prod} & -3.0 & 3.2 & 4.4 & 1.0 & 5.8 & 19.0 & 4.5 & 14.2 & 29.7\\\hline
\multicolumn{10}{p{3.6in}}{\rule{0pt}{2ex}{\bf G} is {\name}, {\bf T+T} is Tez + Tetris and {\bf T+CP} is Tez + CP. The improvements are relative to ${\tt Tez}$. Each group of columns reads out the gaps at the percentile in the label of that group.}\\
\end{tabular}
\end{scriptsize}
\ncaption{Improvements in job completion time across all the workloads \label{tab:eval_lat_datasets}}
\end{subfigure}
\vspace{3mm}\ncaption{Comparing completion time improvements of various schemes relative to ${\tt Tez}$.\label{fig:eval_compl_times}}
\end{figure*}
\else
\begin{figure}
\begin{subfigure}[b]{\textwidth}
\centering
\includegraphics[width=3in]{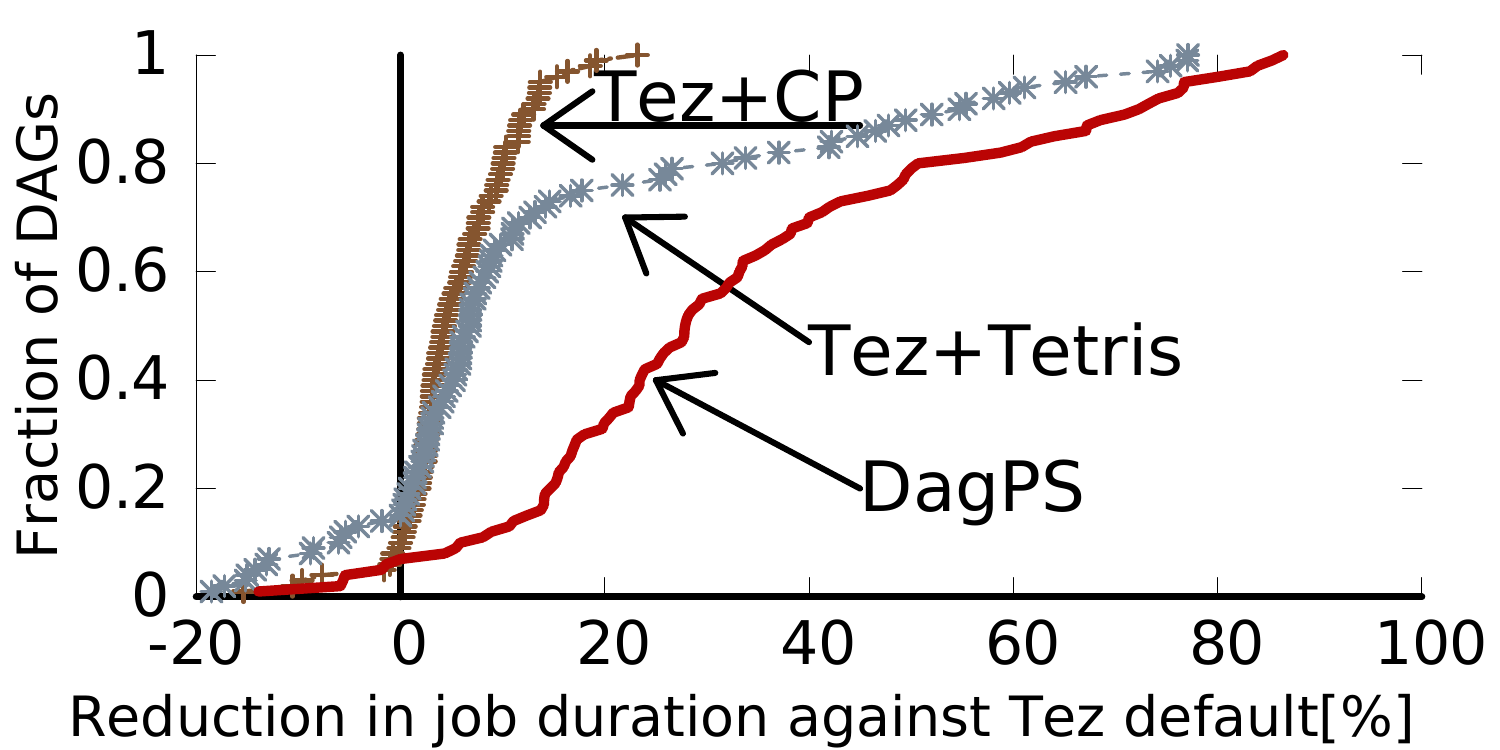}
\ncaption{CDF of gains for jobs on TPC-DS workload\label{fig:eval_cdf_tpcds}}
\end{subfigure}
\begin{subfigure}[b]{\textwidth}
\begin{small}
\vspace{.08in}
\begin{center}
\begin{tabular}[b]{l|rrr|rrr}
& \multicolumn{3}{c|}{$50^{th}$ percentile} & \multicolumn{3}{c}{$75^{th}$ percentile}\\
{Workload} & {\bf D} &{T+C} & {T+T}   & {\bf D} & {T+C} & {T+T}   \\
\hline
{TPC-DS} & {\bf 27.8} & 4.1 & 6.5 & {\bf 45.7} & 8.9 & 16.6  \\
\hline
{TPC-H} & {\bf 30.5} & 3.8 & 8.9  & {\bf 48.3} & 7.7 & 15.0 \\
\hline
{BigBench} & {\bf 25.0} & 6.4 & 6.2  & {\bf 33.3} & 21.7 & 18.5 \\
\hline
{E-Hive} & {\bf 19.0} & 1.0 & 5.8 & {\bf 29.7} & 4.5 & 14.2 \\\hline
\end{tabular}
\end{center}
\vspace{-.08in}
{
{\bf D} stands for {\name}. {T+C} and {T+T} denote ${\tt Tez+CP}$ and ${\tt Tez+Tetris}$ respectively~(see~\xref{subsec:eval_setup}).
The improvements are relative to ${\tt Tez}$. 

}
\end{small}
\ncaption{Improvements in job completion time across all the workloads \label{tab:eval_lat_datasets}}
\end{subfigure}
\ncaption{Comparing completion time improvements of various schemes relative to ${\tt Tez}$.\label{fig:eval_compl_times}}
\end{figure}
\fi
\iflongversion
\begin{figure*}[t!]
\begin{subfigure}[b]{0.19\textwidth}
\centering
\includegraphics*[width=1.41in, height=0.94in]{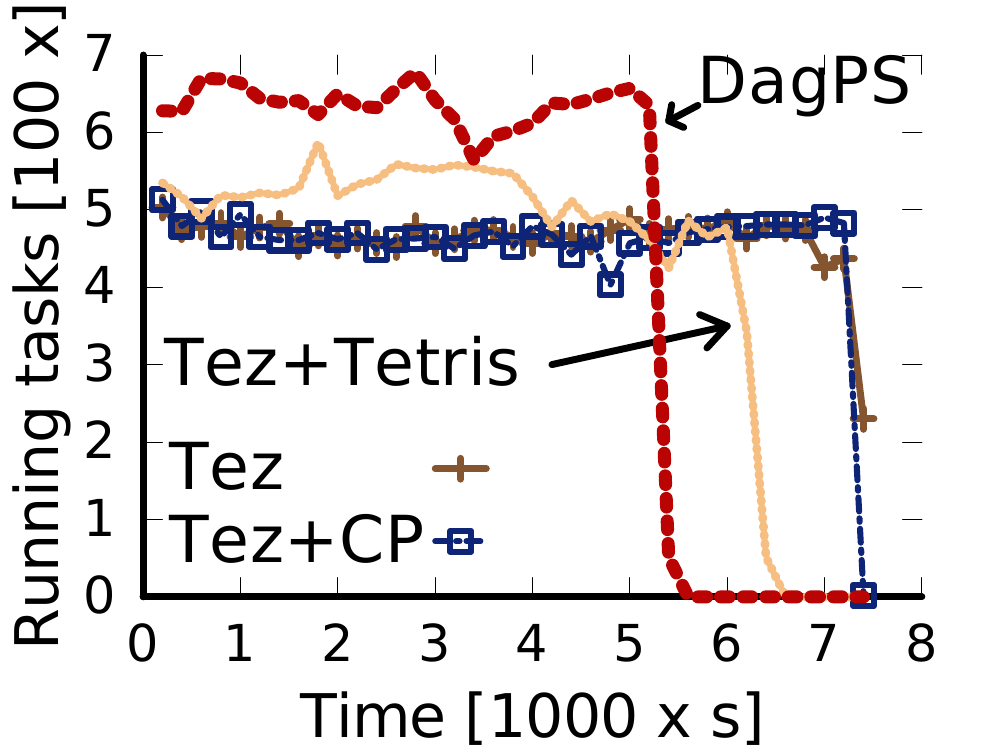}
\vspace{-5mm}\ncaption{Running tasks\label{fig:eval_runningtasks}}
\end{subfigure}
\begin{subfigure}[b]{0.19\textwidth}
\centering
\includegraphics*[width=1.41in, height=0.94in]{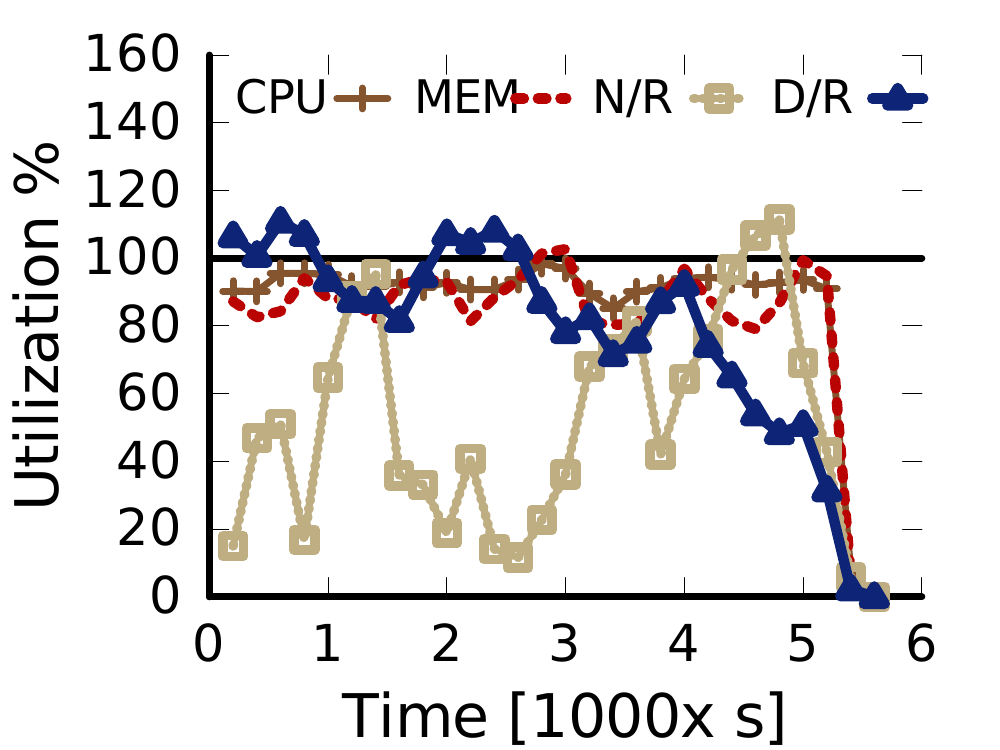}
\vspace{-5mm}\ncaption{{\name}\label{fig:eval_g_resources}}
\end{subfigure}
\begin{subfigure}[b]{0.19\textwidth}
\centering
\includegraphics*[width=1.41in, height=0.94in]{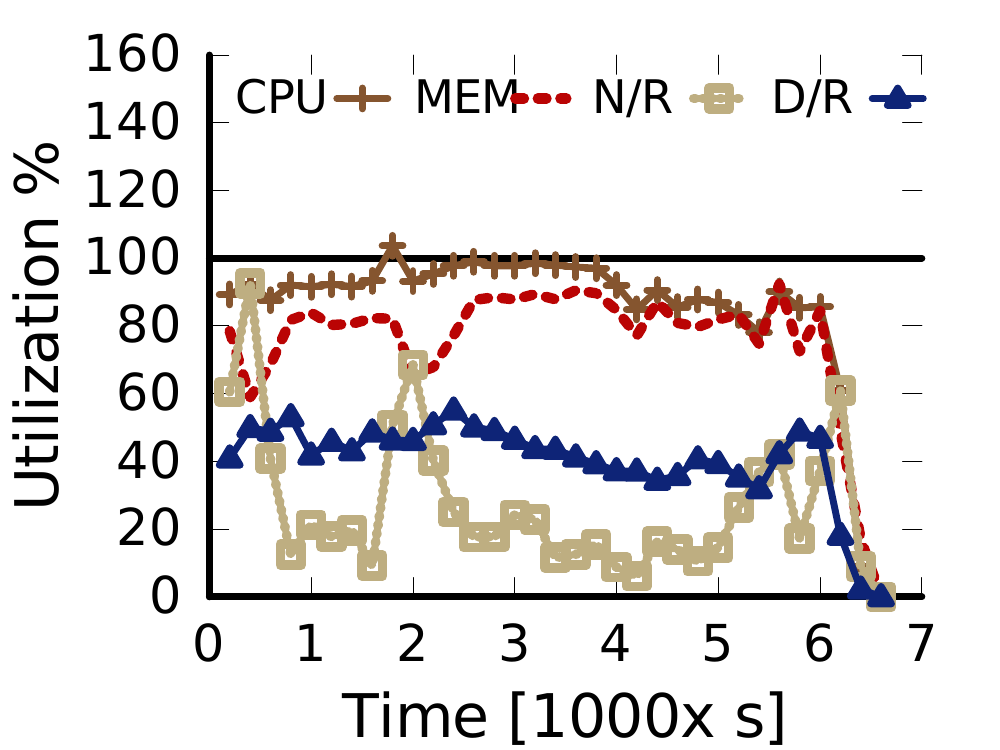}
\vspace{-5mm}\ncaption{${\tt Tez+Tetris}$\label{fig:eval_t_resources}}
\end{subfigure}
\begin{subfigure}[b]{0.19\textwidth}
\centering
\includegraphics*[width=1.41in, height=0.94in]{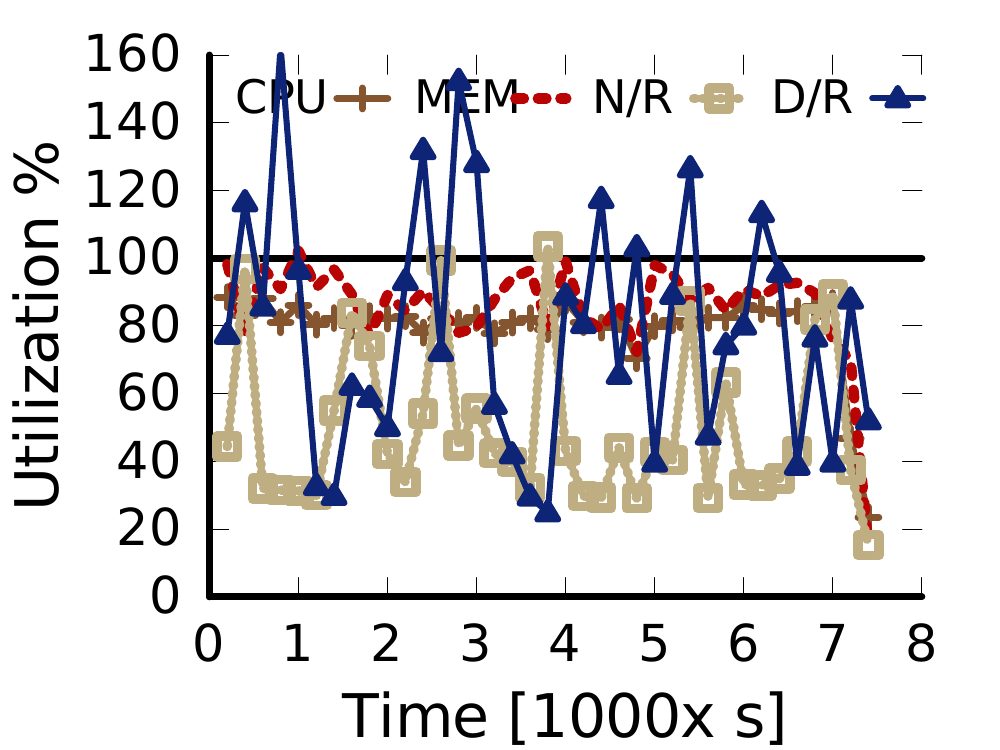}
\vspace{-5mm}\ncaption{${\tt Tez+CP}$\label{fig:eval_cp_resources}}
\end{subfigure}
\begin{subfigure}[b]{0.19\textwidth}
\centering
\includegraphics*[width=1.41in, height=0.94in]{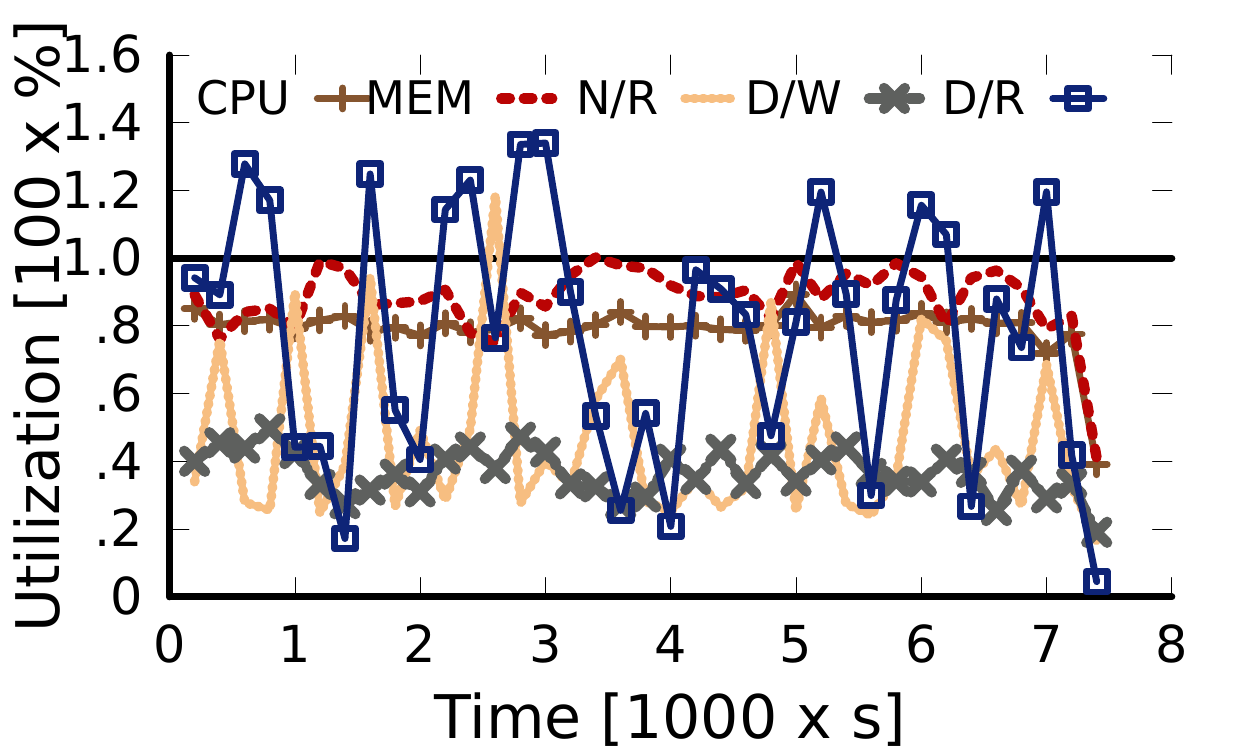}
\vspace{-5mm}\ncaption{${\tt Tez}$\label{fig:eval_default_resources}}
\end{subfigure}
\vspace{3mm}\ncaption{For a cluster run with $200$ jobs, a time lapse of how many tasks are running~(leftmost) and how
many resources are allocated by each scheme. $N/R$ represents the amount of network read, $D/R$ the disk read and $D/W$ the corresponding disk write.\label{fig:eval_why_gains}}
\end{figure*}
\else
\begin{figure}[t!]
\centering
\begin{subfigure}[b]{0.45\textwidth}
\centering
\includegraphics[width=1.6in]{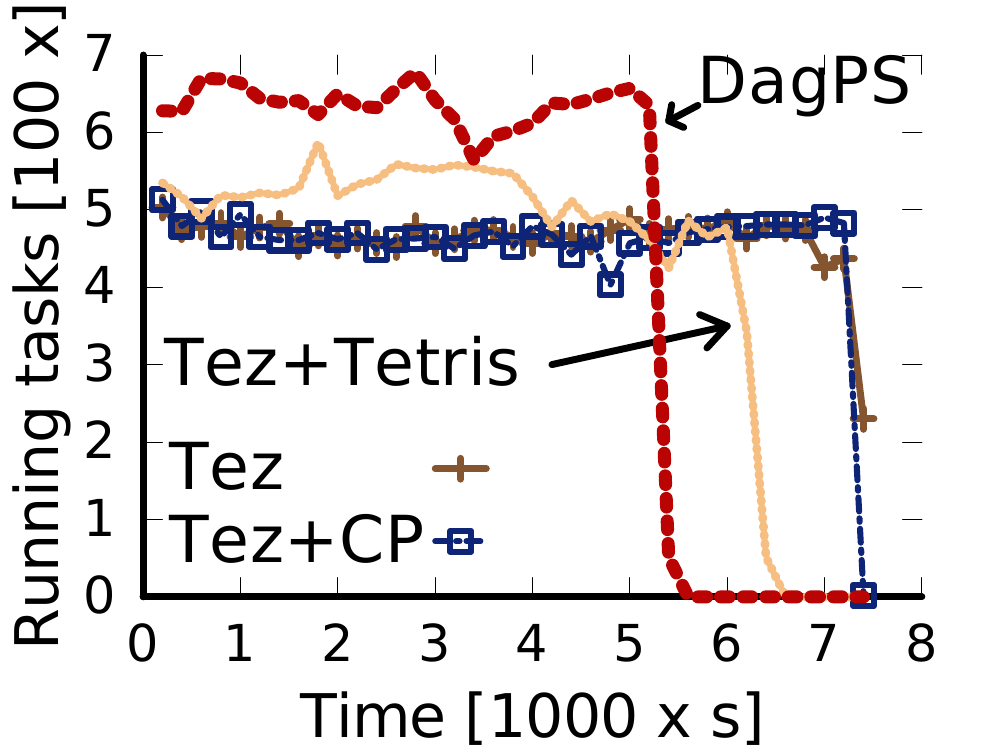}
\vspace{-5mm}
\ncaption{Running tasks\label{fig:eval_runningtasks}}
\end{subfigure}
\begin{subfigure}[b]{0.45\textwidth}
\centering
\includegraphics[width=1.6in]{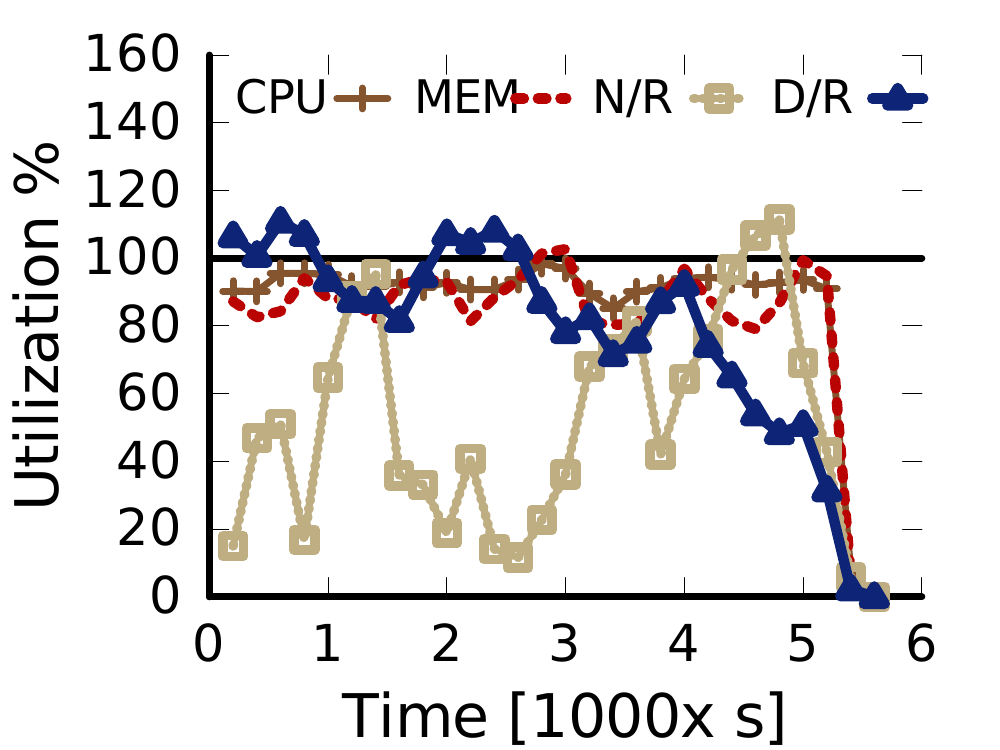}
\vspace{-5mm}
\ncaption{{\name}\label{fig:eval_g_resources}}
\end{subfigure}
\\[0.1in]
\begin{subfigure}[b]{0.45\textwidth}
\centering
\includegraphics[width=1.6in]{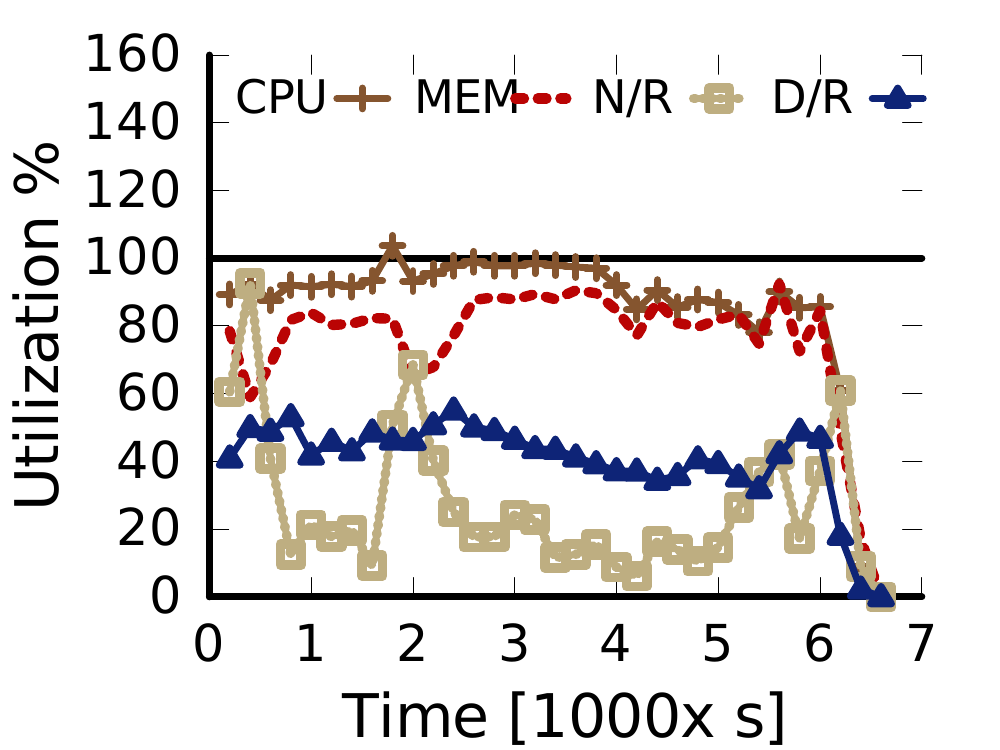}
\vspace{-5mm}
\ncaption{${\tt Tez+Tetris}$\label{fig:eval_t_resources}}
\end{subfigure}
\begin{subfigure}[b]{0.45\textwidth}
\centering
\includegraphics[width=1.6in]{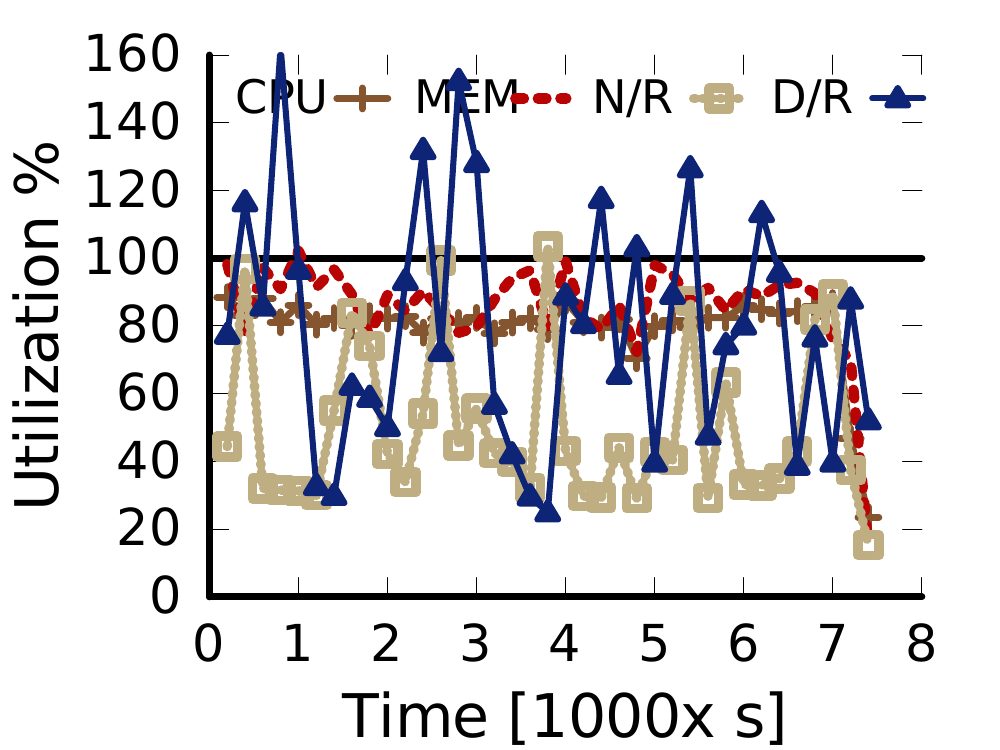}
\vspace{-5mm}
\ncaption{${\tt Tez+CP}$\label{fig:eval_cp_resources}}
\end{subfigure}
\vspace{3mm}\ncaption{For a cluster run with $200$ jobs, a time lapse of how many tasks are running~(leftmost) and how
many resources are allocated by each scheme. $N/R$ represents the amount of network read, $D/R$ the disk read and $D/W$ the corresponding disk write.\label{fig:eval_why_gains}}
\end{figure}
\fi

\subsection{How does {\namesec} do in experiments?}
\label{subsec:eval_cluster}
\subsubsection{Job Completion Time}
\label{subsubsec:jct}

Relative to ${\tt Tez}$, Figure~\ref{fig:eval_compl_times} shows that {\name}
improves half of the DAGs by $19$--$31$\% across various benchmarks.  One quarter of the DAGs
improve by $30$--$49$\%. We see occasional regressions. Up to $5$\% of
the jobs slow down with {\name}; the maximum slowdown is $16$\%. 
We found this to be due to two reasons.
(a) Noise from runtime artifacts such as stragglers and task failures and
(b) Imprecise profiles: in all of our experiments, we use a single profile (the average)
for all tasks in a stage but due to reasons such as data-skew, tasks in a stage
can have different resource needs and durations. 
The table in Fig.~\ref{fig:eval_compl_times} shows results for other benchmarks; 
we see that DAGs from E-Hive 
see the smallest improvement~($19$\% at median) because the DAGs here are mostly two stage
map-reduce jobs. The other benchmarks have more complex DAGs and hence receive sizable
gains. 

\cut{Fig.~\ref{fig:eval_compl_times} depicts the gap in completion time for jobs from cluster runs. Relative to the baseline, default Tez, {\name} offers substantial gains.
At least 50\% of the DAGs improve by $19-31$\% across various workloads.
And, at least 25\% of the DAGs improve by $30-49$\%. 
TPC-H workload has the most gains.
Gains for TPC-DS and BigBench workloads are about the same.
However, the {\tt Prod-B} dataset is dominated by queries  with few tasks and so {\name}'s gains, as well as the potential for gains, are limited on that workload.}

Relative to the alternatives, Figure~\ref{fig:eval_compl_times} shows that {\name} is $15$\% to $34$\% better. 
${\tt Tez+CP}$ achieves only marginal gains over ${\tt Tez}$, hinting 
that critical path scheduling
does not suffice. 
The exception is the BigBench dataset where about half the queries are dominated by work on the critical path.
${\tt Tez+Tetris}$ comes closest to 
{\name} because Tetris' packing logic reduces fragmentation.
But, the gap is still substantial, since Tetris ignores dependencies. In fact, we see that ${\tt Tez+Tetris}$ does not 
consistently beat ${\tt Tez+CP}$. 
Our takeaway 
is that considering both dependencies and packing can substantially improve DAG completion time. 

\cut{Across the board, {\name} is $15$\% to $34$\% better than the closest alternative.  Tez+CPSched is only a marginal improvement over Tez default hinting that preferentially scheduling tasks with more downstream work first does not suffice. The exception is the case of queries that have a dominant critical path. CPSched helps here by preventing work off the critical path from delaying tasks on the critical path. Such queries appear in the BigBench dataset and to a lesser extent in the {\tt Prod-B} dataset. On those datasets, Tez+CPSched does well for some DAGs (see 75th and 90th percentiles for BigBench and 90th for {\tt Prod-B}). Tez+Tetris comes closest to {\name} because Tetris' packing logic at the cluster-wide resource manager reduces fragmentation. But, it is not consistently better than Tez+CPSched. We conclude that for DAGs, considering both dependencies and packing, as {\name} does offers substantial improvement beyond considering them individually.}

Where do the gains come from? 
Figure~\ref{fig:eval_why_gains} offers more detail on an example experimental run. 
{\name} keeps more tasks running on the cluster and hence finishes faster~(Fig.~\ref{fig:eval_runningtasks}). 
The other schemes take over 20\% longer. 
To run more tasks, {\name} gains by reducing fragmentation and by overbooking fungible resources. 
Comparing Fig.~\ref{fig:eval_g_resources} with Figs.~\ref{fig:eval_t_resources}--\ref{fig:eval_cp_resources}, the average allocation of all resources is higher with {\name}. 
Occasionally, {\name} allocates over 100\% of the network and disk. 
${\tt Tez+Tetris}$, the closest alternative, has fewer tasks running at all times
because (a) it does not overbook~(all resource usages are below $100$\% in Fig.~\ref{fig:eval_t_resources})
and (b) it ignores dependencies and packs greedily leading to a worse packing of the entire DAG.
${\tt Tez+CP}$ is impacted negatively by two effects: (a) ignoring disk and network usage leads to 
arbitrary over-allocation~(the ``total'' resource usage is higher because, due to saturation, tasks hold on to allocations for longer)
and (b) due to fragmentation, many fewer tasks run on average. Together these lead to low task throughput and job delays. 

\begin{table}[t!]
{\small
\begin{tabular}{c|c|c|c}
{\bf Workload} & {\tt Tez+CP} & {\tt Tez+Tetris} & {\name}\\\hline
TPC-DS & $+2.1$\% & $+8.2$\% & $+30.9$\%\\
TPC-H & $+4.3$\% & $+9.6$\% & $+27.5$\%\\
\end{tabular}
}
\ncaption{Makespan, gap from ${\tt Tez}$.\label{tab:eval_makespan}}
\end{table}
\begin{table}[t!]
{\small
\scalebox{0.95}{
\begin{tabular}{c|c|c|ccc}
\multirow{2}{*}{\bf Workload} &
\multirow{2}{*}{\bf Scheme} & {\bf 2Q vs. 1Q} & \multicolumn{3}{c}{\bf Jain's fairness index} \\
&& {\bf Perf. Gap} & {$10$s} & {$60$s} & {$240$s}\\\hline
\multirow{3}{*}{TPC-DS}&{\bf Tez} & $-13$\%& $0.82$ & $0.86$ & $0.88$\\
&{\bf Tez+DRF} & $-12$\% & $0.85$ & $0.89$ & $0.90$\\
&{\bf Tez+Tetris} & $-10$\% & $0.77$ & $0.81$ & $0.92$\\
&{\bf {\name}} & $+2$\% & $0.72$ & $0.83$ & $0.89$\\
\end{tabular}}
}
\ncaption{Fairness: Shows the performance gap and Jain's fairness index when used with 2 queues~(even share) versus 1 queue. Here, a fairness score of $1$ indicates perfect fairness.\label{tab:eval_fairness}}
\end{table}

\subsubsection{Makespan}
\label{subsec:eval_other_metrics}
To evaluate makespan, we make one change to experiment setup-- all jobs arrive within the first few minutes. 
Everything else remains the same. 
Table~\ref{tab:eval_makespan} shows the gap in makespan for different
cases.  Due to careful
packing, {\name} sustains high cluster resource utilization, which in turn enables
individual jobs to finish quickly: makespan improves $31$\% relative to ${\tt Tez}$ and over $20$\% relative to 
alternatives.

\subsubsection{Fairness}
Can we improve performance while also being fair?
Intuitively, fairness may hurt performance 
since the task scheduling order needed for high performance (e.g.,
packability or dependencies) differs from the order that ensures
fairness.  To evaluate fairness, we make one change to the experiment
set up.  The jobs are evenly and randomly distributed among two queues
and the scheduler has to divide resources evenly.

Table~\ref{tab:eval_fairness} reports the gap in performance~(median job completion time) for each scheme when run with two queues 
vs.  one queue.
We see that
${\tt Tez}$, ${\tt Tez+DRF}$ and ${\tt Tez+Tetris}$ lose over $10$\% in performance relative to their one queue counterparts. 
The table shows that with two queues, {\name} has a small gain
(perhaps due to experimental noise).  Hence, relatively, {\name}
performs even better than the alternatives if given more
queues~($30$\% gap at one queue in Fig.~\ref{fig:eval_cdf_tpcds}
translates to a $40$\% gap at two queues). But why?
Table~\ref{tab:eval_fairness} also shows Jain's fairness index
computed over 10s, 60s and 240s time windows. We see that {\name} is
less fair at short timescales but is indistinguishable at larger time
windows.  This is because {\name} is able to bound
unfairness~(\xref{subsec:runtime}); it leverages some
short-term {\em slack} from precise fairness to make scheduling choices that improve performance.

\subsection{Comparing with alternatives}
\label{subsec:eval_compare_with_alt}
We use simulations to compare a much wider set of best-of-breed algorithms~(\xref{subsec:eval_setup}) on the much larger DAGs that ran in the production clusters. We mimic the actual dependencies, task durations 
and resource needs from the cluster.

Figure~\ref{fig:eval_compare_alt} compares the schedules constructed by 
{\name} with that from other algorithms. Table~\ref{tab:eval_complete_alt} 
reads out the gaps at various percentiles. We observe that {\name}'s gains at the end of schedule construction are about the same as those obtained at runtime~(Figure~\ref{fig:eval_compl_times}).  
This is interesting because the runtime component only
softly enforces the desired schedules from all the jobs running simultaneously 
in the cluster. It appears that any loss in performance from not adhering to the desired schedule
are made up by the gains from 
better packing and trading off some short-term unfairness.

\begin{figure}[t!]
\centering
\begin{subfigure}[b]{0.45\textwidth}
\centering
\includegraphics[width=1.55in]{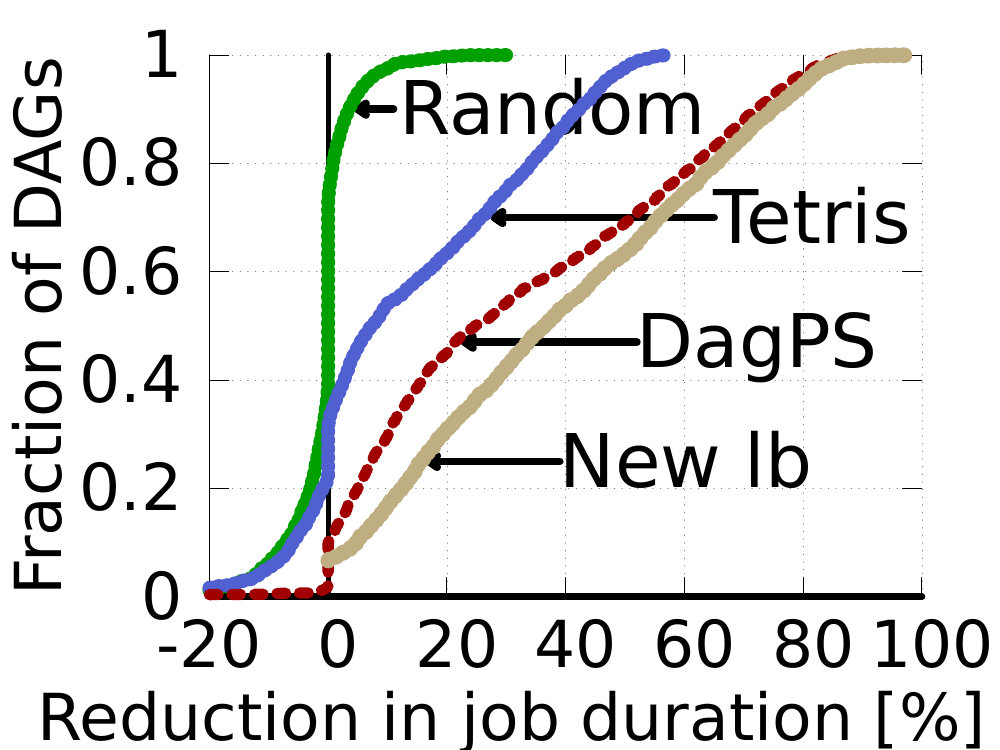}
\vspace{-5mm}\ncaption{{\name} vs. Baselines\label{fig:eval_alt}}
\end{subfigure}
\begin{subfigure}[b]{0.45\textwidth}
\centering
\includegraphics[width=1.55in]{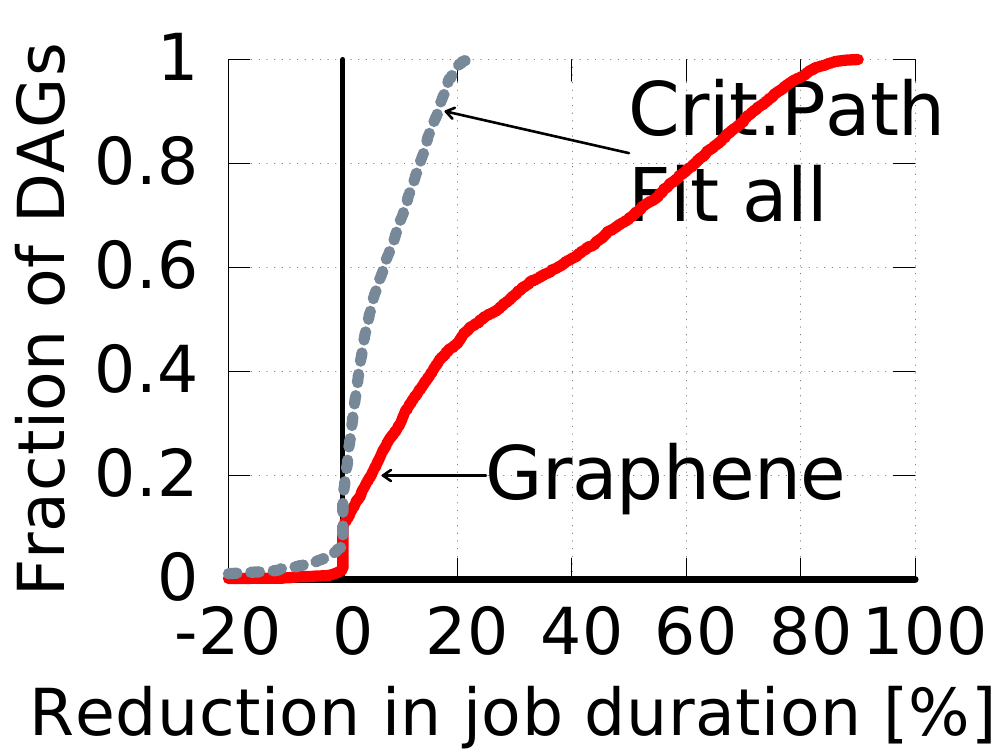}
\vspace{-5mm}\ncaption{{\name} vs. Alternatives\label{fig:eval_cp_g}}
\end{subfigure}
\vspace{4mm}\ncaption{Comparing {\name} with other schemes. We removed the lines from {\tt CG} and {\tt StripPart} because they hug $x=0$; see~\autoref{tab:eval_complete_alt}.} \label{fig:eval_compare_alt}
\end{figure}
\iflongversion
\begin{table}[t!]
\begin{small}
\scalebox{0.95}{
\begin{tabular}{l|l||ccccc}
\multicolumn{2}{l||}{} & {$25^{th}$} & {$50^{th}$} & {$75^{th}$} & {$90^{th}$}\\\hline\hline
\multicolumn{2}{l||}{\bf\name} & $ {\bf 7} $ & ${\bf 25}$ & ${\bf 57}$ & ${\bf 74}$ \\\hline
\multicolumn{2}{l||}{Random} & $ {-2} $ & ${0}$ & $ {1} $ & ${4}$ \\\hline
\multirow{3}{*}{Crit.Path} 
& Fit cpu/mem & $ {-2} $ & ${0}$ & $ {2} $ & ${1}$ \\
& Fit all & $ {1} $ & ${4}$ & $ {13} $ & ${16}$ \\
& \textcolor{red}{Overbooking} & $ \textcolor{red}{2} $ & $\textcolor{red}{9}$ & $ \textcolor{red}{24} $ & $\textcolor{red}{31}$ \\
\hline
\multirow{2}{*}{Tetris} & {Fit all} & $ {0} $ & ${7}$ & $ {29} $ & ${42}$ \\
& \textcolor{red}{Overbooking} & $ \textcolor{red}{0} $ & $\textcolor{red}{11}$ & $ \textcolor{red}{33} $ & $\textcolor{red}{49}$ \\\hline
\multirow{2}{*}{Strip Part.} & {Fit all} & $ {0} $ & ${1}$ & $ {12} $ & ${27}$ \\
& \textcolor{red}{Overbooking} & $ \textcolor{red}{0} $ & $\textcolor{red}{2}$ & $ \textcolor{red}{16} $ & $\textcolor{red}{33}$ \\\hline
\multirow{2}{*}{Coffman-Graham.} & {Fit all} & $ {0} $ & ${1}$ & $ {12} $ & ${26}$ \\
& Fit cpu/mem & $ {-2} $ & ${0}$ & $ {0} $ & ${2}$ \\
\end{tabular}}
\end{small}
\ncaption{Reading out the gaps from Fig.~\ref{fig:eval_compare_alt}; comparing {\name} vs. Alternatives. The improvements are relative to {\tt BreadthFirst} - standard approach used in 
${\tt Tez}$.}
\label{tab:eval_complete_alt}
\end{table}
\else
\begin{table}[t!]
\begin{small}
\scalebox{0.95}{
\begin{tabular}{l|l||ccccc}
\multicolumn{2}{l||}{} & {$25^{th}$} & {$50^{th}$} & {$75^{th}$} & {$90^{th}$}\\\hline\hline
\multicolumn{2}{l||}{\bf\name} & $ {\bf 7} $ & ${\bf 25}$ & ${\bf 57}$ & ${\bf 74}$ \\\hline
\multicolumn{2}{l||}{Random} & $ {-2} $ & ${0}$ & $ {1} $ & ${4}$ \\\hline
\multirow{2}{*}{Crit.Path} 
& Fit cpu/mem & $ {-2} $ & ${0}$ & $ {2} $ & ${1}$ \\
& Fit all & $ {1} $ & ${4}$ & $ {13} $ & ${16}$ \\
\hline
\multirow{1}{*}{Tetris} & {Fit all} & $ {0} $ & ${7}$ & $ {29} $ & ${42}$ \\
\hline
\multirow{1}{*}{Strip Part.} & {Fit all} & $ {0} $ & ${1}$ & $ {12} $ & ${27}$ \\
\hline
\multirow{2}{*}{Coffman-Graham.} & {Fit all} & $ {0} $ & ${1}$ & $ {12} $ & ${26}$ \\
& Fit cpu/mem & $ {-2} $ & ${0}$ & $ {0} $ & ${2}$ \\\hline
\end{tabular}}
\end{small}
\vspace{2mm}
\ncaption{Reading out the gaps from Figure~\ref{fig:eval_compare_alt}; comparing {\name} vs. Alternatives. Each entry is the improvement relative to ${\tt BFS}$.}
\label{tab:eval_complete_alt}
\end{table}
\fi

Second, {\name}'s gains are considerable compared to the alternatives. 
${\tt CP}$ and ${\tt Tetris}$ are the closest. The reason is that {\name} looks
at the entire DAG and places the troublesome tasks first, leading to a more compact schedule
overall.

Third, when tasks have unit durations and nicely shaped demands, ${\tt
  CG}$~(Coffman-Graham~\cite{coffman-graham}) is at most 2 times optimal.  
However, it does not
perform well on production DAGs that have diverse
demands for resources and varying durations.  Some recent extensions
to ${\tt CG}$ handle heterogeneity but ignore fragmentation
issues when resources are divided across many machines~\cite{kwok1999static}.

Fourth, ${\tt StripPart}$~\cite{StripPart} is the best known algorithm that 
combines resource packing and task dependencies. It yields an $O(\log n)$-approx 
ratio on a DAG with $n$ tasks~\cite{StripPart}. The key idea is to partition tasks 
into {\em levels} such that all dependencies go across levels. The primary drawback with ${\tt StripPart}$
is that it prevents overlapping independent tasks that happen to be in different levels. 
A secondary drawback
is that the recommended packers~(e.g.,~\cite{reverse-fit}) do not support
multiple resources and vector packing. 
We see that in practice ${\tt StripPart}$ under-performs 
the simpler heuristics.

\cut{
To tease apart the gains from parts of {\name}, we improve 
existing algorithms by adding some ideas from {\name}. 
Fig.~\ref{fig:eval_cp_g} and Table~\ref{tab:eval_complete_alt} (in red) 
show that substantial gap remains in spite of blessing existing alternatives
with overbooking.
}


\subsection{How close is {\namesec} to Optimal?}
\label{subsec:eval_lb}
\begin{figure}[t!]
\centering
\includegraphics[width=1.5in]{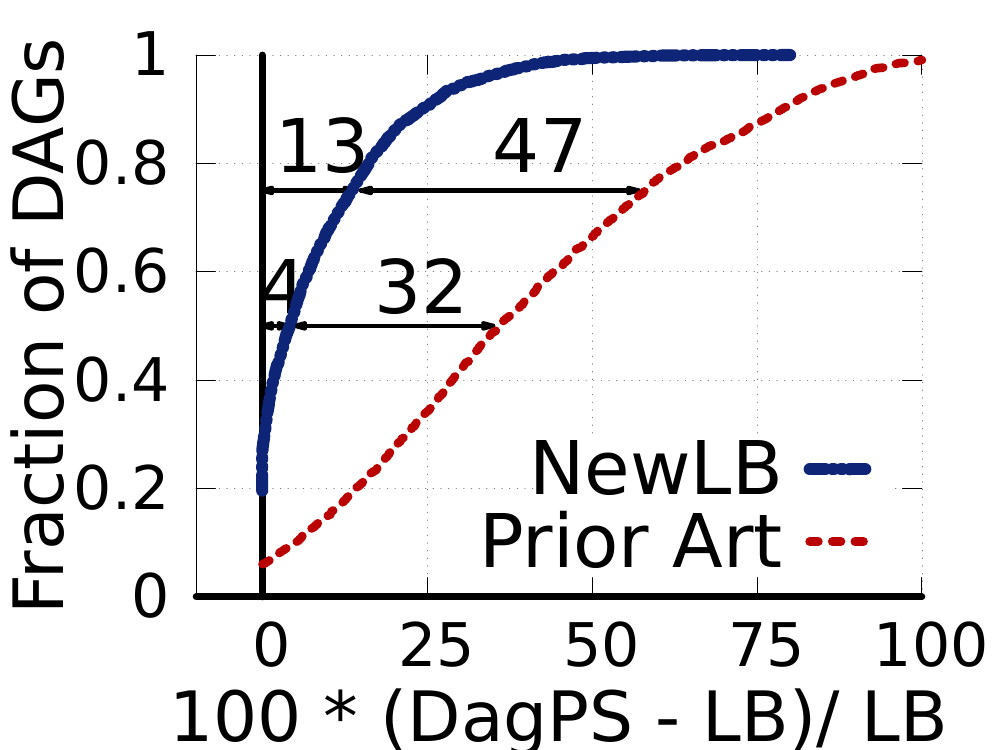}
\vspace{2mm}\ncaption{Comparing {\name} with lower bounds.\label{fig:g_vs_lb}}
\end{figure}
Figure~\ref{fig:g_vs_lb} compares {\name} with ${\tt NewLB}$ and the best previous lower bound~$\max({\tt CPLen}, {\tt TWork})$ (see~\xref{sec:betterlb}). Since the optimal schedule is no shorter than the lower bound, the figure shows that {\name} is optimal for about $40$\% of DAGs. For half~(three quarters) of the DAGs, {\name} is within $4$\%~($13$\%) of the new lower bound. A gap still remains: for the worst $10$\% of DAGs, {\name} takes $25$\% longer. Manually examining these DAGs shows that ${\tt NewLB}$ is loose for most of them. However, the figure also shows that the ${\tt NewLB}$ improves upon previous lower bounds by almost $30$\% for most of the DAGs. We conclude that while more work remains towards a good lower bound, ${\tt NewLB}$ suffices to argue that {\name} is close to optimal for most of the production DAGs.

\begin{figure}[t!]
	\centering
	\begin{subfigure}[b]{0.45\textwidth}
		\centering
		\includegraphics[width=1.5in]{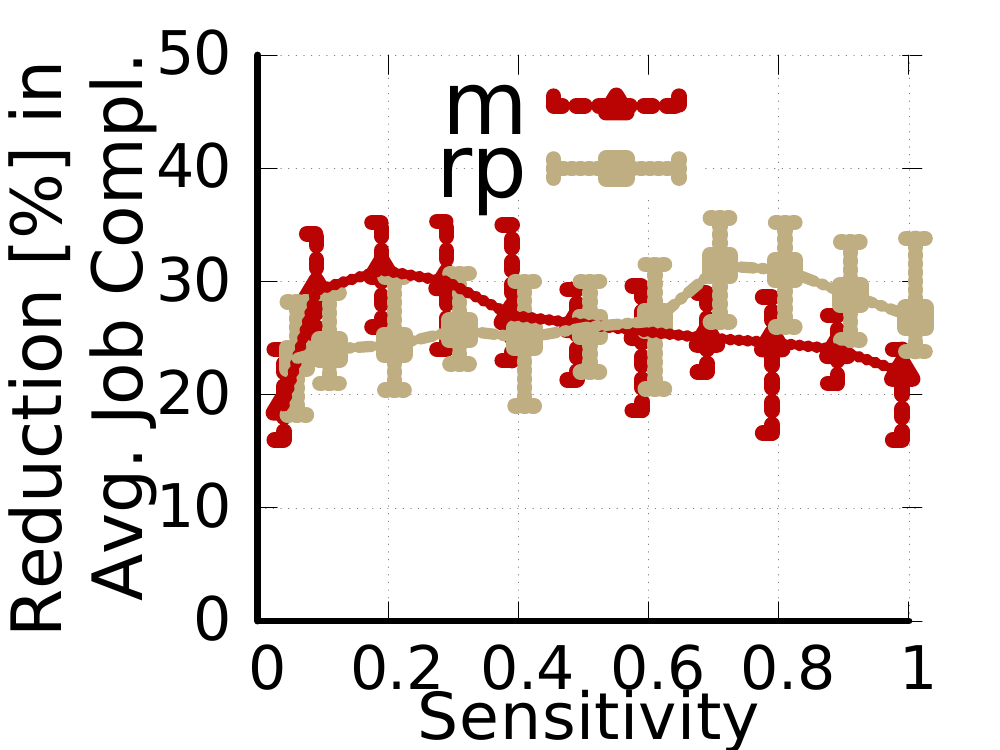}
		\vspace{-5mm}\ncaption{Job Duration\label{fig:sensitivity_job_dur}}
	\end{subfigure}
	\begin{subfigure}[b]{0.45\textwidth}
		\centering
		\includegraphics[width=1.5in]{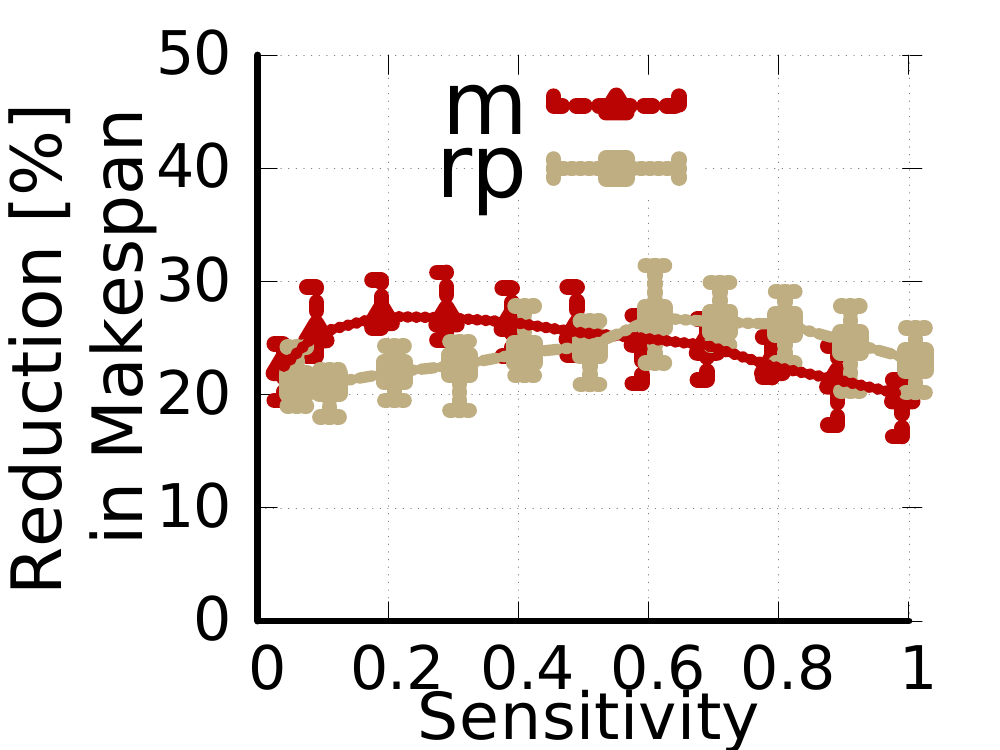}
		\vspace{-5mm}\ncaption{Makespan\label{fig:sensitivity_makespan}}
	\end{subfigure}
	\vspace{3mm}\ncaption{{\name} - sensitivity analysis.} \label{fig:eval_sensitive_analysis}
\end{figure}

\subsection{Sensitivity Analysis}
\label{subsec:eval_sa}
We evaluate {\name}'s sensitivity to parameter choices.
\cut{
Recall that in order to pick tasks for assignment, {\name} computes a score similar to Tetris. 
In addition it incorporates two new factors: a priority factor applied as a multiplication to the total score, 
and an overfit factor which change the magnitude of the alignment score. 
None of the new variables we introduced in the score computation adds additional sensitivity. In the following, we will do a parameter sweep for weighting alignment factor and remote penalty defined in Tetris, but considering the additional constraints introduced by {\name}.
}

\noindent
\textbf{Packing vs. Shortest Remaining Processing
  Time (${\tt srpt}$):} Recall that we combine packing score and ${\tt srpt}$ using a
weighted sum with $\eta$~(first box in Figure~\ref{fig:ps_sched}). 
Let $\eta$ be $m$ times the average over the two expressions that it combines.
Here, we evaluate the sensitivity of the choice of $m$.
Figure~\ref{fig:eval_sensitive_analysis} shows the reduction in average
job completion time~(on left) and makespan~(on right) for different
values of $m$.  Values of $m \in [0.1, 0.3]$ have the most gains.
Lower values lead to worse average job completion time because the
effect of ${\tt srpt}$ reduces.  On the other hand, larger
values lead to moderately worse makespan. Hence, we recommend
$m=0.2$.

\vskip .05in
\noindent
\textbf{Remote Penalty:} {\name} uses a remote penalty ${\tt rp}$ to
prefer local placement.  Our analysis shows that both job completion
time and makespan improve the most when ${\tt rp}$ is between $15\%$
and $30\%$ (Fig.~\ref{fig:eval_sensitive_analysis}). Since ${\tt rp}$ is a multiplicative penalty, lower values of ${\tt  rp}$ cause the scheduler to miss (non-local) scheduling
opportunities whereas higher ${\tt rp}$ can over-use remote resources
on the origin servers. We use ${\tt rp}=0.8$.

\vskip .05in
\noindent
\textbf{Cluster Load:} 
We vary cluster load by reducing the number of available servers without
changing the workload. 
Figure~\ref{fig:eval_cluster_load} shows the job completion times 
and makespan for a query set derived from TPC-DS.
We see that both {\name} and the alternatives offer more gains 
at higher loads. This is to be expected since the need for careful scheduling and packing
increases when resources are scarce. Gains due to {\name} increase by $+10$\%
at $2\times$ load and by $+15$\% at $6\times$ load. However, the gap between
{\name} and the alternatives remains similar across load levels.
\cut{
In this experiment, we examine the performance of {\name} under different cluster load. 
We vary the load by changing the number of servers but maintaining the same workload size. 
Fig.~\ref{fig:eval_cluster_load} shows our findings for a TPC-DS workload. As expected, the gains improve with the cluster load. {\name}'s improvement in job completion time increase by more than 10\% at $2\times$ the current load and goes over 15\% at more than $6\times$. At $4\times$ the current load, {\name} improves makespan by more than 10\%. However, gains don't go over 40\%. This can be explained due to the fact that the cluster load is too constrained by the setup. Fig.~\ref{fig:eval_cluster_load} also shows the performance of alternative schemes under various cluster load. While both {\tt Tez+Tetris} and {\tt Tez+CPSched} improve performance for higher load, the gap w.r.t. to {\name} remains at least the same as in the case of the original cluster load. The reason is that for higher load more opportunities for scheduling exists and the alternative schemes are more prone to choose bad schedules.
}

\begin{figure}[t!]
	\centering
	\begin{subfigure}[b]{0.45\textwidth}
		\centering
		\includegraphics[width=1.45in]{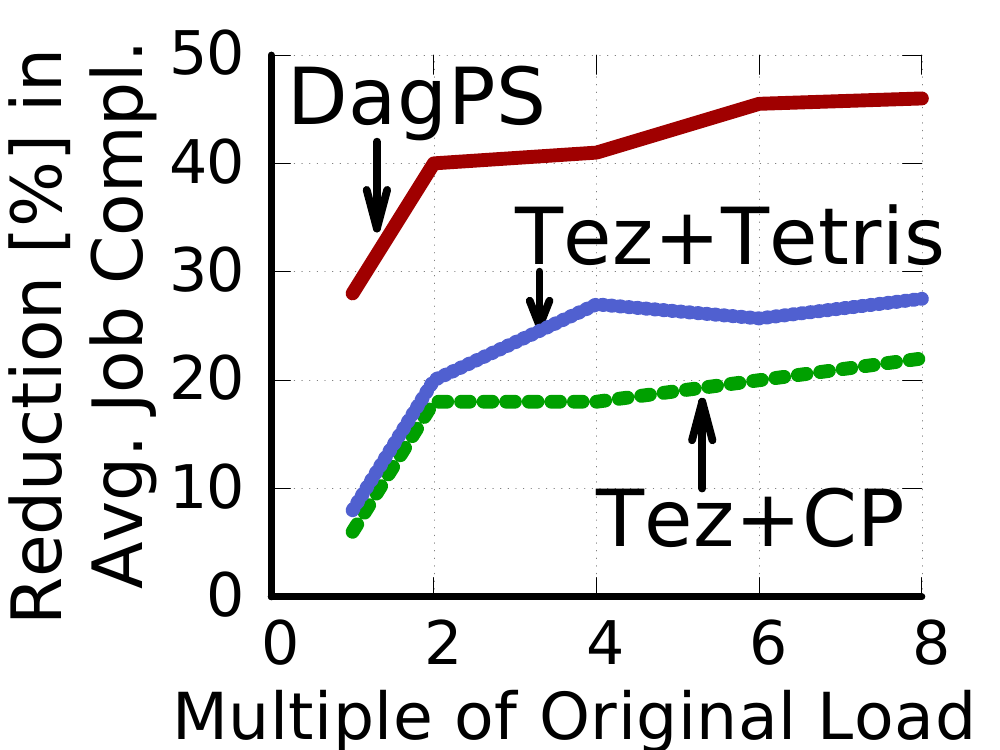}
		\vspace{-5mm}\ncaption{Job Duration\label{fig:cluster_load_job_dur}}
	\end{subfigure}
	\begin{subfigure}[b]{0.45\textwidth}
		\centering
		\includegraphics[width=1.45in]{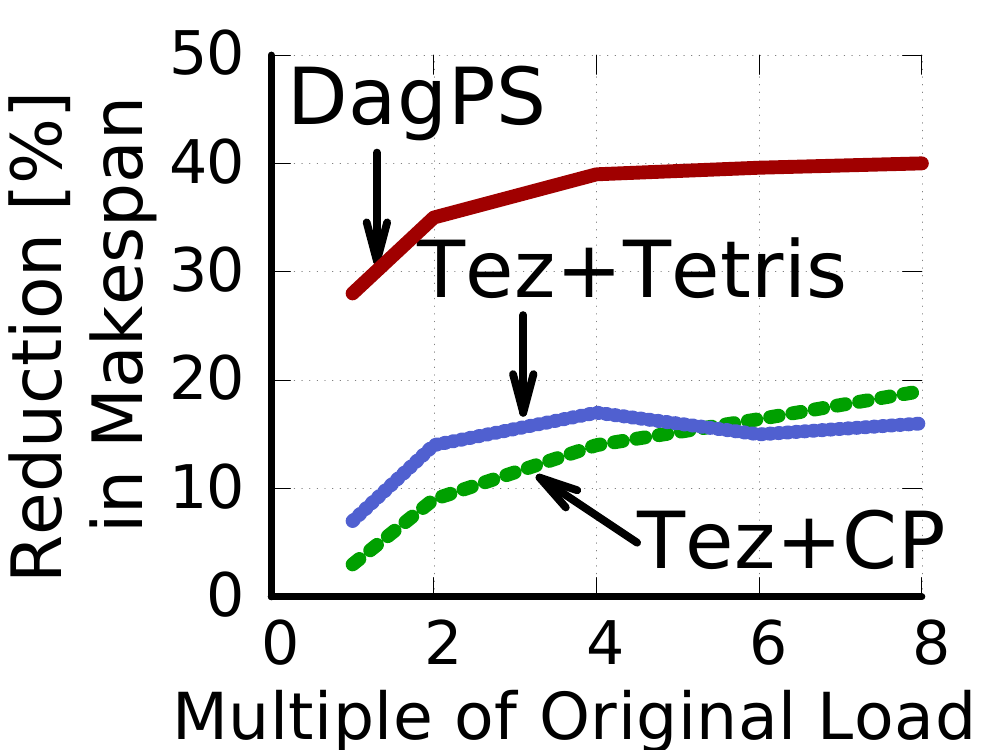}
		\vspace{-5mm}\ncaption{Makespan\label{fig:cluster_load_makespan}}
	\end{subfigure}
	\vspace{2.5mm}\ncaption{{\name}'s gains increase with cluster load.} \label{fig:eval_cluster_load}
\end{figure}

\section{Applying {\name} to other domains}
\label{subsec:other-domains}
We evaluate {\name}'s effectiveness in scheduling the DAGs that arise
in distributed compilation systems~\cite{bazel,cloudmake} and
request-response workflows for Internet services~\cite{kwiken}.

Distributed build systems speed up the compilation of large code
bases~\cite{bazel,cloudmake}. Each build is a DAG with dependencies
between the various tasks~(compilation, linking, test, code
analysis). The tasks have different runtimes and have different
resource profiles. Figure~\ref{fig:cloud_build} shows that {\name} is
$20$\% faster than ${\tt Tetris}$ and $30$\% faster than ${\tt CP}$
when scheduling the build DAGs from a production distributed build
system. Each bar shows the median gain for DAGs of a certain size and
the error bars are quartiles. The gains hold across DAG sizes/ types.

We also examine the DAGs that arise in datacenter-side
request-response workflows for Internet-services~\cite{kwiken}. For
instance, a search query translates into a workflow of dependent RPCs
at the datacenter~(e.g., spell check before index lookup, video and
image lookup in parallel).  The RPCs use different resources, have
different runtimes and often execute on the same server
pool~\cite{kwiken}.  Over several workflows from a production service,
Figure~\ref{fig:kwiken} shows that {\name} improves upon alternatives
by about $24$\%.

These early results, though preliminary, are encouraging and
demonstrate the generality of our work.  

\begin{figure}[t!]
        \begin{subfigure}[b]{0.49\textwidth}
                \centering
                \includegraphics*[width=1.5in]{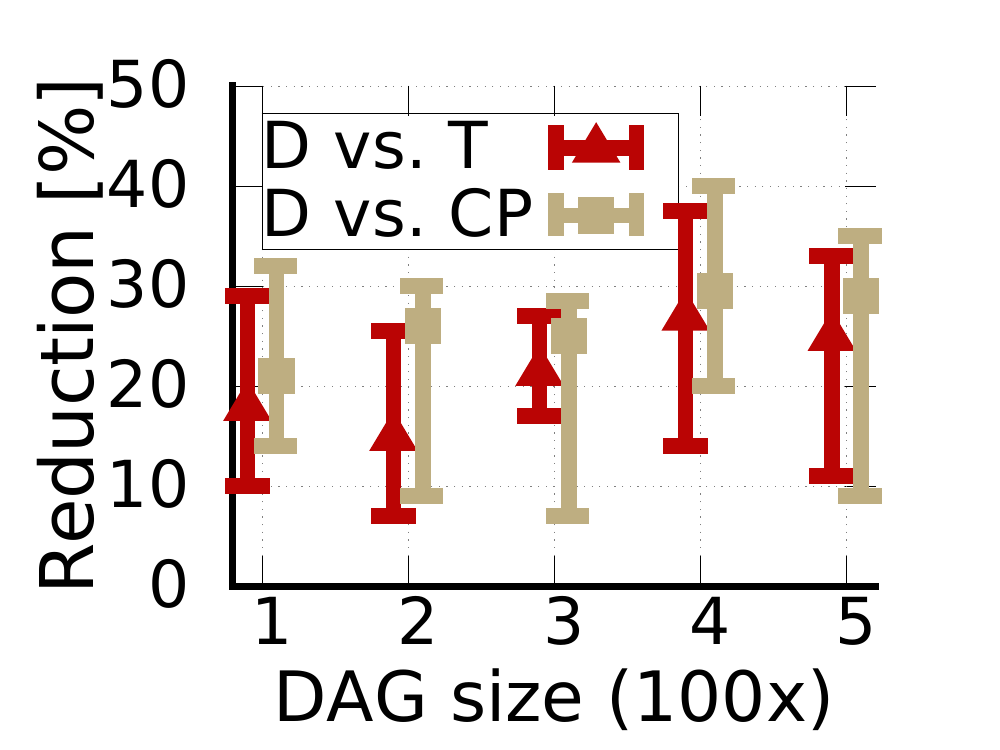}
                \vspace{-1mm}
                \ncaption{Distributed Build Systems: \\Compilation time\label{fig:cloud_build}}
        \end{subfigure}
        \begin{subfigure}[b]{0.49\textwidth}
                \centering
                \includegraphics*[width=1.5in]{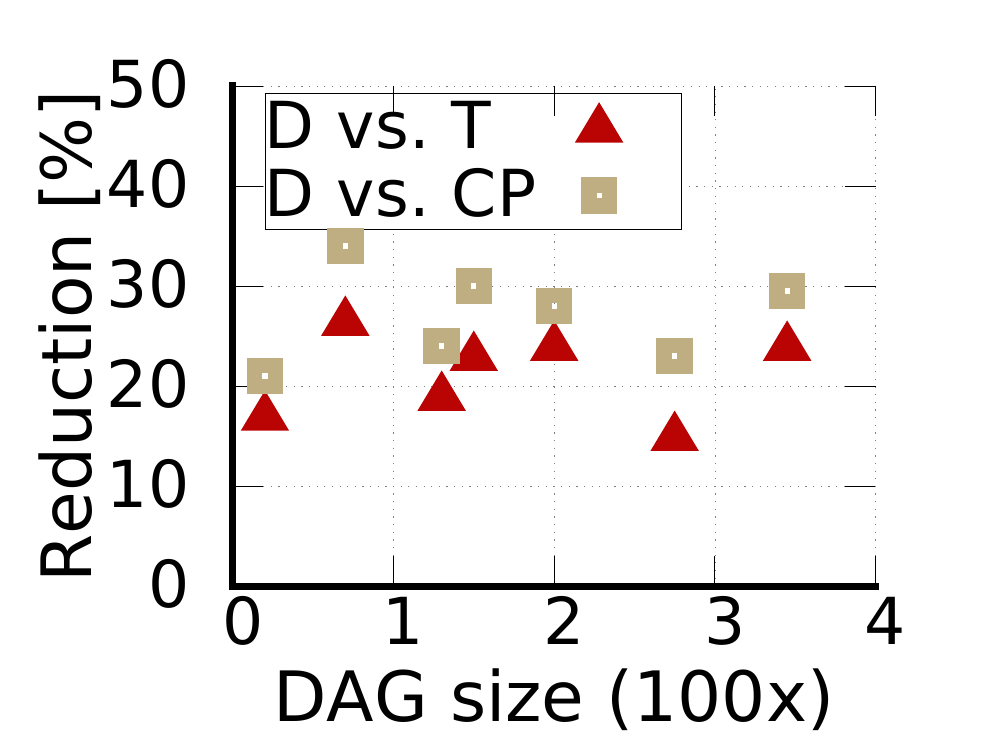}
                \vspace{-1mm}
                \ncaption{Request-response workflows: \\Query latency\label{fig:kwiken}}
        \end{subfigure}
        \ncaption{Comparing {\name}~(${\tt D}$) with Tetris~(${\tt T}$) and Critical path scheduling~(${\tt CP}$) on DAGs from two other domains.} \label{fig:eval_compare_other_domains}
\end{figure}

\section{Related Work} 
\label{sec:related}
To structure the discussion, we ask four questions:
(Q1) does a scheme consider both packing and dependencies, 
(Q2) does it make realistic assumptions,
(Q3) is it practical to implement in cluster schedulers and,
(Q4) does it consider multiple objectives such as fairness?
To the best of our knowledge, {\name} is unique in positively answering these four questions. 

\vskip .03in
\noindent
${\tt Q1: NO.}$ Substantial prior work ignores dependencies but packs tasks with varying demands for multiple resources~\cite{ChekuriK04,ImKMP15,tetris,reverse-fit,AzarCG13}. The best results are when the demand vectors are {\em small}~\cite{psv}. Other work considers dependencies but assumes homogeneous demands~\cite{graham,coffman-graham}.
A recent multi-resource packing scheme, Tetris~\cite{tetris}, succeeds on the three other questions but does not handle dependencies. 
Hence, we saw in~\xref{sec:eval} that it performs poorly when scheduling DAGs. Further, Tetris has poor worse-case performance~(up to $2d$ times off, see Figure~\ref{fig:tetrisdtimesworse}) and can be arbitrarily unfair.

\vskip .05in
\noindent
${\tt Q1: YES, Q2: NO.}$ The packing+dependencies problem has been considered at length under the keyword {\em job-shop scheduling}~\cite{lmr,Czumaj:2000,Goldberg:1997,Shmoys:1994}. Most results assume that jobs are known apriori~(i.e., the offline case). See~\cite{kwok1999static} for a survey. For the online case~(the version considered here), no algorithms with bounded competitive ratios are known~\cite{Mastrolilli:2008, MastrolilliS:2009}. Some other notable work assumes only two resources~\cite{banerjee1990approximate}, applies only for a chain but not a general DAG~\cite{tpkelly1} or assumes one cluster-wide resource pool~\cite{lepere2002approximation}. 

\vskip .03in
\noindent
${\tt Q3: NO.}$ All of the schemes listed above consider one DAG at a time and are not easily adaptable to the online case when multiple DAGs share a cluster. Work on related problems such as VM allocation~\cite{Chowdhury2009a} also considers multi-resource packing. However, cluster schedulers have to support roughly two to three orders of magnitude higher rate of allocation~(tasks are more numerous than VMs). 

\vskip .03in
\noindent
${\tt Q3: YES, Q1: NO.}$ Several notable works in cluster scheduling exist such as 
Quincy~\cite{quincy}, Omega~\cite{omega}, Borg~\cite{borg}, Kubernetes~\cite{kubernetes}
and Autopilot~\cite{Isard2007}. None of these combine multi-resource packing with DAG-awareness.
Many do neither. Job managers such as Tez~\cite{tez} and Dryad~\cite{dryad} use simple heuristics such as breadth-first scheduling which perform quite poorly in our experiments.

\vskip .03in
\noindent
${\tt Q4: NO.}$  There has been much recent work on novel fairness schemes to incorporate multiple resources~\cite{drf} and be work-conserving~\cite{hug}. Several applications arise especially in scheduling coflows~\cite{varys,aalo}. We note that these fairness schemes neither pack nor are DAG-aware. {\name} can incorporate these fairness methods as one of the multiple objectives and trades off bounded unfairness for performance.

\section{Concluding Remarks}
\label{sec:conc}
DAGs are indeed a common scheduling abstraction.  However, we found
that existing algorithms make several key assumptions that do not hold
in practical settings.  Our solution, {\name} is an efficient online
solution that scales to large clusters.  We experimentally validated
that it substantially improves the scheduling of DAGs in both
synthetic and emulated production traces. The core contributions are
threefold: (1) constructing a good schedule by placing tasks
out-of-order on to a virtual resource-time space, (2) an online
heuristic that softly enforces the desired schedules and helps with
other concerns such as packing and fairness, and (3) an improved lower
bound that lets us show that our heuristics are close to optimal. Much
of these innovations use the fact that job DAGs consist of groups of
tasks~(in each stage) that have similar durations, resource needs and
dependencies. We intend to contribute our {\name} implementation to
Apache YARN/Tez projects.  As future work, we are considering applying
these DAG scheduling ideas to related domains, most notably scheduling
the coflows with dependencies that arise in geo-distributed
analytics~\cite{dje,iridium,geode,swag}.






\begin{appendices}
\iflongversion
\else
\section{Value of DAG awareness}
\label{subsec:dag_awareness}
\noindent
{\it Proof of Lemma~\ref{lem:dag_awareness}: }

\qed
\fi

\iflongversion
\else
\section{Worst-case DAG Examples}
\label{sec:worst-case}

\fi

\iflongversion
\section{Minimizing Makespan In The Absence of Task Dependencies}
\label{sec:simple_greedy}
In this section we describe the structure of optimal makespan algorithm in the multi-resource setting for leontief utilities when there are no task dependencies. We first make the problem statement formal. We are given a set $J$ of $N$ tasks, and a set of resources \{1, 2 \ldots D \}.  We assume, without loss of generality, that each resource is available to an extent of 1. Each task $j$ is characterized by following parameters: 1) Processing length (execution time) $p_j$. 2) Resource requirement vector $A_j = ( a_{1j}, a_{2j}, \ldots a_{Dj} )$, which is a  $D$-dimensional vector with each entry at most 1. Each coordinate $a_{dj}$ indicates the  fraction of resource $d \in [D]$ that the task $j$ demands. We assume that scheduler cannot allocate more than $a_{jd}$ fraction of the resource $d$ to the task $j$. We shall assume that all tasks are available for processing at the same time. The goal of a scheduler is to allocate resources to tasks at each time instant.  A feasible allocation is one where no resource is overallocated. 

\medskip

\noindent \textbf{Rate of a job:} \\\\
Let $x_{jd}$ denote the fraction of the demand ($a_{jd}$) of resource $d$ that is allocated to task $j$.   A task $j$ executes at a {\em rate} 
\begin{equation}
\label{e:rate}
x_{j} = \min_{j \in d} \{x_{jd}\}
\end{equation}

The notation $j \in d$ denotes that task $j$ demands resource $d$.  Given an allocation of resources to tasks, the equation (\ref{e:rate}) implies that the tasks execute at a rate which is equal to the minimum fraction of resource that is allocated. 
In other words, the rates of tasks are {\em leontief}. It is also worth noting that the rate of a task completely determines the amount of each resource that it gets. That is, if the rate of a job at time $t$ is $x_{j}$ then the job $j$ receives $x_{j}$ fraction of each resource it demands. 
The rate of a task is simply the total amount of work completed per unit time. We assume, by normalizing, that a task can complete only 1 unit of work per unit time, and it needs to complete $p_j$ units of work before it departs from the system.

\subsection{Optimal Algorithm}
We now state the rate allocation by an optimal algorithm for minimizing makespan. Define,
$$
P := \max_{d} \sum_{j} a_{jd} \cdot p_j.
$$

Then, an optimal algorithm allocates rates to tasks as follows:

$$
x_j := \frac{p_j}{P}
$$

We now prove that this allocation of rate is optimal. First observe that $P$ is a lower bound on the optimal makespan. (We will also assume that $P \geq \max_j p_j$, otherwise the problem is trivial.) 
Hence it remains to show that rates are feasible. 

\begin{lemma}
	Rates allocated by the optimal solution are feasible. 
\end{lemma}

\begin{proof}
	Fix a resource $i \in [D]$. We show $i$ is allocated to an extent of at most $1$. To see this, the total allocation of resource $i$ is given by,
	$$
		\sum_{j} a_{jd} \cdot \frac{p_j}{P} = \sum_{j} a_{jd} \cdot p_j \cdot \frac{P}{P} = 1.
	$$
\end{proof}

\begin{lemma}
	Every task completes by time $P$.
\end{lemma}
\begin{proof}
	Since task $j$ executes at the rate of $\frac{p_j}{P}$, it completes by time $p_j \cdot \frac{P}{p_j} = P$.
\end{proof}

The above two lemmas complete the proof. Observe that for an input instance, $P$ is a constant and hence in the optimal schedule a task executes at a rate proportional to its length.

\else
\fi

\iflongversion
\section{Proof of deadlock freedom and completeness}
\label{sec:app_deadlock_free}
\noindent
{\bf Proof of lemma 2.}
This claim follows from combining these observations.
(1) {\small\sf T} is a closure over troublesome tasks.
(2) When placing the tasks in  {\small\sf T}, both forward and backward placements respect all dependencies
between tasks in {\small\sf T}.
(3) The allowed placements for subsets {\small\sf P, O, C} are such that
forward placement is allowed for a subset if and only if all ancestors of tasks
in that subset have been placed already. For example, {\small\sf O} in {\small\sf TPOC} and 
{\small\sf C} in every allowed order.
(4) By similar reasoning as above, backward placement is allowed if and only if 
all descendants of tasks in that subset have been placed already.
\eat{
It could also help to note that all ancestors of tasks in {\small\sf P} are in {\small\sf P}  and all descendants 
of a task in {\small\sf C} are in {\small\sf C}. 
}

\section{Additional pseudocode}
Some of the deferred pseudocode is in Fig.~\ref{fig:ps_sched_app}.
\begin{algorithm}[t!]
{\scriptsize
\chline\\
 	 \textbf{Func:} \textsf{PlTasksBack}: \codecomment{// only backward order, i/o same as ${\tt PlaceTasks}$}\\
 	 \textbf{Input}: $\mathcal{V}$, $\mathcal{S}$, $\mathcal{G}$; \textbf{Output}: $\mathcal{S}'$\\
 	 $\mathcal{S} \leftarrow \mbox{\tt Clone}(\mathcal{S}_{\mbox{in}})$\\
 	 finished placement set $\mathcal{F} \leftarrow \{ v \in \mathcal{G} | v \mbox{ already placed in } \mathcal{S}\}$\\
 	 ready set $\mathcal{R} \leftarrow \{ v \in \mathcal{V} - \mathcal{F} | \mathcal{C}(v, \mathcal{G}) \mbox{ already placed in } \mathcal{S}\}$\\
 	 \While{$R \neq \varnothing$}
 	 {
 	 $v' \leftarrow \mbox{task in } \mathcal{R} \mbox{ with longest runtime}$\\
 	 $t \leftarrow \min_{v \in \mathcal{C}(v, \mathcal{G})} {\tt BeginTime}(v, \mathcal{S})$\\
 	 \codecomment{// place $v'$ at the latest time $\leq t - v'.\mbox{duration}$ when its resource demands can be met}\\
 	 $\mathcal{F} \leftarrow \mathcal{F} \cup v'$\\ 
 	 $\mathcal{R} \leftarrow \mathcal{R} \cup \{v \in \mathcal{P}(v', \mathcal{G}) -\mathcal{F}  | \mathcal{C}(v, \mathcal{G}) \mbox{ placed in } \mathcal{S}\}$\\
 	 }
 	 \chline\\
 }
	\vspace{-1.5mm}\ncaption{Deferred parts of Fig.~\ref{fig:ps_sched_helpers}.\label{fig:ps_sched_app}\vspace{-.2in}}
\end{algorithm}

\section{Lower bound: example and proof}
\label{sec:app_lb}
\begin{figure}[t!]
\centering
\includegraphics[width=1.5in]{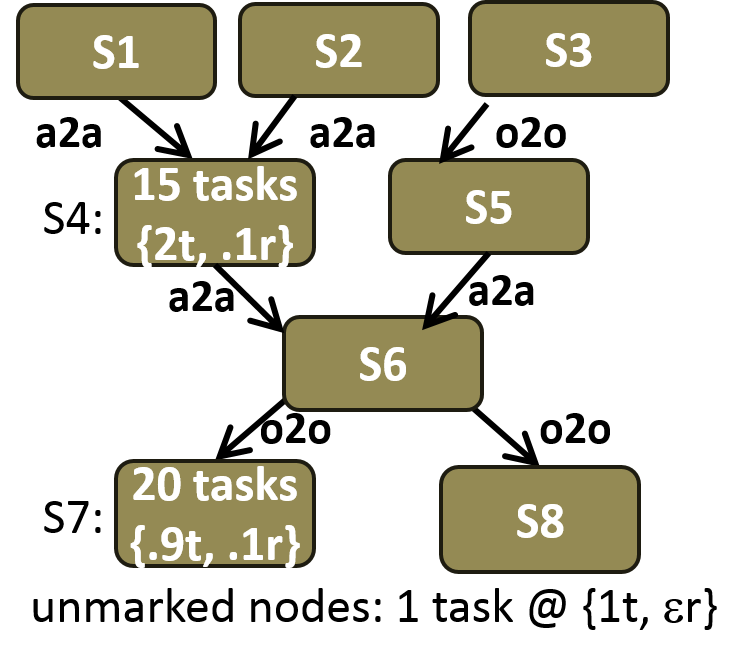}\qquad
\includegraphics[width=1.12in]{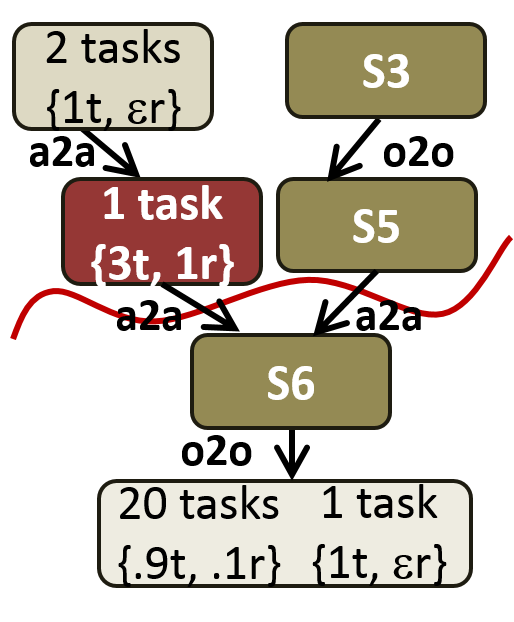}
\ncaption{Original DAG on left. Unmarked nodes have $1$ task each with duration $1t$ and demand $\varepsilon r$. 
The modified DAG used by our lower bound is on the right. Lower bound
improves by $40$\% from $5t$ to $6.8t$. Edge labels ${\tt o2o}, {\tt a2a}$ denote one-to-one and all-to-all 
dependencies between the tasks in the corresponding stages.\label{fig:lb_example}}
\end{figure}
Fig.~\ref{fig:lb_example} shows an example DAG and the various lower bounds. 
Note that, ${\tt CPLen}$ is $5t$ and ${\tt TWork}$ is $4.8t$.  
The longest critical path is $\{S_1, S_2\} \rightarrow S_3 \rightarrow S_6 \rightarrow S_8.$
The only stages with non-trivial work are $S_4$ at $3t$ and $S_7$ at $1.8t$. We assume the cluster capacity $C$ is $1r$.

The DAG on the right shows the modifications used by {\name}.
We can cut at $S_6$ per logic in~\xref{sec:simplify}.
We can replace $S_4$ with its $s_{\mbox{modDur}}$ of $3t$~(see Eqn.~\ref{eqn:stage_moddur}).
We can group $S_7$ and $S_8$ and replace them with their $s_{\mbox{modDur}}$ of $1.8t$.
We can also merge $S_1$ and $S_2$ but their $s_{\mbox{modDur}}$ remains $1t$ because of their 
small resource requirement.
The overall lower bound is now $6.8t$; $4t$ from ${\tt CPLen}$ of the DAG part
above the cut and $2.8t$ from the ${\tt ModCP}$ of the part below the cut.

\vskip .1in
\noindent
{\bf Proof of lemma 3.}
First, observe that the maximum of the lower bounds is also a lower bound.
Second, observe that by definition the DAG is cut into parts that have no overlap.
Hence, the lower bound of the DAG is equal to the sum of the lower bounds of the parts.
This supports Eqn.~\ref{eqn:newlb}.

Next, grouping stages with identical parent and child stages is appropriate because (a) there are no dependencies
between these tasks and (b) the definition of a stage has been a group of independent tasks that can run in parallel
and adhere to a specific dependence pattern with tasks in parent and child stages.

Using the total work to be done in a stage instead of the duration of a single task is appropriate.
This holds because either all of the work has to be done in line (when all parents and children have all-to-all edges)
or at least one stage on a path through the DAG will have to finish all of its work. 
This supports Eqns.~\ref{eqn:stage_moddur},~\ref{eqn:dag_modcp}.

Stepping back, this new lower bound was possible because of the abundance of groups of independent tasks that is common in data-parallel DAGs.
As we did not relax either dependence satisfaction or resource capacity limits, this lower bound
is much tighter than other bounds based on linear programs that relax those aspects~\cite{ec15,spaa13,icalp14}.

\eat{
\section{System details}
\label{sec:app_system}

We present some more details of our implementation including how we
estimate tasks resource demands, our effort of integrating {\name}
into open source, and how we dealt with some of the challenges
during this process.

\subsection{Estimating  Resource Demands of Tasks}
\label{ssec:res_demands}
To estimate task demands and durations, we use two methods. First,
recurring jobs are fairly common in the cluster (up to 40\% in some of
our clusters~\cite{rope}); these are jobs that execute periodically on
newly arriving data (e.g., updating metrics for a dashboard). For such
jobs, {\name} extracts statistics from prior runs. Second, {\name} monitors
the progress and resource usage of tasks at runtime. Note that tasks belonging to the 
same stage do similar computation albeit on different inputs.
Furthermore, these tasks often run in multiple waves. Hence, {\name}
uses statistics from the completed tasks as well as the progress
reports of running tasks to estimate demands for the remaining
tasks. We do note however that requiring user annotation is also
possible, and there are some promising efforts to infer task
requirements from program analysis~\cite{percolator}.

\subsection{Open Sourcing Process}
\label{ssec:os_process}
As we mentioned in the introduction, we are in the process of
committing the core {\name} ideas to the Apache YARN and Tez
frameworks. We have filed several JIRAs~\cite{jira} to track and
commit our patches. The main JIRA,
\href{https://issues.apache.org/jira/browse/YARN-2745}{YARN-2745}
provides a detailed design for the implementation and links to the corresponding sub-jiras.

Concretely, our effort has four main parts: (1) Implement new
task matching logic at the RM
(\href{https://issues.apache.org/jira/browse/YARN-2967}{YARN-2967})
that implements Pseudocode~\ref{fig:ps_online}. 
To enable its implementation in open source, we
have also redesigned the RM container allocation logic
(\href{https://issues.apache.org/jira/browse/YARN-4056}{YARN-4056}) to
enable arbitrarily many containers to be allocated in one
iteration over \textit{queues, applications} and \textit{priorities}.
We call this {\em bundling}.
(2) Extend the YARN API to encapsulate additional information in
the AM-to-RM heartbeat such as the remaining work in the job, tasks' 
resource demands and durations etc. The RM scheduler needs this
additional information to drive  the above logic:
(\href{https://issues.apache.org/jira/browse/YARN-2966}{YARN-2966},
\href{https://issues.apache.org/jira/browse/YARN-4093}{YARN-4093}).
(3) Enhance the node managers~(NMs) that run on each machine to monitor and report the actual
resource usage of the running tasks as well as the aggregate usage of the machines to the RM
(\href{https://issues.apache.org/jira/browse/HADOOP-12210}{HADOOP-12210},
\href{https://issues.apache.org/jira/browse/HADOOP-12211}{HADOOP-12211},
\href{https://issues.apache.org/jira/browse/YARN-2965 }{YARN-2965}).
(4) Implement the schedule construction logic~(Pseudocode~\ref{fig:ps_sched}) in Apache Tez.

For most of the JIRAs mentioned above, {\name}'s patches are under
review by committers from Microsoft, Cloudera and Hortonworks.

\begin{algorithm}[t!]
	{\scriptsize
		\chline\\
		\textbf{Definitions:} a \mbox{\textit{bundle}} keeps information about a set of tasks\\
		\chline\\
		\codecomment{// whenever a heartbeat arrives from node \mbox{X}}\\
		\While{canAssign on \mbox{X}}
		{
			\ForEach{\mbox{queue} $\mathcal{Q} \in {\tt avail. queues}$}
			{
				\ForEach{\mbox{app.} $\mathcal{A} \in {\tt \text{runnable apps.} \in \mathcal{Q}}$}  	
				{
					\ForEach{\mbox{prio.} $\mathcal{P} \in {\tt \mathcal{A}\text{'s prios.}}$}  	
					{
						\If{\mbox{task($\mathcal{P}$, $\mathcal{A}$, $\mathcal{Q}$)} fits on \mbox{X}}
						{
							\textbf{\textcolor{red}{
								\If{the assignment is deferred}
								{
									\codecomment{// update the bundled set of tasks }\\
									\mbox{\textit{bundle}}.add(\mbox{task($\mathcal{P}$, $\mathcal{A}$, $\mathcal{Q}$}))\\
									continue
								}
							}}
							\Else
							{
								assign \mbox{task($\mathcal{P}$, $\mathcal{A}$, $\mathcal{Q}$)} to \mbox{X}\\
								\codecomment{// break all \textbf{for} loops }\\
								break
							}
						}
					}
				}
			}
		}
		\textbf{\textcolor{red}{assign every task in the bundle in one shot;}}
	}
	\ncaption{Simplified pseudocode for the assignment logic. In red are the main logical changes due to the bundling logic.\label{fig:os_bundling}\vspace{-.2in}}
\end{algorithm}

\subsection{Implementation Challenges: Bundling}
\label{ssec:os_challenge}
In the YARN Resource Manager~(RM),  {\name} primarily changes
the scheduler. Of which there are two implementations: Capacity and Fair.
We target Capacity Scheduler which is much more widely used than FairSched. 

Architecturally, Capacity Scheduler offers hierarchical fair queueing. Each queue
has a share relative to its siblings. Applications~(jobs) arrive in each queue. Many applications
can share a queue. And, within an application, tasks can have many priorities. Recall, we enforce
our desired schedule order using priorities.

Capacity Scheduler iterates over \textit{queues}, \textit{applications} and then \textit{priorities} to
find the first task that fits within a machine. Once discovered, the loop breaks and restarts 
because the queues and applications have to be resorted. This design has two complications: (1)
it is slow, only one task is picked per iteration over the three {\em loops}  and (2) implementing new logic
is cumbersome since everything has to be implemented as a sort order over queues or applications or priorities.

We propose the notion of a \textit{bundle}.
In each loop iteration, all tasks that could be assigned to a machine are added to the bundle.
The choice of which tasks are best suited is made at the end of the loop by examining 
all tasks in the bundle.

This speeds up the scheduler since multiple tasks can be assigned in a single loop iteration.
It also allows for non-greedy choices over tasks: such as picking $t_2$ and $t_3$ together 
instead of $t_1$ in spite of $t_1$ being discovered earlier in the loop. Picking $t_1$ greedily,
as the current scheduler does, may preclude picking any other task because not enough
resources are available.
We have implemented a variety of bundling policies including one that mimics the 
current scheduler and one that implements Pseudocode~\ref{fig:ps_online}.
Fig.~\ref{fig:os_bundling} describes how the bundling logic fits onto
the existing scheduling logic. 

A crucial challenge with bundling is the need for backward compatibility.
Bundlers should also account for several aspects. 
(1) \textbf{Reservations} are currently implemented by remembering 
the number of missed scheduling opportunities for task. When starvation is detected, 
a reservation is made. To not affect reservations, a bundler will (a) allocate reserved containers as before
and (b) create reservations/ increment missed-scheduling-opportunities only for the first legitimate
task discovered per loop.
(2) \textbf{AM assignments} refer to the launching of the job managers. 
To not delay the assignment of AM containers,
if an AM container request is seen before any legitimate task, the bundler
will terminate the iteration early and assign the AM container to the machine right away.
(3) \textbf{Delay Scheduling}: similar to the case of reservations, a bundler will update the 
scheduling opportunity counts only for tasks seen in each iteration before any other legitimate task. 
(4) \textbf{Statistics and other Resource checks}: by carefully accounting for delayed action, 
the bundler maintains correctness of all existing checks 
on queue and user resource limits and node labels.


}

\eat{
\begin{figure}[t!]
	\centering
	\begin{subfigure}[b]{\textwidth}
		\centering
		\includegraphics[width=2in]{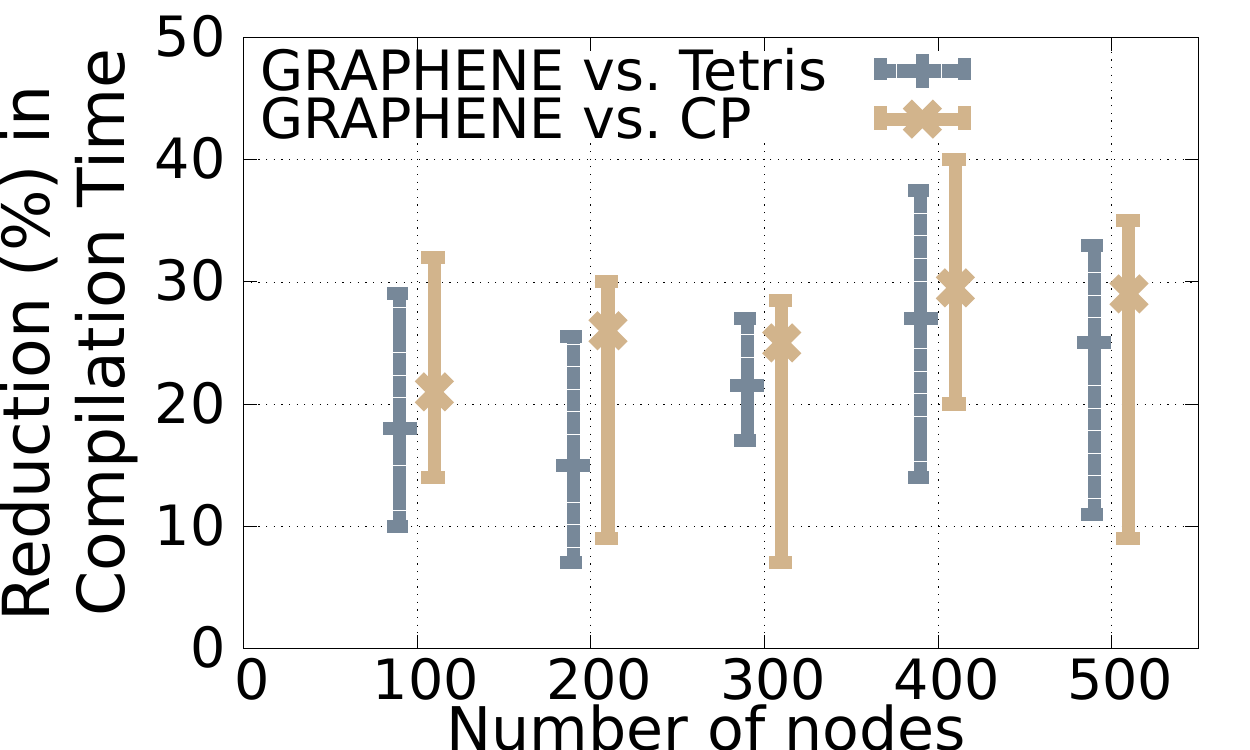}
		\ncaption{Compilation Systems\label{fig:cloud_build}}
	\end{subfigure}
    \par\bigskip
	\begin{subfigure}[b]{\textwidth}
		\centering
		\includegraphics[width=2in]{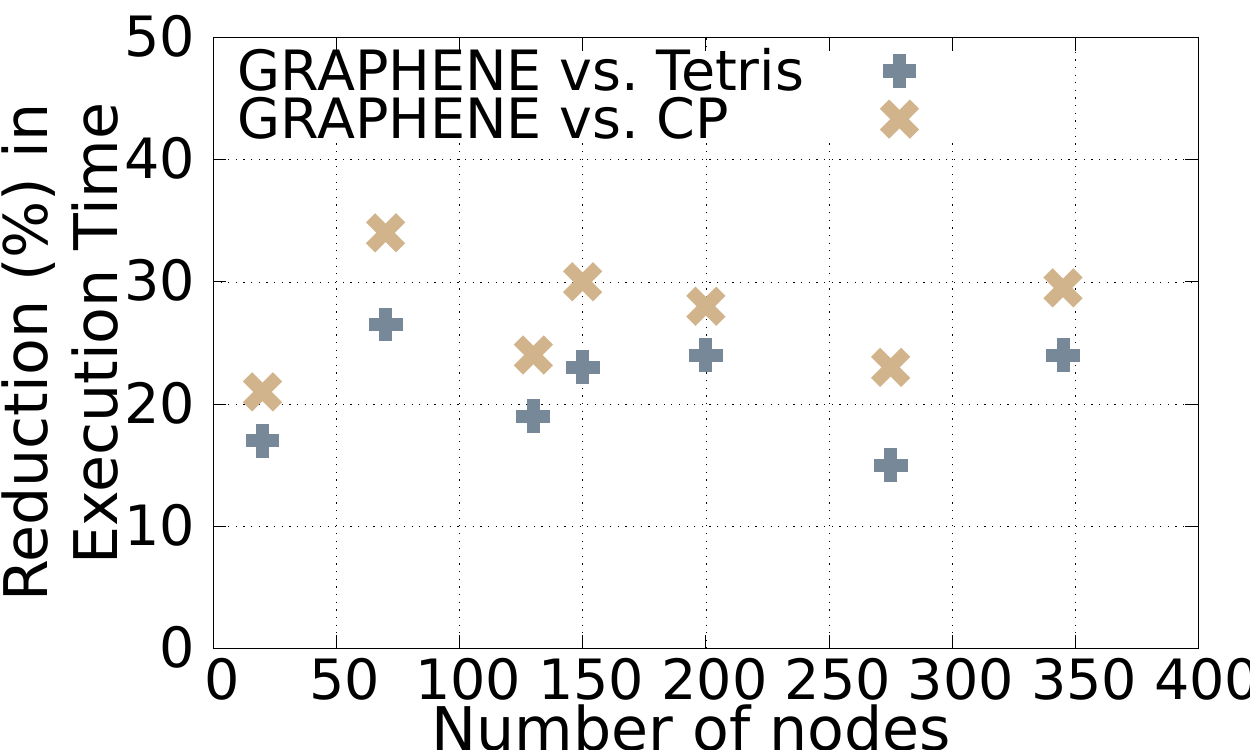}
		\ncaption{Execution Systems\label{fig:kwiken}}
	\end{subfigure}
	\vspace{2.5mm}\ncaption{{\name} benefits to other domains.} \label{fig:eval_compare_other_domains}
\end{figure}
\section{Applying {\name} to other domains}
\label{sec:otherdom}
Distributed compilation systems are being built to speed up compilation of large code bases~\cite{bazel,cloudmake}. 
Each compilation job recompiles the files that have changed and links with the pre-compiled libraries for the 
other files. The various tasks~(compilation, linking, test, code analysis) can take different amounts of time and 
require different resources. We evaluate {\name} on the DAGs that arise in compilations of two large code bases. 
Fig.~\ref{fig:cloud_build} shows that {\name} is $20$\% faster than ${\tt Tetris}$ and $30$\% faster than 
${\tt CPSched}$ at median. The error bars indicate quartiles and we see that the gains hold across job sizes. 

A second use case is that of datacenter-side request-response
workflows for Internet-services~\cite{kwiken}. A search query made to
Bing translates into a workflow of RPC calls at the datacenter with
specific dependencies between calls~(e.g., spell check before index
lookup, video and image lookup in parallel).  Different calls use
different resources and can take very different times. Often, most
call-tasks execute on the same server pool~\cite{kwiken}. Upon
analyzing several workflows from Bing, we found that {\name} improves
upon alternatives by about $24$\% (Fig.~\ref{fig:kwiken}).
}

\begin{figure}[!t]
	\includegraphics[width=2.8in]{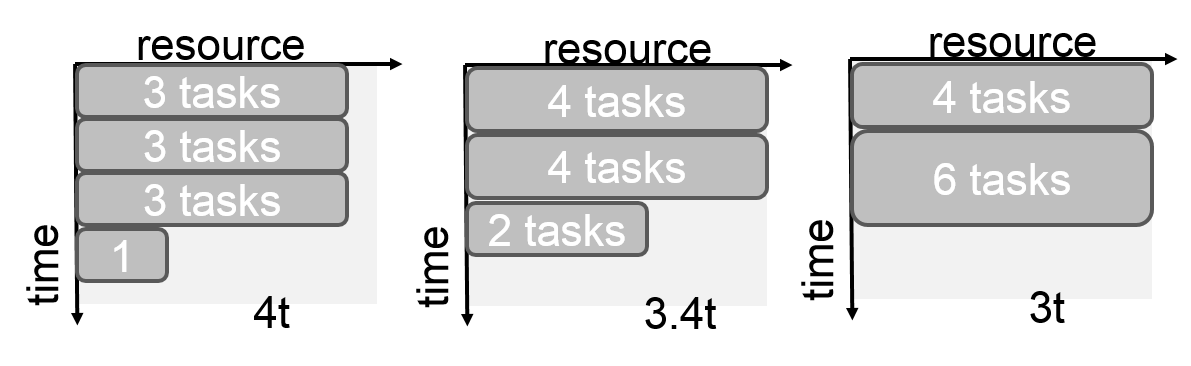}
	\vspace{-1.5mm}\ncaption{Wavekilling: $10$ tasks with duration $1t$ and demands $.3r$. Available capacity is $1$r.  Ensuring fit leads to the schedule on the left. Overbooking helps. But overbooking the last {\em wave} slightly more can reduce runtime substantially~($25$\%).\label{fig:wavekilling}}
\end{figure}
\begin{figure}[t!]
	\includegraphics[height=.581in]{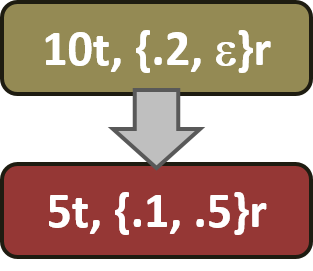}
	\includegraphics[height=.681in]{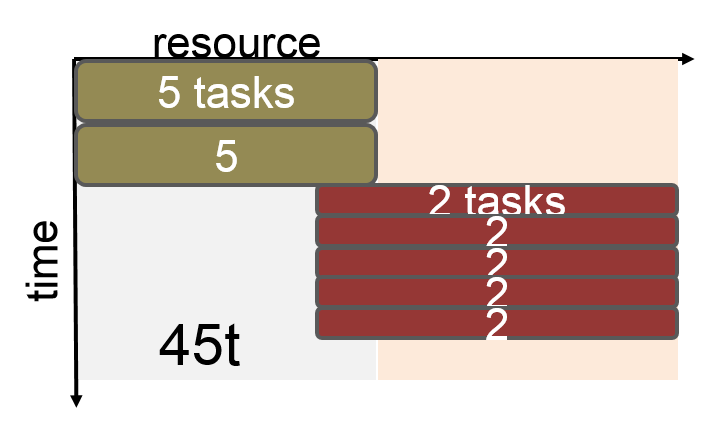}
	\includegraphics[height=.681in]{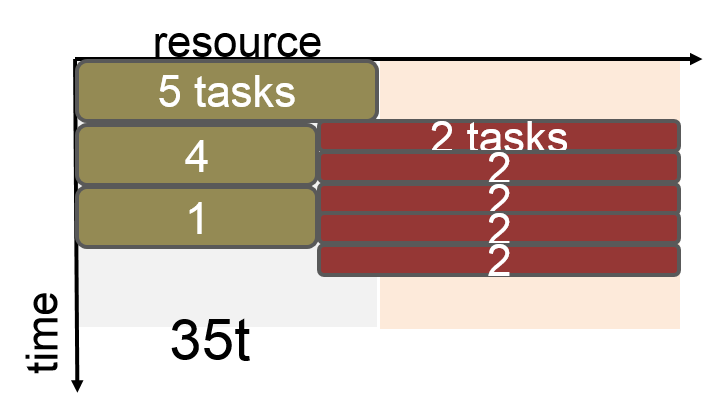}
	\vspace{-1.5mm}\caption{Pipelining: when stages have different dominant resources, more opportunities arise. As we saw above in Fig.~\ref{fig:wavekilling}, if the DAG has just the first stage~(S1), then its optimal completion time is $30t$. Scheduling the second stage~(S2) subsequently takes five waves each $5t$ long leading to the middle schedule of length $55t$. However, observe that upon completion of the first wave of S1, four tasks of S2 are runnable due to the one-to-one dependency.  These can be packed two-at-a-time beside the next wave of S1. Note here that since $2$ of S2's tasks already use up all of the available resource, there is no gain from wavekilling S1. Since S2's tasks take $5t$ each, maintaining at least a width of $4$ for S1, is ideal since all of S2's tasks that become runnable due to the previous wave can overlap with the subsequent wave. This leads to the schedule on the right; gains of $\sim 30\%$. The more the number of tasks in S1 and S2, the larger the gains due to such pipelining.\label{fig:pipelining}}
\end{figure}
\section{More examples}
This section provides two 
further examples of techniques used by {\name}. 

Wavekilling is an idea that allows for more aggressive overbooking
for tasks that are at the end of a wave. Consider the example in Figure~\ref{fig:wavekilling}.
No overbooking takes $4$t. The overbooking logic described in~\xref{sec:overbooking}
takes $3.4$t. But more aggressively overbooking the last wave concludes in $3$t.

Pipelining is an idea suited for the case when a parent-child pair of stages
has resource demands such that they can be better packed together.
Fig.~\ref{fig:pipelining} shows an example with a detailed account of why pipelining
helps.
\eat{
Further examples to motivate {\name}'s benefits are presented in Figs.~\ref{fig:exampleset1},~\ref{fig:exampleset2}.

\begin{figure}[t]
\centering
\begin{subfigure}[b]{\textwidth}
\includegraphics[height=.8in]{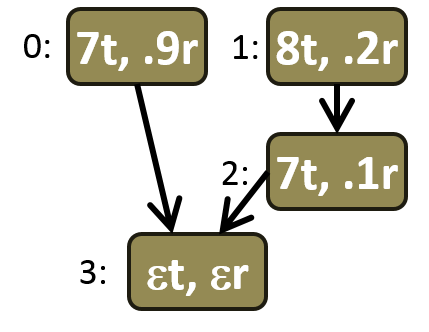}
\includegraphics[height=.9in]{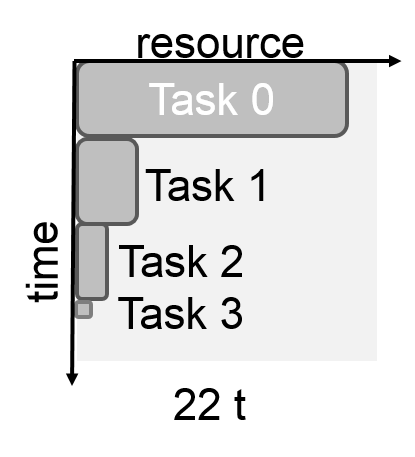}
\includegraphics[height=.9in]{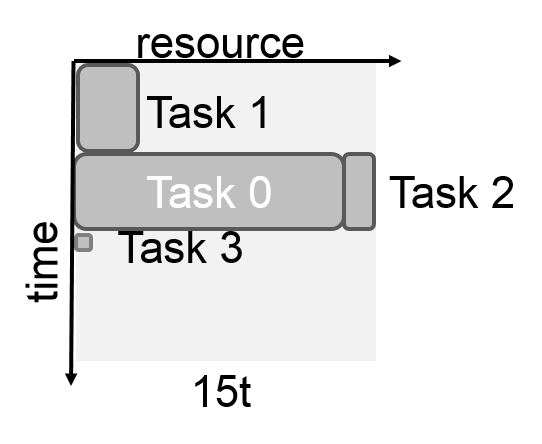}
\ncaption{DAG on the left. Label illustrates duration (suffix t) and resource demand (suffix r). 
Assume total resource available is $1r$ and $\varepsilon \rightarrow 0$. Middle and right are schedules drawn to scale. Breadth-first scheduling leads to the middle schedule, Critical Path to the right schedule; Random has an even chance of either. Here CriticalPath = Opt whereas others are $~50\%$ off. This example illustrates the need to pack.\label{fig:p_eg1}}
\end{subfigure}
\vskip .2in
\begin{subfigure}[b]{\textwidth}
\includegraphics[height=.8in]{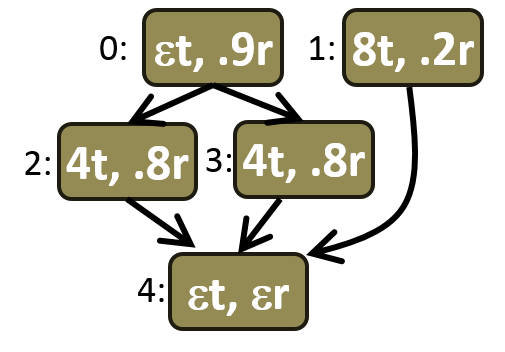}
\includegraphics[height=.8in]{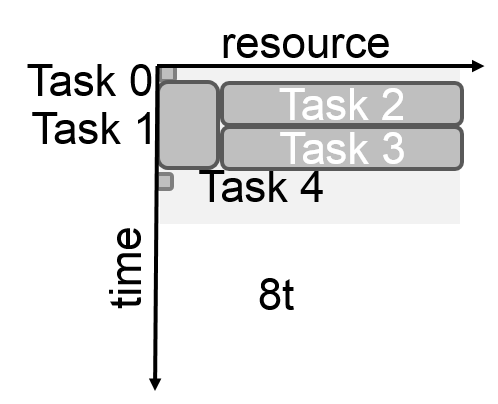}
\includegraphics[height=.8in]{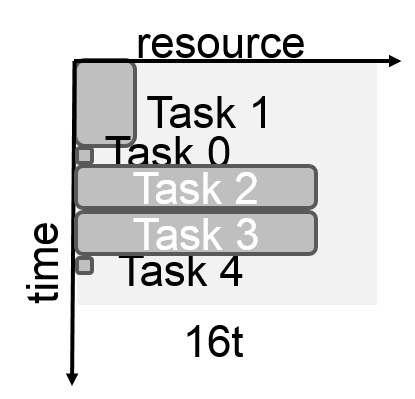}
\ncaption{As above, breadth-first scheduling leads to the middle schedule, Critical Path to the right and Random has an even chance of either. Critical path $=2\times$ Opt. But BreadthFirst = Opt. This example illustrates that different algorithms are better for different DAGs.\label{fig:p_eg2}}
\end{subfigure}
\ncaption{Examples to illustrate the challenges in scheduling DAGs.\label{fig:exampleset1}}
\end{figure}

\begin{figure}[t!]
\centering

\begin{subfigure}[b]{\textwidth}
\includegraphics[height=.7in]{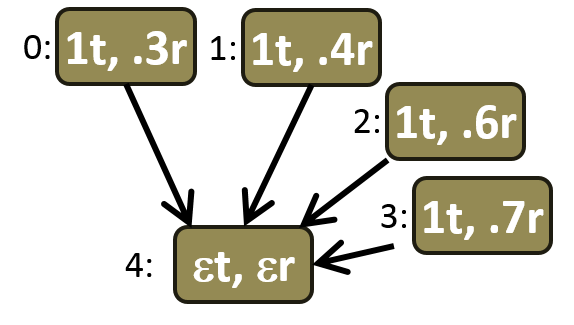}
\includegraphics[height=.7in]{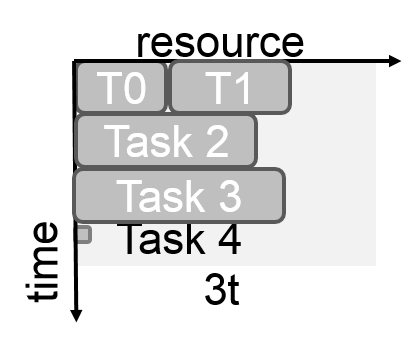}
\includegraphics[height=.7in]{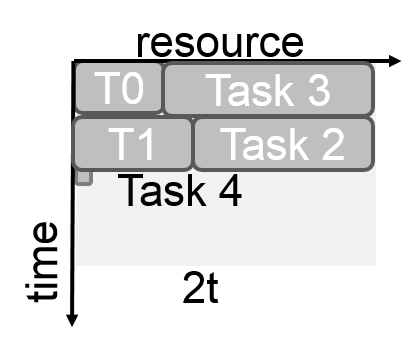}
\ncaption{Fragmentation: Schedulers that do not pack such as CPSched, Breadth-first, etc. can yield the schedule in the middle. Packers such as Tetris would yield the right schedule which is 33\% better.\label{fig:p_eg6}}
\end{subfigure}
\vskip .2in
\begin{subfigure}[b]{\textwidth}
\includegraphics[height=.7in]{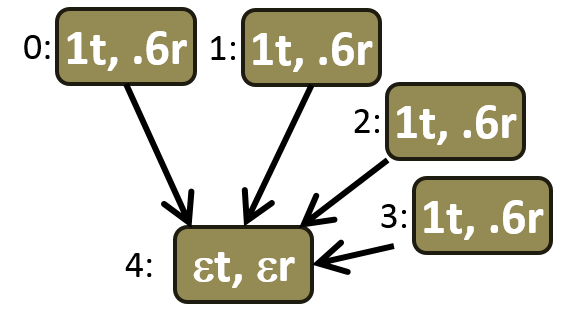}
\includegraphics[height=.7in]{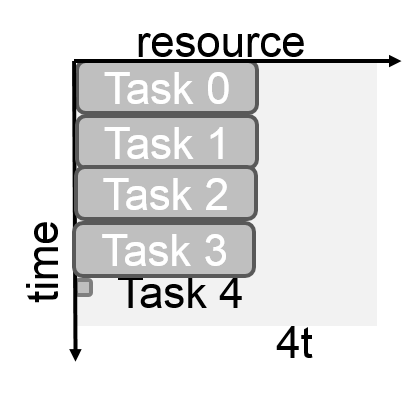}
\includegraphics[height=.7in]{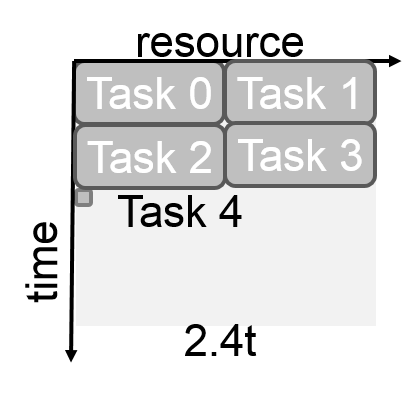}
\ncaption{For fungible resources, overbooking can help. Tasks'
get only $.5r$ each and their runtime increases to $1.2t$. But, by eliminating fragmentation, the job finishes in $2.4t$ instead of $4t$~($40$\%).\label{fig:p_eg4}}
\end{subfigure}
\vskip .2in
\begin{subfigure}[b]{\textwidth}
\includegraphics[height=.7in]{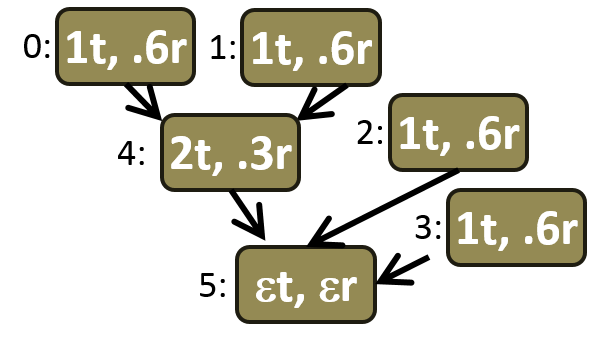}
\includegraphics[height=.7in]{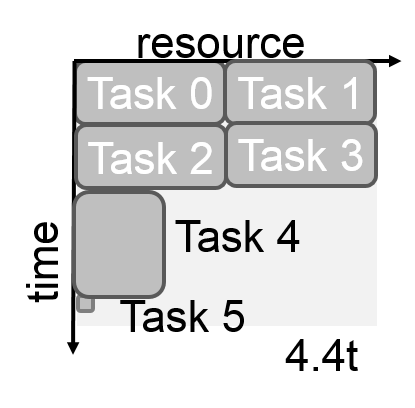}
\includegraphics[height=.7in]{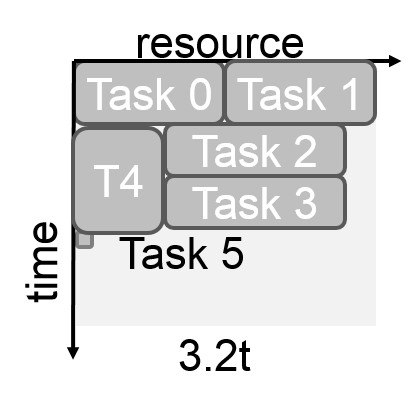}
\ncaption{Overbooking has to be done judiciously.
The middle schedule which naively overbooks has a chokepoint; no task to accompany T4. $4.4t \rightarrow 3.2t$~($\sim$25\%).\label{fig:p_eg5}}
\end{subfigure}
\caption{Overbooking helps but has to be done judiciously.\label{fig:exampleset2}}
\end{figure}
}
\eat{
\section{More results}
\begin{figure}[!]
\begin{subfigure}[b]{1\textwidth}
	\centering
	\begin{scriptsize}
		\begin{tabular}[b]{l|rrr}
			& \multicolumn{3}{c}{$75^{th}$ \%} \\
			{Workload} & {T+CP} & {T+T} & {\bf G} \\
			\hline
			{TPC-DS} & 8.9 & 16.6 & {\bf 45.7} \\
			\hline
			{TPC-H} & 7.7 & 15.0 & 48.3 \\
			\hline
			{BigBench} & 21.7 & 18.5 & 33.3 \\
			\hline
			{Skype} & 4.5 & 14.2 & 29.7 \\\hline
		\end{tabular}
	\end{scriptsize}
\end{subfigure}
\vspace{-3mm}\ncaption{Improvements in job completion time. This result is part of Fig.~\ref{fig:eval_compl_times}\label{fig:eval_compl_times_appendix}}
\end{figure}
}

\eat{
	\begin{subfigure}[b]{0.60\textwidth}
		\begin{scriptsize}
			\begin{tabular}[b]{l|rrr|rrr|rrr}
				& \multicolumn{3}{c|}{$25^{th}$ \%} & \multicolumn{3}{c|}{$50^{th}$ \%} & \multicolumn{3}{c}{$90^{th}$ \%}\\
				{Workload} & {T+CP} & {T+T} & {\bf G} & {T+CP} & {T+T} & {\bf G} & {T+CP} & {T+T} & {\bf G} \\
				\hline
				{TPC-DS} & 2.0 & 1.9 & {\bf 16.0} & 4.1 & 6.5 & {\bf 27.8} & 12.0 & 52.0 & {\bf 70.8} \\
				\hline
				{TPC-H} & 1.8 & 1.5 & 7.6 & 3.8 & 8.9 & 30.5 & 13.6 & 42.1 & 63.7  \\
				\hline
				{BigBench} & 4.1 & 2.0 & 5.6 & 6.4 & 6.2 & 25.0 & 47.2 & 42.8 & 67.0 \\
				\hline
				{MS-Prod} & -3.0 & 3.2 & 4.4 & 1.0 & 5.8 & 19.0 & 25.2 & 43.6 & 58.1\\\hline
				\multicolumn{10}{p{3.6in}}{\rule{0pt}{2ex}{\bf G} is {\name}, {\bf T+T} is Tez + Tetris and {\bf T+CP} is Tez + CP. The improvements are relative to ${\tt Tez}$. Each group of columns reads out the gaps at the percentile in the label of that group. $75^{th}$\% result is presented in Fig.~\ref{fig:eval_compl_times_appendix}.}\\
			\end{tabular}
		\end{scriptsize}
		\ncaption{Improvements in job completion time across all the workloads \label{tab:eval_lat_datasets}}
	\end{subfigure}

\section{Analytical model for a DAG of tasks}
\fixme{this section is TBE}
To connect the packing+scheduling problem with prior theoretical work, we extend the formalization in prior work~\cite{tetris} to also consider dependencies between tasks.

\vskip .05in
\noindent{\bf Notations:} We consider demands of tasks along four resources: CPU, memory, disk, and network bandwidth~(see Table~\ref{tab:task_resource_demands}). For every resource $r$, we denote ($i$) the capacity of that resource on machine $i$ as $c_i^r$, and ($ii$) the demand of task $j$ on resource $r$ as $d_j^r$.  To encode the dependencies between tasks, let $p_j$ denote the set of {\em parent}s of task $j$.

\begin{table}[t!]
\begin{small}
\begin{tabular}{p{.6in}|p{2.4in}}
{\bf Term} & {\bf Explanation}\\\hline
$d_j^{\text{CPU}}, f_j^{\text{CPU}}$ & CPU demand (i.e., cores)~($d$) and CPU cycles~($f$) \\[.05in]\hline
$d_j^{\text{mem}}$ & Peak memory used by task \\[.05in]\hline
$d_j^{\text{diskR}}, d_j^{\text{diskW}}$, $f_{ij}^{\text{diskR}}, f_{j}^{\text{diskW}}$ & Peak disk read/write bandwidth of the task~($d$), bytes to be read from machine $i$ and output~($f$)  \\[.05in]\hline
$d_j^{\text{netIn}}, d_j^{\text{netOut}}$ & Peak network bandwidth in/out of a machine that the task can use~($d$)\\ \hline\hline
\end{tabular}
\end{small}
\caption{Resource demands of task $j$. Note that the demands for network resource depend on task placement.
}
\label{tab:task_resource_demands}
\end{table}

Next, we define variables that encode the task schedule and resource allocation. Note that the allocation happens across machines (spatial) and over time (temporal). Let $Y_{ij}^t$ be an indicator variable that is $1$ if task $j$ is allocated to machine $i$ at time $t$ (time is discretized). Also, task $j$ is allocated $X_{ij}^{r,t}$ units of resource type $r$ on machine $i$, at time $t$. \camera{Let $i^*_j$ denote the machine that task $j$ is scheduled at, and $\text{start}_j$ and $\text{duration}_j$ denote the corresponding aspects of task $j$.\footnote{While we use $\text{start}_j$ and $i_j^*$ in the constraints below for convenience, note that they are unknown apriori. Hence, all their occurrences need to be replaced with equivalent terms using the indicator variables $Y_{ij}^t$. For example $X_{i_j^*j}^{r,t} = \sum_i Y_{ij}^t X_{ij}^{r,t}$ and $\mbox{start}_j = \min_{i,t}\{Y_{ij}^t > 0\}.$}} We do not model preemption here for simplicity. Note that a task may need network and other resources at multiple machines in the cluster due to remote reads.

\vskip .05in
\noindent{\bf Constraints:} 
\camera{First}, the cumulative resource usage on a machine $i$ at any given time cannot exceed its capacity:
\begin{equation}
\sum_j X_{ij}^{r,t} \leq c_i^r, \forall{i, t, r}.
\label{eqn:alloc_cap}
\end{equation}
Second, tasks need no more than their peak requirements and no resources are allocated when tasks are inactive.
\begin{eqnarray}
0 \leq X_{ij}^{r,t} \leq d_j^r, \forall{r, i, j,t}.\label{eqn:alloc_peak}\\
\forall_{i, r, t}\  X_{ij}^{r,t} = 0 \mbox{ if } t \notin [\text{start}_j, \text{start}_j+\text{duration}_j].
\end{eqnarray}
Third,  to avoid costs of preemption, schedulers may prefer to uninterruptedly allot resources until the task completes.
\begin{equation}
\sum_{t=\text{start}_j}^{\text{start}_j  + \text{duration}_j} Y_{ij}^t = 
\left\{
\begin{array}{ll}
\text{duration}_j & \mbox{machine $i = i_j^*$}\\
0 & \mbox{other machines}
\end{array}\right.
\label{eqn:indicator_runtime}
\end{equation}
Fourth, the duration of a task depends on task placement and resource allocation:
\begin{equation}
\text{duration}_j = \max \left(
\begin{array}{l}
\frac{f_j^{\text{CPU}}}{\sum_{t} X_{i_j^*j}^{\text{CPU},t}}, 
\frac{f_j^{\text{diskW}}}{\sum_{t} X_{i_j^*j}^{\text{diskW},t}}, \\
\forall i \ \ \frac{f_{ij}^{diskR}}{\sum_{t} X_{ij}^{\text{diskR},t}}, \\
\frac{\sum_{i\neq i_j^*}f_{ij}^{diskR}}{\sum_{t} X_{i_j^*j}^{\text{netIn},t}},\\
\forall i \neq i_j^*, \ \  \frac{f_{ij}^{diskR}}{\sum_{t} X_{ij}^{\text{netOut},t}}\\
\end{array}
\right)
\label{eqn:runtime}
\end{equation}
In each term, the numerator(f) is the total resource requirement (e.g., CPU cycles or bytes), while the denominator is the allocated resource rate (e.g., cores or bandwidth). 
From top to bottom, the terms correspond to CPU cycles, writing output to local-disk, reading from (multiple) disks containing the input, network bandwidth into the machine that runs the task and network bandwidth out of other machines that have task input.
Note that cores and memory are only allocated at the machine that the task is placed, but the disk and network bandwidths are needed at every machine that contains task input. We assumed here for simplicity that task output just goes to local disk and that the bandwidths are provisioned at each time so as to not bottleneck transfers, i.e., $$\forall{t, j, i \neq i_j^*}\ \ X_{ij}^{\text{netOut}, t} \geq X_{ij}^{\text{diskR},t} \text{ and } X_{i_j^*j}^{\text{netIn}, t} \geq \sum_{i\neq i^*}X_{ij}^{netOut,t}.$$ We also assumed that the task is allocated its peak memory size, since unlike the above resources, tasks' runtime can be arbitrarily worse (due to trashing) if it were allocated less than its peak required, i.e., $  \forall t \ \ X_{i_j^*j}^{\text{mem},t} = d_j^{\text{mem}}$. 

Finally, we have to encode dependencies between tasks. That is, a task cannot begin until all its parent tasks have finished.
\begin{equation}
\mbox{start}_j \geq \max_{k \in p_j} \left(\mbox{start}_k + \mbox{duration}_k\right).
\end{equation}

\begin{table}[t!]
\begin{small}
\begin{tabular}{p{.6in}|p{2.4in}}
{\bf Term} & {\bf Explanation}\\\hline
$c_i^{\text{CPU}}, c_i^{\text{mem}}$ & Number of cores, Memory size\\
$c_i^{\text{diskR}}, c_i^{\text{diskW}}$ & Disk read and write bandwidth\\
$c_i^{\text{netIn}}, c_i^{\text{netOut}}$ & Network bandwidth in/out of machine\\
\end{tabular}
\end{small}
\caption{Resources measured at machine $i$. Note that the availability of these resources changes as tasks are scheduled on the machine.\label{tab:machine_rsrc}}
\end{table}
\vskip .05in
\noindent{\bf Objective function:} 
\camera{
Our setup lends itself to various objective functions of interest to distributed clusters including minimizing {\em makespan}~\footnote{Makespan = $\max_{\mbox{job}\ J} \max_{\mbox{task}\ j \in J} \max_{\mbox{time}\ t} (Y_{ij}^t>0)$} and {\em job completion time}~\footnote{Job $J$'s finish time is $\max_{\mbox{task } j \in \mbox{ job } J} \max_{\mbox{time }t} (Y_{ij}^t >0)$}. Minimizing makespan is equivalent to maximizing packing efficiency. 
Further, we can pose fairness as an additional constraint to be meet at all times.\footnote{Dominant resource share of $J$ at time $t$ $=\max_{\mbox{resource }r} \frac{\sum_{i, j\in J} X_{ij}^{r,t}}{\sum_i c_i^r}$.} 
}

\vskip .05in
\noindent{\bf Takeaways: }
{\bf Why dependencies make this so much harder?}
The above analytical formulation is essentially a hardness result.
Regardless of the objective, several of the constraints are non-linear.
Resource malleability~(eqn.\ref{eqn:alloc_peak}), task placement~(eqn.\ref{eqn:runtime}) and how task duration relates to the resources allocated at multiple machines~(eqn.\ref{eqn:runtime}) are some of the complicating factors. Fast solvers are only known for a few special cases with non-linear constraints~(e.g., the quadratic assignment problem).
Finding the optimal allocation, is therefore computationally expensive.
Note that these non-linearities remain {\em inspite} of the several simplifications used throughout (e.g.,  no task preemption). 
In fact, scheduling theory shows that even when placement considerations are eliminated, the problem of packing multi-dimensional balls to minimal number of bins is APX-Hard~\cite{apx}. 
Worse, re-solving the problem whenever new jobs arrive requires online scheduling.
For the online case, recent work also shows that no polynomial time approximation scheme is possible with competitive ratio of the number of dimensions unless NP=ZPP~\cite{onlineVectorPacking}.
}

\section{Judicious overbooking}
\label{sec:overbooking}

The key tussle in overbooking is as follows.
Overbooking improves throughput. For example, suppose tasks with duration $t$ require $0.6r$.
Running two tasks instead of one improves
throughput by $60\%$ from $1$ task/$t$ to $2$ tasks/$1.2t$.
Here, we assume task runtime increases linearly with overbooking amount.
Note that the potential gains from overbooking depend on the amount of idle resources that are
reclaimed by overbooking~($.4r$ in the above case).
Conversely, overbooking is counter productive if resources will become free at some other machine soon.
Suppose another machine can run the second task at time $+\varepsilon$. Without overbooking,
tasks will finish at $\{t, t+\varepsilon\}$ and with overbooking both finish at $1.2t$. Even in simple settings, optimal overbooking is NP-hard~\cite{icalp14}.

{\name} offers a heuristic for overbooking.
First, it uses micro-benchmarks to determine how  task runtimes will be delayed when each resource is overbooked.
These functions are concave and vary across resource types.
Next, per potential task to overbook, {\name} runs a what-if analysis to decide between overbooking and
waiting.
We compute the expected completion times of all affected tasks after overbooking.
Note that resource overbooking delays these tasks
but the extent of delay can vary.
The ${\tt benefit}$ of overbooking is how much earlier would the new task finish with overbooking versus having to wait for next-free-resource.
The ${\tt cost}$ is the increase in runtime of all the other tasks due to overbooking.
Thus, overbooking score equals ${\tt benefit} - {\tt cost}$.

\else
\fi
\end{appendices}
\bibliographystyle{acm}
\bibliography{bibs/paper}

\end{sloppypar}
\end{document}